\RequirePackage{fix-cm}
\documentclass[smallextended]{svjour3}       
\smartqed  
\usepackage{amssymb,latexsym,amsfonts,amsmath,mathrsfs}
\usepackage{graphicx}
\usepackage{geometry}
\usepackage{dsfont}
\usepackage{setspace}
\usepackage{tabularx}
\usepackage{placeins}
\usepackage{hyperref}
\usepackage{subfigure}
\usepackage{lscape}
\usepackage{rotating}
 \usepackage{bm}
\usepackage[rightcaption]{sidecap}
\hypersetup{colorlinks=true, linkcolor=blue, citecolor=red}

\usepackage{appendix}
\usepackage{ctable}

\allowdisplaybreaks

\newcommand{\xqedhere}[2]{%
  \rlap{\hbox to#1{\hfil\llap{\ensuremath{#2}}}}}

\begin{document}

\title{Virus-immune dynamics determined by prey-predator interaction network and epistasis in viral fitness landscape}


\author{Cameron J. Browne      \and Fadoua Yahia
}


\institute{C.J. Browne \at Mathematics Department, University of Louisiana at Lafayette, Lafayette, LA \\
            \email{cambrowne@louisiana.edu}    \and   F. Yahia \at Mathematics Department, University of Louisiana at Lafayette, Lafayette, LA    
  }

\maketitle
\begin{abstract}

Population dynamics and evolutionary genetics underly the structure of ecosystems, changing on the same timescale for interacting species with rapid turnover, such as virus (e.g. HIV) and immune response.   Thus, an important problem in mathematical modeling is to connect ecology, evolution and genetics, which often have been treated separately.  Here, extending analysis of multiple virus and immune response populations in a resource - prey (consumer) - predator model from Browne and Smith \cite{browne2018dynamics}, we show that long term dynamics of viral mutants evolving resistance at distinct epitopes (viral proteins targeted by immune responses) are governed by epistasis in the virus fitness landscape. 
In particular, the stability of persistent equilibrium virus-immune (prey-predator) network structures, such as nested and one-to-one, and bifurcations are determined by a collection of \emph{circuits} defined by combinations of viral fitnesses that are minimally additive within a hypercube of binary sequences representing all possible viral epitope sequences ordered according to \emph{immunodominance} hierarchy.   Numerical solutions of our ordinary differential equation system, along with an extended stochastic version including random mutation, demonstrate how pairwise or multiplicative epistatic interactions shape viral evolution against concurrent immune responses and convergence to the multi-variant steady state predicted by theoretical results.  Furthermore, simulations illustrate how periodic infusions of subdominant immune responses can induce a bifurcation in the persistent viral strains, offering superior host outcome over an alternative strategy of immunotherapy with strongest immune response.

\end{abstract}

\section{Introduction}
The evolution of ecological networks depends on the underlying population dynamics, genetics, and structure of the composite species.   Interactions between populations, for example prey-predator or competitive forces, constrain and shape the network, in concert with evolution also diversifying and adapting species variants.  The complexity of these eco-evolutionary dynamics have challenged researchers to classify patterns in rapidly evolving communities.  In the single species context, theoretical models of fitness landscapes have simplified the study of adaptation by reducing individuals to either genotypes or phenotypes, whose reproductive success is determined by a single trait, namely fitness.  Although a multitude of evolutionary pathways exist, evolution predictability can be driven by genetic variant constraints.  A more analytically challenging scenario is the evolution or coevolution of prey-predator systems whereby the predator range and selection of prey resistance balanced by constraints on reproduction together form a dynamic fitness landscape.  Examples include phage-microbe and virus-immune response networks, with the latter, specifically HIV, being a primary motivation for this work.

During HIV infection, a diverse collection of viral strains, often called a quasispecies, compete for a target cell population (mainly CD4$^{+}$ T-cells) while the host immune response population (e.g. CD8$^{+}$ T-cells) predates and proliferates upon pathogen recognition.  HIV can also rapidly evolve resistance to immune response attack at different epitopes (proteins in virus genome displayed on infected cells), inducing a dynamic network of interacting viral and immune variants.  Deciphering patterns in the trajectories of virus and immune response populations, along with their interactions, can advance biological theory and have applications for vaccine or immunotherapy development \cite{Walker,chakraborty2017rational}.  Analogous questions in other biological systems, such as phage-microbe communities, have mostly led to models of species compositions in the face of ecological interactions \emph{independent} of explicit genetic mutations.  The properties of these ecosystem models have classically  been studied using dynamical systems, where concepts such as stability, equilibria and population persistence are used to characterize feasible species assemblages.  Recently, generalized Lotka-Volterra (L-V), chemostat and ecosystem models have been utilized to understand how different motifs, such as nested or one-to-one networks, are built through invading species and convergence to stable equilibria \cite{jover2013mechanisms,korytowski2015nested,browne2016global}.  Additionally, several works have developed polymorphic evolution sequences, where an individual based stochastic model converges to solutions of L-V equations in the limit of small mutation rates and large populations \cite{champagnat2011polymorphic,costa2016stochastic}. However, how population dynamics, genetics and evolution together determine network structures for rapidly evolving ecosystems is not generally established.

From an evolutionary genetics perspective, a high mutation rate allows HIV populations to change and explore sequence space on short timescales,  lending themselves to being studied as model biological systems, along with the significant clinical interest.   Disease progression, escape pathways, and treatment fate depend on viral fitness.    To estimate \emph{in vivo} fitness landscapes, several evolutionary models have linked fitness to viral genotype frequencies, for example the quasispecies model \cite{seifert2015framework} and multi-strain versions of a standard within-host virus model \cite{Vitaly1}.   These models can be mathematically tractable, allowing for analysis of equilibria and stability in terms of mutation rates and fitness quantities.  In particular, the common setting of finite binary sequences, the assumed form of viral genotypes in this current paper, enables geometric or algebraic properties of the binary hypercube space to be exploited for characterizing equilibrium distributions \cite{bratus2019rigorous}.   Inclusion of viral mutation from multiple dynamic immune response populations complicates matters, as neither the virus strain fitness or immune response strength simply determine epitope escape \cite{Vitaly3,leviyang2015broad}.  However, correlation analysis \cite{liu2013vertical} and a statistical physics model of viral sequences with \emph{epistasis} (discussed further below)  \cite{Barton}  applied to HIV patient datasets have found determinants viral evolution based on viral fitness landscapes and immnodominance hierarchies (relative expansion levels of the responding immune populations).  

Epistasis refers to nonlinearity in the fitness landscape or dependence of fitness change from a mutation on the genetic background.  Epistatic interactions play a critical role in fitness landscape features, and ultimately evolutionary trajectories,  thus measuring epistasis has received much attention when studying evolution.   
However, the large amount of interactions within a genome challenge both theoretical and experimental quantification of epistasis.  Several methods for computing epistasis have been proposed \cite{mani2008defining,ferretti2016measuring}.  Here we focus on the concept of circuits introduced by Beerenwinkel et al. \cite{beerenwinkel2007epistasis} as fundamental measures of epistatic interactions and underlying geometry of the fitness landscape.   Circuits have been utilized to characterize single species fitness landscapes in both theoretical and data-driven studies \cite{hallgrimsdottir2008complete,crona2017inferring,gould2018microbiome}.  

In this paper, we investigate how epistasis impacts evolution of prey-predator interacting species, specifically how  virus fitness landscapes affect the overall virus (prey) and host immune response (predator) ecosystem dynamics.
We show that connecting population genetics and dynamics offers a way to extract biological meaningful relationships from the equilibria stability conditions of a complex network differential equation for interacting species' variants.  We build off of our previous analysis of a multi-variant virus-immune model \cite{browne2018dynamics}, which established different regimes of attractors, each with a distinct set of viral strains persisting by extending Lyapunov function methods first applied to generalized L-V equations \cite{Goh,hofbauer1998evolutionary}.  In particular, the stability of certain equilibria structures and associated bifurcations are sharply determined by relevant circuits, which recast strain invasion rates as algebraic combinations of binary sequences shaping viral fitness landscape epistasis.  Furthermore, we simulate eco-evolutionary dynamics showing that our theoretical calculations can carry over to an extended stochastic model with mutations, and also illustrate how distinct immunotherapies can be incorporated in our system to shed light on potential strategies.  We conclude with a discussion on how our study supports the utility of evolutionary genetics concepts, in particular the construction of circuits for measuring epistasis of fitness landscapes, applied to characterizing bifurcations in virus-immune response population dynamics, which represents a specific example of a prey-predator ecosystem model.

\section{General model and binary sequence case}\label{BMcase}
We begin by considering the following rescaled model introduced to describe a network of viral and immune response variants during host infection \cite{browne2018dynamics}:
\begin{align}
\dot x &= 1-x- x\sum_{i=1}^m \mathcal R_i y_i, \notag \\
\dot y_i &= \gamma_i y_i\left(\mathcal R_i x -1 -\sum_{j=1}^n a_{ij}Z_j \right), \quad i=1,\dots,m \label{ode3}  \\
  \dot Z_j &= \frac{\sigma_j}{\rho_j} Z_j\left(\sum_{i=1}^m a_{ij} y_i - \rho_j \right), \quad j=1,\dots,n. \notag
\end{align}
 Here $x$ denotes the population of target cells, along with $m$ competing virus strains ($y_i$ denotes strain $i$ infected cells), and $n$ variants of immune response ($Z_j$).  The parameter $\mathcal R_i$ represents the basic reproduction number of virus strain $i$.  The $m\times n$ nonnegative matrix $A=\left(a_{ij}\right)$ describes the virus-immune \emph{interaction network}, which determines each immune response population's avidity to the distinct viral strains.  Then $\rho_j$ represents the reciprocal of the immune response fitness excluding the (rescaled) avidity to each strain $j$.  Additionally, $\gamma_i$ and $\sigma_j$ represent scaling factors for corresponding viral and immune variant growth rates.      

Each virus strain $i$ (cells infected with strain $i$), $y_i$, has a set of immune responses, $Z_j$, that recognize and attack $y_i$.  We call this set the \emph{epitope set of} $y_i$, denoted by $\Lambda_i$ where $\Lambda_i:=\left\{ j\in [1,n] : a_{ij}>0 \right\}$.
Here $j\in \Lambda_i$ if $y_i$ is \emph{not} completely resistant to immune response $Z_j$.  We remark that the system generally models a tri-trophic ecosystem with a single resource consumed by $m$ prey (or consumer) populations subject to potential attack by $n$ distinct predators (prey $i$ subject to attack by any predator $j$ in $\Lambda_i$).  For example, this model can describe bacteria-phage communities in a chemostat (or single resource environment), where the set $\Lambda_i$ classifies the infection network (whom infects who).   

In this article, we specialize system \eqref{ode3} to the case where each virus strain is represented by a binary sequence of length $n$, exactly coding the loci (epitopes) for which $n$ specific immune responses can recognize and attack.  Note that consideration of binary sequences is perhaps the most common way to represent distinct variants which can differ at some loci of their genome (e.g. quasispecies, haploid models).  A major goal of this work is to connect concepts in evolution and genetics with population dynamics, so this special case is an appropriate setting.  Here the $n$ viral epitopes have two possible alleles: the wild type (0) and the mutated type (1) which has escaped recognition from the cognate immune response.   For each virus strain $y_i$, we associate a binary sequence of length $n$,  $y_i\sim\mathbf i=\left(i_1,i_2,\dots,i_n\right)\in\left\{0,1\right\}^n$, coding the allele type at each epitope.   We assume that each immune response ($Z_j$) targets its specific epitope at the specific rate $a_j$ for virus strains containing the wild-type (allele 0) epitope $j$, whereas $Z_j$ completely loses ability to recognize strains with the mutant (allele 1) epitope $j$, i.e.
\begin{align}
\forall i\in[1,m],  \ y_i\sim\mathbf i \ \Rightarrow \  a_{ij}=(1-i_j)a_j \ \   (\text{or} \ \ a_{ij}=a_j \ \  \text{if} \ \ j\in\Lambda_i,  \ a_{ij}=0 \ \text{otherwise}), \label{assumeaj}
\end{align} 
with $\Lambda_i$ the \emph{epitope set} defined earlier for model \eqref{ode3}  (see Fig. \ref{fig2}).  For example, the wild-type (founder) virus strain, denoted here by $y_w$, is represented by the sequence of all zeroes, denoted $\mathbf 0$, and epitope set $\Lambda_w=\left\{1,\dots,n\right\}$ since it is susceptible to attack by all immune responses. With assumption \eqref{assumeaj}, we can define an \emph{immune reproduction number} corresponding to each $Z_j$: 
 \begin{align}
 \mathcal I_j:=\frac{a_j}{\rho_j} .\label{IR0}
 \end{align} 
Then there are $m=2^n$ possible viral mutant strains, each distinguished by binary sequence $\mathbf i=\left(i_1,i_2,\dots,i_n\right)$ and denoted $y_{\mathbf i}$, governed by the following system:
\begin{align}
\dot x &= 1-x- x\sum_{\mathbf i\in\left\{0,1\right\}^n} \mathcal R_{\mathbf i} y_{\mathbf i}, \notag \\
 \dot y_{\mathbf i} &= \gamma_{\mathbf i} y_{\mathbf i}\left(\mathcal R_{\mathbf i} x -1 -\sum_{j=1}^n (1-i_j) z_j \right), \quad \mathbf i\in\left\{0,1\right\}^n,  \label{odeS} \\
  \dot z_j &= \frac{\sigma_j}{s_j} z_j\left(\sum_{\mathbf i\in\left\{0,1\right\}^n} (1-i_j)y_{\mathbf i} -s_j \right), \quad j=1,\dots,n \notag
\end{align}
where $z_j=a_j Z_j$ and $s_j=1/\mathcal I_j$.


 \begin{figure}[t]
\centering\subfigure[][]{\label{fig2a}\includegraphics[width=12cm,height=5cm]{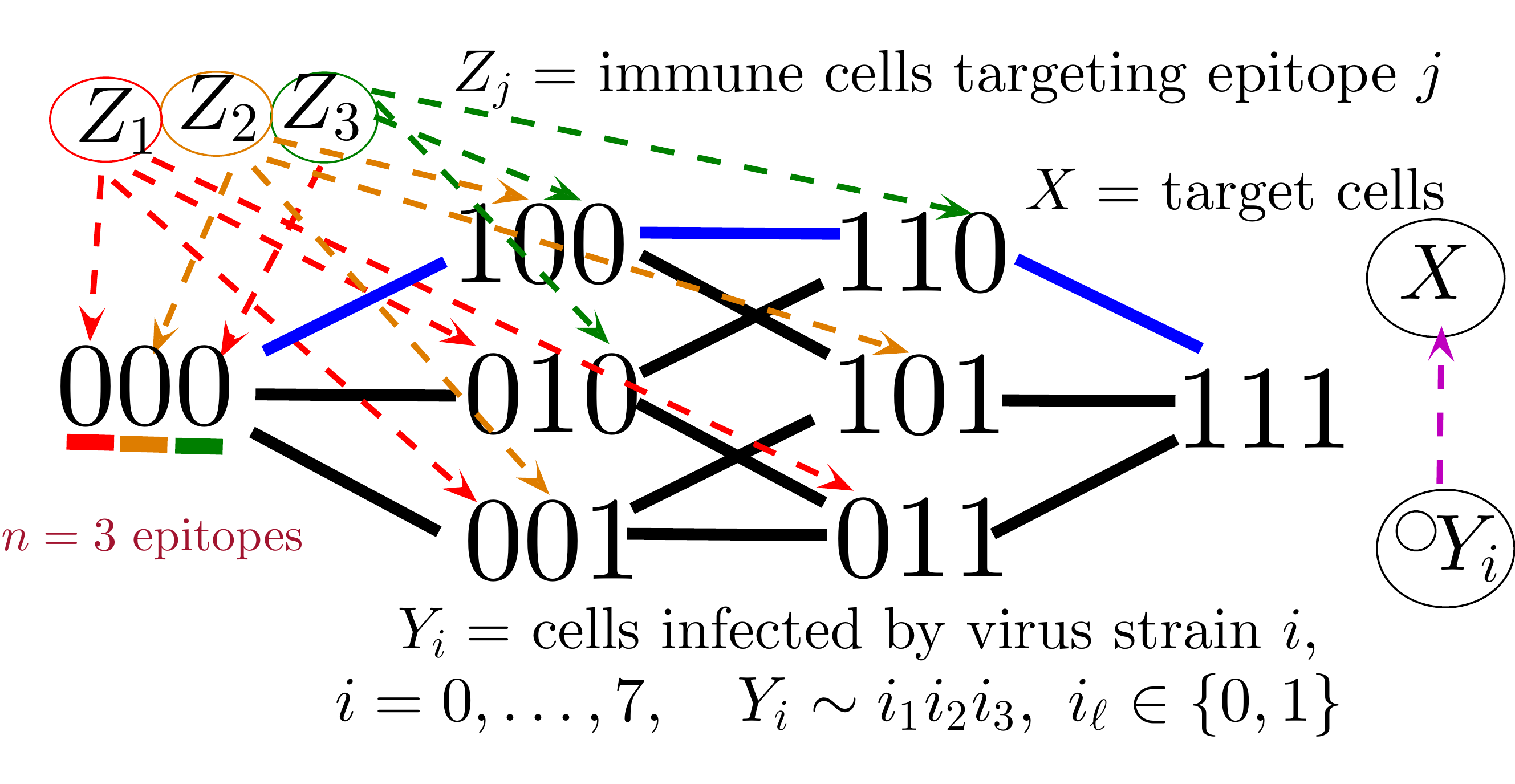}} \\
\subfigure[][]{\label{fig2b}\includegraphics[width=5.5cm,height=3cm]{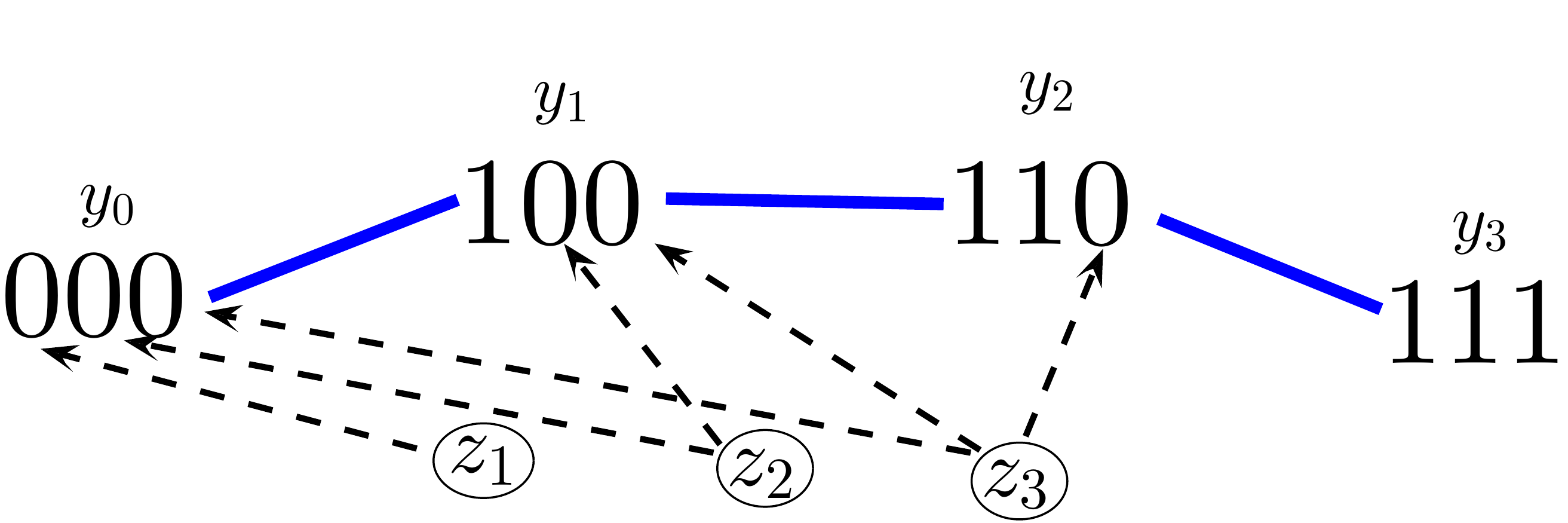}}
\subfigure[][]{\label{fig2c}\includegraphics[width=4.35cm,height=3cm]{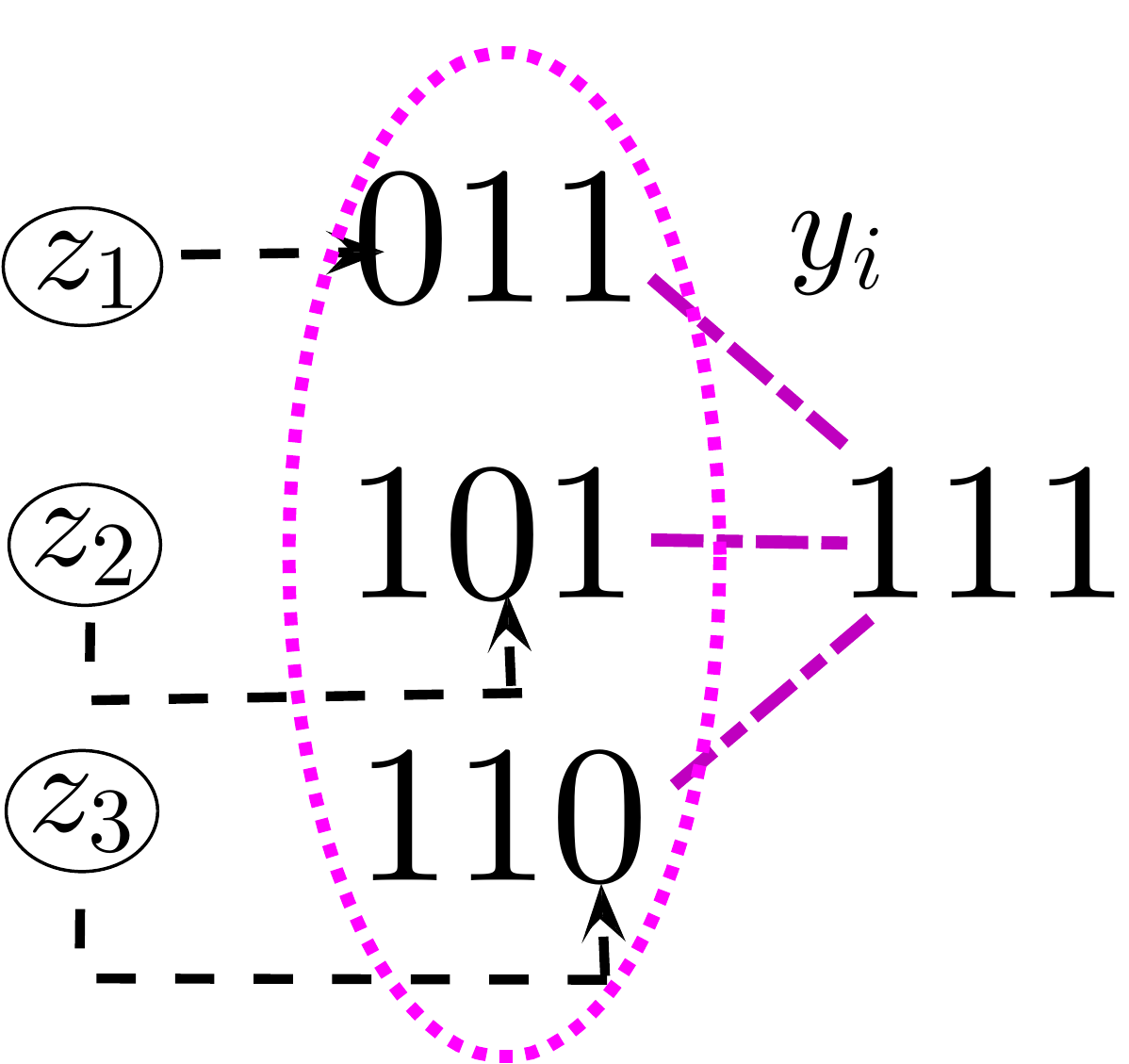}}
\subfigure[][]{\label{fig2d}\includegraphics[width=4.35cm,height=3cm]{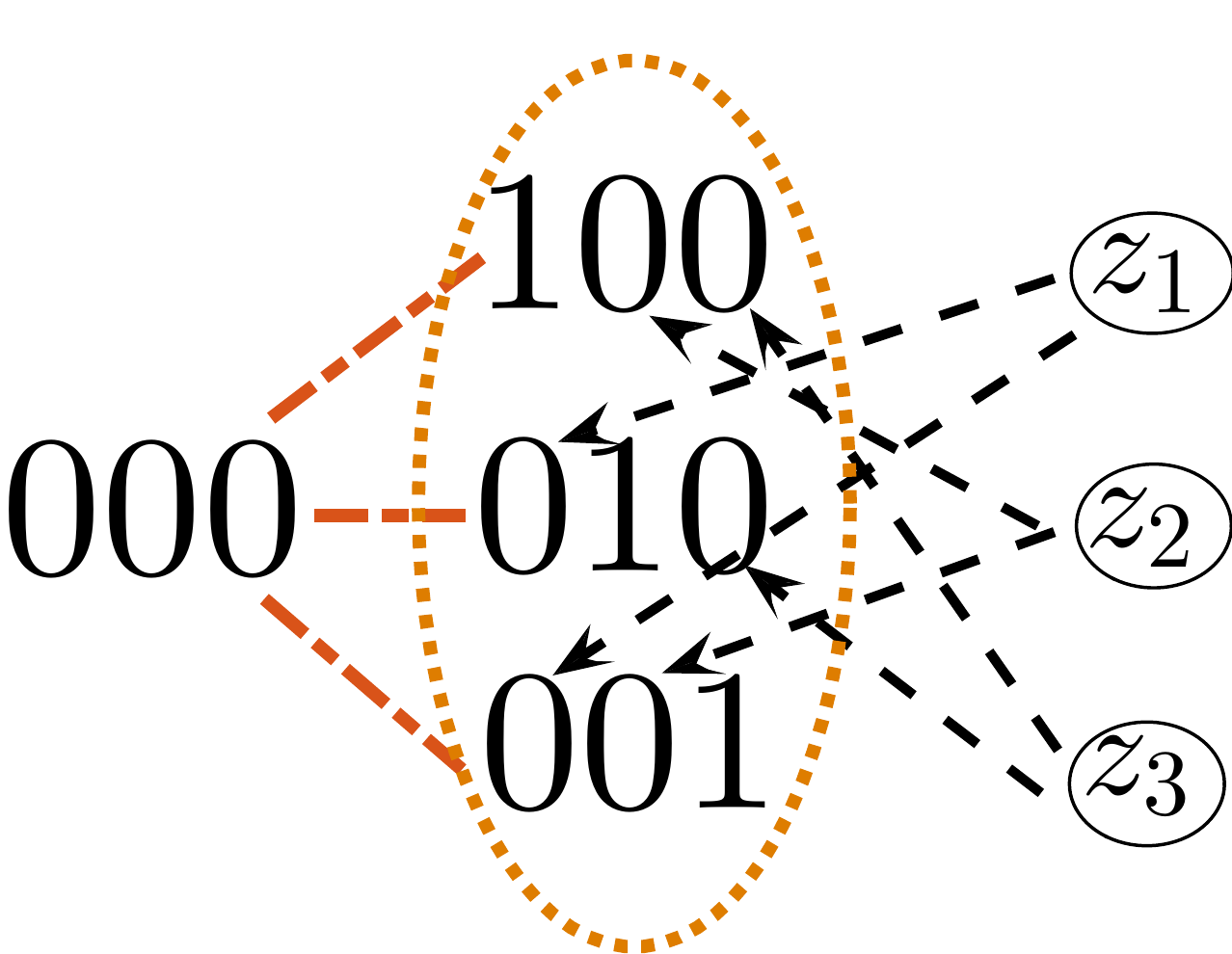}}
\caption{\textbf{(a)}  The virus-immune network on $n=3$ epitopes for model \eqref{odeS} overlying the hypercube. Here each viral strain $y_i$, $i=0,\dots,7$, is associated with a unique binary string $\mathbf i\in\left\{0,1\right\}^3$ coding their allele type, susceptible (0) or resistant (1), at each epitope.  Immune response $z_j$ attacks $y_i\sim \mathbf i$ if $i_j=0$, or equivalently if $j$ is in epitope set of $y_i$ ($j\in\Lambda_i$). The wild-type virus, $y_0\sim 000$, can evolve resistance to each epitope-specific immune response $z_j$ by successive single epitope mutations forming a path in the hypercube graph to the completely resistant viral strain ($111$).  The number of epitope mutations which viral strain $y_i$ has accumulated is $d(y_i,y_0)$ (Hamming distance between $\mathbf i$ and $000$).   Note that system \eqref{odeS} does not explicitly include mutation between viral strains.  \textbf{(b)}  The perfectly nested network, as a subgraph of the hypercube.  In this case, sequential mutations of epitopes appear in immunodominance order with specialist to generalist virus (prey) resistance and immune (predator) attack. \textbf{(c)}  The one-to-one network, with strain-specific immune responses, is representative of a completely modular ecosystem.   \textbf{(d)}  The $\leq 1$ mutation network signifyies constrained evolution.   All three subgraphs appear as feasible equilibrium structures of the system and are analyzed in Section \ref{mrsec}. }
  \label{fig2}
  \end{figure}

The $2^n$ potential virus strains can be viewed in a fitness landscape; each strain $\mathbf i$ is a vertex in an $n$-dimensional \emph{hypercube graph} with fitness $\mathcal R_{\mathbf i}$, as shown in Figs. \ref{fig2a} and \ref{fig2d} in the case of $n=3$ and $n=2$ epitopes.  Viral strains $y_{\mathbf i}$ and $y_{\mathbf k}$ are connected by an edge,  if the sequences $\mathbf i$ and $\mathbf k$ differ in exactly one bit, i.e. their \emph{Hamming distance} -- denoted by $d(y_{\mathbf i},y_{\mathbf k})$ -- is one.   Each mutation of an epitope comes with a fitness cost, so we assume that 
\begin{align}\label{fitnesscost}
\text{If} \ d(y_{\mathbf i}, y_{\mathbf k})=1 \ \text{and} \ d(y_{\mathbf i},y_w)<d(y_{\mathbf k},y_w), \ \text{then} \  \mathcal R_{\mathbf i} > \mathcal R_{\mathbf k} .
\end{align}
The occurrence of fitness costs (in reproduction rate) for gaining resistance to predation is a general concept in eco-evolutionary systems, for example in bacteria-phage networks.
Finally we say that an immune response $z_j$ is \emph{immunodominant} over another immune response $z_k$ if $\mathcal I_j>\mathcal I_k$ and assume, without loss of generality, the ordered \emph{immunodominance hierarchy}; \begin{align}\label{immunodom}
 \mathcal I_1\geq \mathcal I_2 \geq \dots \geq \mathcal I_n, \ \ \text{i.e.}  \ \  s_1\leq s_2 \leq \dots \leq s_n.
\end{align}

System \eqref{odeS} generalizes many previous model structures in the sense that they can be seen as subgraphs of our ``hypercube network''.  For instance, the ``strain-specific'' (virus-immune response) network \cite{nowak1996population} (also called ``one-to-one network'' in phage-bacteria models \cite{jover2013mechanisms,korytowski2015nested}) is equivalent to restricting \eqref{odeS}  to the $m=n$ viral strains which have mutated $n-1$ epitopes (Figure \ref{fig2b}).  The ``perfectly nested network'' restricts \eqref{odeS} to the $m=n+1$ viral strains which have sequential epitope mutations in the order of the immunodominance hierarchy (Figure \ref{fig2b}).  Nested networks were considered in HIV models \cite{browne2016global}, along with phage-bacteria models \cite{korytowski2015nested}, and may be a common persistent structure in ecological communities \cite{gurney2017network}.  The ``full hypercube network'' has been considered for modeling CTL escape patterns in HIV infected individuals \cite{Althaus,vanDeutekom}.    

\section{Necessary population dynamics/genetics definitions and results}
\subsection{Stability and persistence}
First we review some relevant definitions and results on the equilibria and asymptotic dynamics valid in the general model \eqref{ode3} that are further detailed in prior work \cite{browne2018dynamics}.  For a non-negative equilibrium point, $\mathcal E^*=\left(x^*,y^*, Z^* \right)\in\mathbb R_+^{1+m+n}$, define the ``\emph{persistent variant sets}'' associated with $\mathcal E^*$ as:
\begin{align}
\Omega_y= \left\{i \in [1,m]: y^*_i>0 \right\} \quad \text{and} \quad \Omega_z= \left\{j \in [1,n]: Z^*_j>0 \right\}.
\end{align}
  In addition, define the following subsets of $\mathbb R_+^{1+m+n}$:
 \begin{align}
 \Omega &= \left\{ \left(x,y, Z\right)\in \mathbb R_+^{1+m+n} \ | \  y_i,z_j>0,\ i \in\Omega_y, j\in\Omega_z\right\}, \  \Gamma_{\Omega} = \Omega \cap \left\{    y_i,z_j=0, \ i\notin\Omega_y,j\notin\Omega_z\right\}.
 \end{align}
 Here $\Gamma_{\Omega}$, consisting of only those state vectors having the same set of positive and zero components as equilibrium $\mathcal E^*$, is called the \emph{positivity class} of $\mathcal E^*$.  Notice that the dimension of the subset $\Gamma_{\Omega}$ is $1+|\Omega_y|+|\Omega_z|$, where the notation $|\Omega_y|$ ($|\Omega_z|$) denotes the \emph{cardinality} of the set $\Omega_y$ ($\Omega_z$).  The equilibrium $\mathcal E^*$ must satisfy the following equations:
\begin{align}
\sum_{i\in\Omega_y} a_{ij} y_i^* &= \rho_j, \quad j\in \Omega_z  \label{LVequil1}  \\
\sum_{j\in\Omega_z} a_{ij} Z_j^* &= \mathcal R_i x^* -1,  \quad i\in \Omega_y  \label{LVequil2} \\ 
1+ \sum_{i\in\Omega_y} \mathcal R_i y_i^* &= \frac{1}{x^*}  \label{genEqcon} 
\end{align}
 We note that $\mathcal R_i>1,\ i\in \Omega_y$ must hold, even in the absence of immune response.

The following proposition provides the condition for uniqueness of an equilibrium within a positivity class, and shows that in such equilibria the number of virus strains either is equal to or exactly one more than the number of immune responses.
\begin{proposition}\cite{browne2018dynamics}\label{prop33}
Suppose the equilibrium $\mathcal E^*=\left(x^*, y^*, Z^* \right)$ exists in positivity class $\Gamma_{\Omega}$, where $\left(y^*,Z^*\right)$ satisfy the linear system of equations \eqref{LVequil1}-\eqref{LVequil2} and the cardinality of $\Omega_y$ and $\Omega_z$ are $|\Omega_y|=m'$ and $|\Omega_z|=n'$.  Then $\mathcal E^*$ is the unique equilibrium in $\Gamma_{\Omega}$, i.e. $\vec v=\left(y^*,Z^*\right)^T$ is the unique solution to \eqref{LVequil1}-\eqref{LVequil2}, if and only if ${\rm Ker}(A')^T\cap \vec{\mathcal R}'^{\perp}=\left\{0\right\}$ and ${\rm Ker}(A')=\left\{0\right\}$.
Moreover, if $\mathcal E^*$ is  not unique in its positivity class $\Gamma_{\Omega}$, then $\Gamma_{\Omega}$ contains an infinite number (a continuum) of equilibria.  Conversely, if $\mathcal E^*$ is unique in a positivity class $\Gamma_{\Omega}$ (with $m'$ and $n'$ persistent virus and immune responses), 
 then one of the following holds:
\begin{itemize}
\item[(i)] $m'=n'$, and $x^*=1/\left(1+(\vec{\rho}\,')^T(A')^{-1}\vec{\mathcal R}'\right)$.
\item[(ii)] $m'=n'+1$, and $x^*=\vec 1^{\,T} C^{-1}_{(n'+1)}$, where $C^{-1}_{(n'+1)}$ is the last column in the $(n'+1)\times (n'+1)$ matrix inverse of $C=\begin{pmatrix} A' & \vec{\mathcal R'}\end{pmatrix}^T$.
\end{itemize}
\end{proposition}
This proposition, along with prior results on competitive exclusion, demonstrate that virus (prey) or ecosystem diversity in our model is entirely mediated by the immune response (predator) populations.  Thus the model provides a good system for exploring how prey-predator ecosystems can diversify and patterns in their underlying structure.  

Next we are concerned with the stability of equilibria, and which populations persist in the long run.  First, based on the idea of being ``weakly stable'' against missing species \cite{hofbauer1998evolutionary}, we call an equilibrium $\mathcal E^*=(x^*,y^*,Z^*)$ of \eqref{ode3} \emph{saturated} if the following holds:
\begin{align}
\frac{\dot y_i}{\gamma_i y_i}_{\big\rvert_{\mathcal E^*}}=\mathcal R_i x^*-1-\sum_{j\in\Omega_z} a_{ij} Z_j^*\leq 0, \  \forall i\notin\Omega_y, \quad  \frac{s_j\dot z_j}{\sigma_jz_j}_{\big\rvert_{\mathcal E^*}}= \sum_{i\in\Omega_y} a_{ij} y_i^*-\rho_j \leq 0,  \ \forall j\notin\Omega_z  \label{inequ}
\end{align}
Here each term in \eqref{inequ} gives the sign of the ``invasion rate'' of a missing species.  For a notion of persistent populations, define $\Omega_{yz}$ \emph{persistence} as
 \begin{align*}
 & \exists \ \epsilon>0 \ \text{and} \ T(\vec w_0) \ \text{such that} \  y_i(t),Z_j(t) >\epsilon, \ i \in\Omega_y, j\in\Omega_z, \ \forall t>T(\vec w_0), \ \text{and} \\
  & \lim_{t\rightarrow\infty}y_i(t),Z_j(t)= 0, \ i \notin\Omega_y, j\notin\Omega_z, \ \ \text{for every solution with initial condition} \ \vec w_0 \in \Omega.
   \end{align*}  
   We describe the individual populations $i \in\Omega_y, j\in\Omega_z$ in the above definition of $\Omega_{yz}$ \emph{persistence} as being \emph{uniformly persistent}.  Now we state a main theorem of \cite{browne2018dynamics} concerning the stability of equilibria and persistence of viral and immune variants of model \eqref{ode3}.   
\begin{theorem}[\cite{browne2018dynamics}]\label{genThm}
Suppose that $\mathcal E^*=\left(x^*,y^*, Z^* \right)$ is a non-negative equilibrium of system (\ref{ode3}) with positivity class $\Gamma_{\Omega}$.  Suppose further that $\mathcal E^*$ is saturated, i.e. the inequalities (\ref{inequ}) hold.
Then $\mathcal E^*$ is locally stable and $x(t)\rightarrow x^*$ as $t\rightarrow\infty$.

Furthermore, if $\mathcal E^*$ is the unique equilibrium in its positivity class $\Gamma_{\Omega}$ and the inequalities (\ref{inequ}) are strict, then  $y_i,Z_j\rightarrow 0$ for all $i\notin\Omega_y, j\notin\Omega_z$.  If $i\in \Omega_y$ and $a_{ij}=0 \ \forall j\in \Omega_z$, i.e. $\Lambda_i\cap \Omega_z=\emptyset$, then $y_i\rightarrow y_i^*$ and $x^*=1/\mathcal R_i$.  In addition, assuming positive initial conditions, for each $i\in\Omega_y, j\in\Omega_z$,  $y_i$ and $Z_j$ persist (the system is $\Omega_{yz}$ permanent) with asymptotic averages converging to equilibria values, i.e.
\begin{align*}
\lim_{t\rightarrow\infty} \frac{1}{t} \int\limits_0^t y_i(s) \, ds = y_i^*, \quad \lim_{t\rightarrow\infty} \frac{1}{t} \int\limits_0^t Z_j(s) \, ds = Z_j^*,
\end{align*}
In the case that there are less than or equal to two persistent viral strains with non-empty epitope sets (restricted to $\Omega_z$), i.e. $|\left\{ i\in\Omega_y: \Lambda_i\cap \Omega_z\neq \emptyset \right\}|\leq 2$, then $\mathcal E^*$ is globally asymptotically stable.
\end{theorem}
Note that the global convergence of the persistent variants to equilibria values is still an open question when there are more than two persistent immune responses.

\subsection{Fitness and epistasis}
In the rest of this article we consider the ``binary sequence'' case of model \eqref{ode3}, which leads to the simplified system \eqref{odeS} through assumption \eqref{assumeaj}. For our virus-immune ecosystem, we are considering the situation where $n$ immune response populations $z_j$ each targeting the corresponding epitope $j$ in the virus strains at a rate solely dependent on the allele type of epitope $j$; (0) wild-type or (1) mutated form conferring full resistance to $z_j$.  The avidity of immune response $z_j$ and (wild-type) epitope $j$ is described by the immune reproduction number $\mathcal I_j$ given by \eqref{IR0}, and are according to the immunodominance hierarchy \eqref{immunodom}.
 As opposed to this simple immune fitness ordering, the collection of virus reproduction numbers (fitnesses) in our model can have much more complex relationships among each other depending on the fitness landscape, formally defined below.

Consider the space of binary sequences of length $n$, $\left\{0,1\right\}^n$, which contain all possible $2^n$ virus strains. For a given strain $i$ with sequence $\mathbf i\in\left\{0,1\right\}^n$, we also denote its reproduction number in terms of binary sequence; $\mathcal R_{\mathbf i}$.  The reproduction numbers can be described in terms of the fitness cost (relative to wild-type fitness $\mathcal R_{\mathbf 0}$) associated with the corresponding combinations of epitope mutations.  The \emph{fitness landscape} is defined as the precise map between the virus sequences and their reproduction numbers: 
$$w: \left\{0,1\right\}^n \rightarrow \mathbb R, \quad w(\mathbf i) =\mathcal R_{\mathbf i}.  $$  The set of all reproduction numbers is the image of the fitness landscape, $$\mathcal F:=w\left(\left\{0,1\right\}^n\right)=\left\{\mathcal R_{\mathbf i}\right\}_{\mathbf i\in\left\{0,1\right\}^n}=\left\{\mathcal R_{i}\right\}_{i=0}^{2^n-1}, $$
where we can utilize either the sequence or integer indices for viral strains.  An important special case of a fitness landscape is when $w$ is \emph{additive}.  In an additive fitness landscape, 
\begin{align}
\mathcal R_{\mathbf i}=\mathcal R_{\mathbf 0}-\mathbf c \cdot \mathbf i, \label{additive_def} 
\end{align} where $\mathbf c=\left< c_1,c_2,\dots,c_n\right>$ is the vector of individual fitness costs for mutating each epitope, with the assumption that $\mathbf c \cdot \mathbf 1<\mathcal R_{\mathbf 0}$ so that all viral reproduction numbers remain positive.  

Whereas an additive fitness landscape is solely
determined linearly by the wild-type and single-mutant fitness values, the concept of epistasis allows for combinations of mutations to have more general nonlinear fitness landscapes.  Informally, a system has epistasis if the effect of a mutation depends on genetic background.  Here we generally define epistasis as a deviation from additivity.  A common way to incorporate epistasis is via pairwise interactions between loci, as in the quadratic Ising or Pott's model \cite{stadler2002fitness} which has been used in applications to HIV-immune data \cite{Barton}.  Let $B$ be a strictly upper triangular matrix encoding (possibly random) pairwise interactions and define 
\begin{align}
\mathcal R_{\mathbf i}=\mathcal R_{\mathbf 0}-\mathbf c \cdot \mathbf i +\sum_{j=1}^n i_j \sum_{k>j} i_k B_{jk}, \quad \mathbf i\in\left\{0,1\right\}^n, \label{pwfit}
\end{align}
where $B$, $\mathbf c$ are suitable to fit our requirements for the viral fitness (cost) landscape \eqref{fitnesscost}.

To consider epistasis in general, first consider a subset of the sequence space $\mathcal S\subset \left\{0,1\right\}^n$ and the associated fitness landscape occupied by the sequences, $w(\mathcal S)$.  Define a \emph{vanishing linear form} on $\mathcal S$ as a linear form $g=\sum_{\mathbf k\in\mathcal S} a_{\mathbf k} \mathcal R_{\mathbf k}$ with integer coefficients $a_{\mathbf k}$, which is zero for any fitness landscape $w$ that is additive, and satisfies $\sum_{\mathbf k\in\mathcal S} a_{\mathbf k} =0$ with some $a_{\mathbf k} \neq 0$ \cite{beerenwinkel2007epistasis,crona2017inferring}.  Note that an equivalent definition, can be formed from the following observation upon consideration of additive fitness \eqref{additive_def}:  
\begin{align*}
g&=\sum_{\mathbf k\in\mathcal S} a_{\mathbf k} \mathcal R_{\mathbf k} =\sum_{\mathbf k\in\mathcal S}a_{\mathbf k}\left(\mathcal R_{\mathbf 0}-\mathbf c \cdot \mathbf k\right) =-\sum_{\mathbf k\in\mathcal S}  a_{\mathbf k}\mathbf c \cdot \mathbf k=-\mathbf c \cdot \sum_{\mathbf k\in\mathcal S} a_{\mathbf k} \mathbf k  \\
\Rightarrow g&=0 \ \forall \mathbf c\in\mathbb R^n  \Leftrightarrow  \sum_{\mathbf k\in\mathcal S} a_{\mathbf k} \mathbf k=\mathbf 0.
\end{align*}
So a vanishing linear form on $\mathcal S\subset \left\{0,1\right\}^n$ equivalently satisfies $\sum_{\mathbf k\in\mathcal S} a_{\mathbf k} \mathbf k=\mathbf 0$ and $\sum_{\mathbf k\in\mathcal S} a_{\mathbf k} =0$ with not all $a_{\mathbf k} = 0$. 
The two conditions can be combined by adding to every binary sequence in $\left\{0,1\right\}^n$ a $1$ at the end of the sequence.  Considering each extended binary sequence $\mathbf i 1= (i_1\dots i_n 1)$ as a vector in $\mathbb R^{n+1}$, existence of a vanishing linear form on $\mathcal S\subset \mathbb R^{n+1}$ simply signifies $\mathcal S$ to be a linearly dependent set of vectors.  

A \emph{circuit} $\mathcal C\subset \left\{0,1\right\}^n$ is a minimal set which has a vanishing linear form.  In other words, there exists a vanishing linear form on a circuit $\mathcal C$ and no proper subset of $\mathcal C$ has a vanishing linear form.    Considering the extended binary sequences in $\mathbb R^{n+1}$, a circuit is a minimally linearly dependent collection of vectors, i.e. a linearly dependent collection of vectors in which any proper subset is linearly independent \cite{crona2020rank}.  Circuits allow for detection of sign epistasis \cite{beerenwinkel2007epistasis,crona2017inferring}, which can be defined as follows.  Suppose $\mathcal C$ is a circuit with vanishing linear form $g=\sum_{\mathbf k\in\mathcal C} a_{\mathbf k} \mathcal R_{\mathbf k}$.  The circuit $\mathcal C$ has \emph{positive epistasis} for fitness landscape $w$ if $\sum_{\mathbf k\in\mathcal C} a_{\mathbf k} \mathcal R_{\mathbf k}>0$.  We analogously define \emph{negative epistasis} on the circuit $\mathcal C$ if $\sum_{\mathbf k\in\mathcal C} a_{\mathbf k} \mathcal R_{\mathbf k}<0$.  In a strictly additive fitness landscape the vanishing linear forms on each circuit would all be zeros, i.e. vanish.  The signs of the coefficients in a circuit are not unique since there are two possible assignments of positive versus negative coefficients, but for the type of circuits we will concern with this in this paper, we will define a unique way of assigning signs based on stability results.

The simplest class of circuits measure the conditional or marginal epistasis of groups of loci.  In particular, against a background where a subset of loci are fixed, consider two distinct pairs of \emph{(ones') complement} sequences, defined to be sequences $\mathbf{\tilde k}$ and $\overline{\mathbf{\tilde k}}$ where $\mathbf{\tilde k}+\overline{\mathbf{\tilde k}}=\mathbf{\tilde 1}$ for a \emph{subset} of loci $\tilde J\subset[1,n]$.  For example, in the case $n=3$, $\left\{100, 010\right\}$ and $\left\{000, 110\right\}$ are two distinct pairs  which together form a circuit giving the marginal epistasis of the first two loci against the third fixed locus, with the following linear form:
\begin{align*}
\mathcal A_{12}(w)=\mathcal R_{000}-\mathcal R_{100}-\mathcal R_{010}+\mathcal R_{110}.
\end{align*} 
Here the circuit epistasis sign is positive for a fitness landscape whenever the pairwise interaction between epitopes 1 and 2 are synergistic, so that the double mutant has larger reproduction number than it would have under additivity.  We can identify how the ``pairwise epistatic fitness landscape'' \eqref{pwfit} directly relates to conditional epistasis circuits in any dimension $n$. Consider loci $1\leq j< k\leq n$.  A circuit measuring the conditional epistasis of $j,k$ against any background will resolve as follows; $\mathcal A_{jk}:=-\mathcal R_{\cdot 0 \cdot 0 \cdot}+\mathcal R_{\cdot 1 \cdot 0 \cdot}+\mathcal R_{\cdot 0 \cdot 1 \cdot}-\mathcal R_{\cdot 1 \cdot 1 \cdot}=B_{jk}$, where the changing alleles occur in the $j,k$ positions.  An example of a marginal epistasis circuit can be given by linear form $\mathcal A(w)=-\mathcal R_{000}+\mathcal R_{001}+\mathcal R_{110}-\mathcal R_{111}$, which relates marginal epistasis of the first 2 loci (as a block) and the third locus.  
Note that because the sum of coefficients and weighted sum of sequences must vanish, along with a circuit being minimal, the number of binary sequences in a ``\emph{(ones') complement}'' circuit must be four.  In general, the number of circuits rapidly grows with $n$ (there are 20 circuits for $n=3$, 1348 circuits for $n=4$ \cite{eble2020higher}) and can be interpreted geometrically in terms of shapes formed by vertices of the $n$-cube \cite{beerenwinkel2007epistasis}.  

\section{Main Results} \label{mrsec}

In this section, we present our main theorems and their ramifications.  Proofs to new results appear in the Appendix. Our major goal is to rigorously connect the concept of circuits with bifurcations and stable equilibria in  model \eqref{odeS}.  First, in order to demonstrate a general link between circuits and the dynamical system, we establish that persistent viral strains comprise a circuit only in a critical case. In particular, we show that a circuit has positive components in a feasible equilibrium only when this circuit is additive with respect to fitness landscape, in which case a degenerate infinite dimensional subspace of equilibria appears.     Indeed, the following proposition generalizes a previous result in \cite{browne2018dynamics} on degeneracy of equilibria forming a cycle in virus sequence hypercube.
\begin{proposition}\label{prop1new}
 Consider the binary sequence model \eqref{odeS} with $2^n$ viral strains identified in $\left\{0,1\right\}^n$. Suppose that $\mathcal C\subset \left\{0,1\right\}^n$ is a \emph{circuit}, has vanishing linear form $g=\sum_{\mathbf k\in\mathcal C} a_{\mathbf k} w(\mathbf k)$ for any additive fitness landscape $w$, and consider the fixed fitness landscape with image (reproduction numbers) denoted by $\mathcal R_{\mathbf k}$ for $\mathbf k\in\left\{0,1\right\}^n$.      If $\sum_{\mathbf k\in\mathcal C} a_{\mathbf k} \mathcal R_{\mathbf k} \neq 0$, then there does not exist an equilibrium $y^*$ with $y^*_{\mathbf k}>0$ for all $\mathbf k\in\mathcal C$.  On the other hand if $\sum_{\mathbf k\in\mathcal C} a_{\mathbf k} R_{\mathbf k}=0$ and there exists an equilibrium with $y^*_{\mathbf k}>0$ for all $\mathbf k\in\mathcal C$, then there are infinitely many equilibria, $\bar{\mathbf y}$, in the positivity class of $\mathbf y^*$, with components parametrized by $y_{\mathbf k}=y^*_{\mathbf k}+\alpha a_{\mathbf k}$ for some $\alpha\in\mathbb R$.  
\end{proposition}
The proposition implies that any equilibrium with persistent strains forming a circuit must be unstable, in particular as part of a continuum of equilibria.  The dimension of the infinite dimensional subspace of equilibria is the number of linearly independent vanishing forms corresponding to the circuit, where the dimension can be greater than one if the circuit contains distinct (sub-) circuits as subsets.  Although unstable, the lines of equilibria will be seen in the ensuing sections as bifurcations where certain types of stable equilibria are invaded with strain replacement and stability being sharply determined by signed epistasis of the corresponding circuits.  


\subsection{Nested network determined by epistasis} \label{nestsect}

Next, we focus on (perfectly) nested equilibria, which describe sequential mutations of epitopes in the order of the immunodominance hierarchy and persistence of all strains along this pathway.   The successive rise of more broadly resistant prey (coming with a fitness cost) and weaker but more generalist predators, in a nested fashion, has been proposed in bacteria-phage communities \cite{jover2013mechanisms,korytowski2015nested,weitz2013phage}, and there is some evidence that nestedness is a feature of HIV and immune response dynamics \cite{kessinger2015inferring,liu2013vertical,vanDeutekom}.   
  Furthermore, this specialist-generalist structure is a well studied pattern in a variety of ecosystems, in particular nested networks are of interest in explaining the biodiversity and structure of 
 mutualistic (e.g. plant-pollinator) communities \cite{bascompte2003nested}.  

First, we describe equilibria of model \eqref{odeS}, where the persistent network is constrained to be nested, which were described in \cite{jover2013mechanisms,korytowski2015nested,browne2016global}.  We introduce a ``nested priority'' indexing for the viral strains, which allows convenient definition of threshold quantities for nested networks.  The $n+1$ binary sequences contained in nested equilibria are of the form $1^k0^{n-k}$ (in power notation for the length $n$ binary string), where $0\leq k \leq n$.  Let $y_k$ denote the viral strain with binary sequence $1^k0^{n-k}$, $0\leq k \leq n$.  For $k\geq 1$ define:
 \begin{align}
 \mathcal Q_{k}= \mathcal Q_{k-1}+ (s_{k}-s_{k -1})\mathcal R_{k-1}, \quad\text{where} \quad  \mathcal Q_0=1, s_0=0, s_k=1/\mathcal I_k. \label{nestedQ}
  \end{align}
Then, for each $k\in[1,n]$, define the following \emph{nested} equilibria:
 \begin{align}
 \widetilde{\mathcal E}_{k}   = (\widetilde x,\widetilde y,\widetilde z),  \qquad  & \widetilde x=\frac{1}{\mathcal R_{k}}, \ \widetilde y_j=s_j-s_{j-1}\  \ \text{for} \ \ 0\leq j<k, \ \widetilde y_{k}=1-\frac{ \mathcal Q_{k}}{\mathcal R_{k}},\label{equilib1} \\  & \widetilde z_j=\frac{\mathcal R_{j-1}-\mathcal R_{j}}{\mathcal R_{k}} \ \  \ \text{for} \ \ 1\leq j< k,  \ \widetilde z_{k}=0, \quad \widetilde y_j=\widetilde z_j=0 \  \ \ \text{for} \ \ k<j\leq n \notag \\
 \bar{\mathcal E}_k = (\bar x,\bar y,\bar z),  \qquad &  \bar x=\frac{1}{\mathcal Q_k}, \ \bar y_j=s_j-s_{j-1}\  \ \text{for} \ \ 0\leq j<k,  \label{equilib2}  \\ & \hspace{-.5cm} \bar z_j=\frac{\mathcal R_{j-1}-\mathcal R_{j}}{\mathcal Q_k} \  \ \text{for} \ \ 1\leq j<k,  \ \bar z_k=\frac{\mathcal R_{k-1}}{\mathcal Q_k}-1, \quad  \bar y_j=\bar z_j=0 \  \ \ \text{for} \ \ k<j\leq n  \notag
  \end{align}
Equilibrium $\widetilde{\mathcal E}_{k}$ represents the appearance of escape mutant $y_{k}$ from the equilibrium  $\bar{\mathcal E}_k$ containing $k$ viral strains $y_0,\dots,y_{k-1}$ and immune responses $z_1,\dots z_k$.  The stability of these equilibria \emph{restricted within the nested network} (non-nested strains $y_{n+1},\dots y_{2^n-1}$ are set to zero) was proved to be determined which of equilbria \eqref{equilib1} and \eqref{equilib2} are positive \cite{browne2016global}.


Along with the specialist to generalist ordering in nested equilibria, another interesting observation is that nested networks are evolutionary pathways in the full fitness landscape hypercube.  As opposed to some other feasible equilibria, such as the one-to-one network, the persistent strains in the nested equilibria form a path from the wild-type to the most resistant strain as single mutations accumulate in stepwise fashion.  In a single (quasi-)species system, the underlying viral fitness landscape, which is generally shaped by epistatic interactions, determines evolutionary trajectories.  When another trophic level is added, as immune response (predators) here, the overall viral fitnesses are expected to be dynamic since they depend upon the immune response populations.  However, here we show that the nested trajectory in our system is solely dependent on the relevant epistasis in the viral fitness landscape.  

Define a certain pathway on the hypercube of binary sequences to have \emph{positive (negative) epistasis if every circuit with all but one node contained on the path has positive (negative) epistasis}.  The interpretation of this definition is that each of these circuits represent potential alternate pathways, which correspond to strain invasion in the model.  Our main result, Theorem \ref{newthm1} below, proves that the nested network is stable and persistent if and only if it has positive epistasis as a pathway in the viral fitness landscape.  In particular, we decode the general saturated equilibria inequalities \eqref{inequ} conferring stability and persistence by Theorem \ref{genThm} into biological meaningful conditions on sign epistasis of associated ``invasion circuits''.  Although our model does not explicitly include mutation, the persistent variants of stable equilibria can still represent evolutionary outcomes, as later simulations show.  Thus the following theorem suggests a necessary and sufficient condition based on epistasis in the viral fitness landscape for a nested trajectory in a generalized eco-evolutionary version of model \eqref{odeS}. 
\begin{theorem}\label{newthm1}
  Consider the binary sequence model \eqref{odeS} with $2^n$ viral strains and $n$ immune responses.  Assume that $\mathcal R_0>\mathcal Q_1$ (so that at least one virus strain and immune response persists). Let $k$ be the largest integer in $[1,n]$ such that $\mathcal R_{k-1} > \mathcal Q_k$.   Then $\widetilde{\mathcal E}_{k}$ (or $\overline{\mathcal E}_{k}$ if $\mathcal R_k>\mathcal Q_k$) is stable with uniformly persistent strains $y_0,y_1,\dots, y_k$ (and $y_{k+1}$ if $\mathcal R_k>\mathcal Q_k$) if and only if (saturated) inequalities \eqref{inequ} hold or equivalently each of the $2^n-n-1$ invasion circuits corresponding to a non-nested strain ($y_i, \ i=n+1,\dots,2^n-1$) union a subset of nested strains ($\mathcal S\subset\left\{y_0,y_1,\dots,y_n\right\}$) has positive epistasis.   In other words, the nested network is stable and persistent if and only if it has positive epistasis as a pathway in the viral fitness landscape.
\end{theorem}

We provide two proofs of the above theorem, given in the appendix.  First, we prove the stability condition pattern by adopting a linear algebra approach where each binary sequence is extended by an additional fixed bit.  This leads to a solvable system of equations for the linear forms and circuits determining nested equilibria stability.   Second, we apply a combinatorial technique to find the strains in the nested network forming the circuit and linear form for each possible invading strain not in the nested network.  In particular, we distinguish a ``non-nested sequence'' $\mathbf i$ by existence of a (01) string, and utilize an induction argument on the number of such strings.  Each method yields equivalent, yet distinct, characterizations of the critical circuits $\mathcal C_i$ and linear forms $\mathcal A_i$, summarized  below in a corollary to Theorem \ref{newthm1}. 

\begin{corollary} \label{maincor}
A necessary and sufficient condition for stability and persistence of the nested network is the positivity of $2^n-n-1$ linear forms $\mathcal A_i$ corresponding to circuits $\mathcal C_i$, each containing a single missing strain $y_i, \ i\in[n+1,2^n-1]$, along with strains in the nested network dependent on the $y_i$ sequence $\mathbf i=(i_1\dots i_n)$ in the  following equivalent ways:
\begin{itemize}
\item[i] Define the sequence $\left(a_{j}\right), \ j=0,1,\dots,n$, where $a_{0}=1-i_{1}$, $a_{j}=i_{j}-i_{j+1}$ for $j=2,\dots,n-1$, $a_{n}=i_{n}$.  Let $\mathcal J_i$ be the nonzero terms in sequence $(a_j)$, i.e. $\mathcal J_i:=\left\{ j\in[0,n]: a_j\neq 0\right\}$, where $a_j=\pm 1$ for $a_j\in\mathcal J_i$.  Then 
\begin{align}
\mathcal C_i&=y_i \cup \left\{ y_j \right\}_{j\in\mathcal J_i}, \qquad \mathcal A_i= -\mathcal R_i +\sum_{j\in\mathcal J_i}a_j \mathcal R_{j}. \label{nestcirc1}
\end{align}
\item[ii]  Let $0\leq m_1< p_1< m_2 < \dots <p_s < m_{s+1}\leq n$ denote the positions $p_1,\dots,p_s$ beginning the $s$ (01) strings and positions $m_1,\dots,m_{s+1}$ of the last ``1'' before and after the (01) strings.  In other words, the sequence $\mathbf i$ in ``power notation'' is given by \\ $\mathbf i= 1^{m_1}0^{p_1-m_1}1^{m_2-p_1} \dots 0^{p_s-m_s}1^{n-m_{s+1}}$.  Then
\begin{align}
\mathcal C_i&=y_i \cup \left\{ y_{m_j},y_{p_j} \right\}_{j=1}^s \cup y_{m_{s+1}}=\mathbf i \cup \left\{1^{m_j}0^{n-m_j},1^{p_j}0^{n-p_j}\right\}_{j=1}^s \cup 1^{m_{s+1}}0^{n-m_{s+1}}, \notag \\
\mathcal A_i&=-\mathcal R_i +\sum_{j=1}^{s+1}\mathcal R_{m_j} -  \sum_{j=1}^{s}\mathcal R_{p_j}. \label{nestcirc2}
\end{align}
\end{itemize}
\end{corollary}

In order to illustrate Theorem \ref{newthm1} and accompanying Corollary \ref{maincor}, we first discuss the model dynamics in the case $n=2$, which is depicted in Fig. \ref{3a} and was found to have precisely 10 distinct feasible persistent variant sets (global asymptotic stability in 8 of these regimes) in \cite{browne2018dynamics}.  In this case, there is just one ``non-nested'' strain, $y_{01}$, with the single mutation escaping the second (subdominant) immune response $z_2$.  The single circuit consists of this strain together with the nested strains, totaling the whole sequence space, i.e. $\mathcal C= \left\{01,00,10,11\right\}$, along with the corresponding linear form $\mathcal A=\mathcal R_{00}-\mathcal R_{10}-\mathcal R_{01}+\mathcal R_{11}$.   Thus the sign of the single quantity $\mathcal A$ determines the stability and persistence of the nested network.  Here $\mathcal A>0$ implies that the persistent strains and positive components of the stable equilibria lie within the nested network $\mathcal N=\left\{00,10,11\right\}$.  The  precise persistence structure when $\mathcal A>0$ depends upon which of equilibria \eqref{equilib1} and \eqref{equilib2} are positive.  In particular, the diversity increases stepwise from just the wild-type virus $y_0$ to both immune responses $z_1,z_2$ and three nested strains $y_0,\dots, y_k$ based upon the largest $k$ such that $\mathcal R_{k-1}>\mathcal Q_k$ and whether $\mathcal R_k>\mathcal Q_k$, $k=1,2$, where $z_1,z_2$ persist when $\mathcal R_1>\mathcal Q_2$.   On the other hand when $\mathcal A<0$ (which implies $\mathcal R_1>\mathcal Q_2$ and $z_1,z_2$ persist), the nested equilibrium is invaded by $y_3 \ (01)$.  Yet $y_1 \ (10)$ always persists when any immune escape occurs, independent of the sign of $\mathcal A$ and even when $y_3$ would have a larger escape rate in the single epitope case.  Thus, we suggested in \cite{browne2018dynamics} that immunodominance may be the most important factor in multi-epitope escape, which was also inferred from data analysis in a previous study of HIV \cite{liu2013vertical}.



The feasible strain invasions obtained for $n=2$ in previous work \cite{browne2018dynamics} can be seen as the simplest example of a more general pattern for bifurcations from nested equilibria obtained from Theorem \ref{newthm1} and Proposition \ref{prop1new}.   When the (sign) epistasis in one of the circuits defining the nested pathway becomes negative, the nested network becomes unstable and a transcritical bifurcation occurs.  In particular, a missing strain invades the nested network when the corresponding circuit goes from positive to negative epistasis.  In the critical case of zero epistasis, or circuit additivity, there is a line of equilibria, given by Proposition \ref{prop1new}, which connects the nested equilibrium with the invasion equilibrium.  Indeed consider the nested equilibrium $\widetilde{\mathcal E}_n$.  We arrange the (persistent) nested virus components, together with the invading strain, in the vector $\widetilde{v}=( \widetilde y, 0)^T$, where the $\widetilde y$ is from \eqref{equilib1} and the last component is the invading strain, $y_i$, which is zero when at equilibrium $\widetilde{\mathcal E}_n$.  In the critical case, where the linear form $\mathcal A_i$ corresponding to circuit $\mathcal C_i$ is zero, there is a line of equilibria given by $\mathbf v^*=\widetilde{\mathbf v}-\alpha \mathbf a$ where $\mathbf a$ is the (circuit) coefficients of $\mathcal A_i$ and $0\leq \alpha \leq C$ with $C=\min\left\{ a_k \widetilde y_k : a_k>0, k=0,1,\dots,n\right\}$.  Thus, in the bifurcation where $y_i$ invades $\widetilde{\mathcal E}_n$, the invading strain replaces one of the nested strains in the circuit with positive coefficient $(a_k>0)$, in particular the above ``$C$-minimizing'' nested strain, $\arg\min\left(\left\{ a_k \widetilde y_k : a_k>0, k=0,1,\dots,n\right\}\right)$.  By the proof of Theorem \ref{newthm1}, the positive coefficients correspond to a subset of nested strains given in order as:  $y_{m_i}, \ i=1,\dot,s$, where $0\leq m_1 < p_1 < m_2< \dots <p_s < m_{s+1} \leq n+1$ count the maximal position of a $1$ before each of $s$ $01$ strings (each at position $p_1,\dots, p_s$) in the sequence of the missing strain.  
Which of these feasible strains are replaced depends on the model parameters.  Notice that if $a_k=1$ for all $k$ such that $a_k>0$, then in a feasible equilibrium after invasion by $y_i$, the replaced strain would be the ``circuit positive coefficient'' nested strain with smallest value at the nested equilibrium.  Thus, the replaced strain must have the property of being the inferior competitor in the nested hierarchy with a positive coefficient in circuit linear form.  In the following subsection, we will see a similar principle in invasion of another equilibria structure besides the nested structure, namely the one-to-one network.  

  \begin{figure}[t!]
\subfigure[][]{\label{3a} \includegraphics[width=.51\textwidth,height=.25\textwidth]{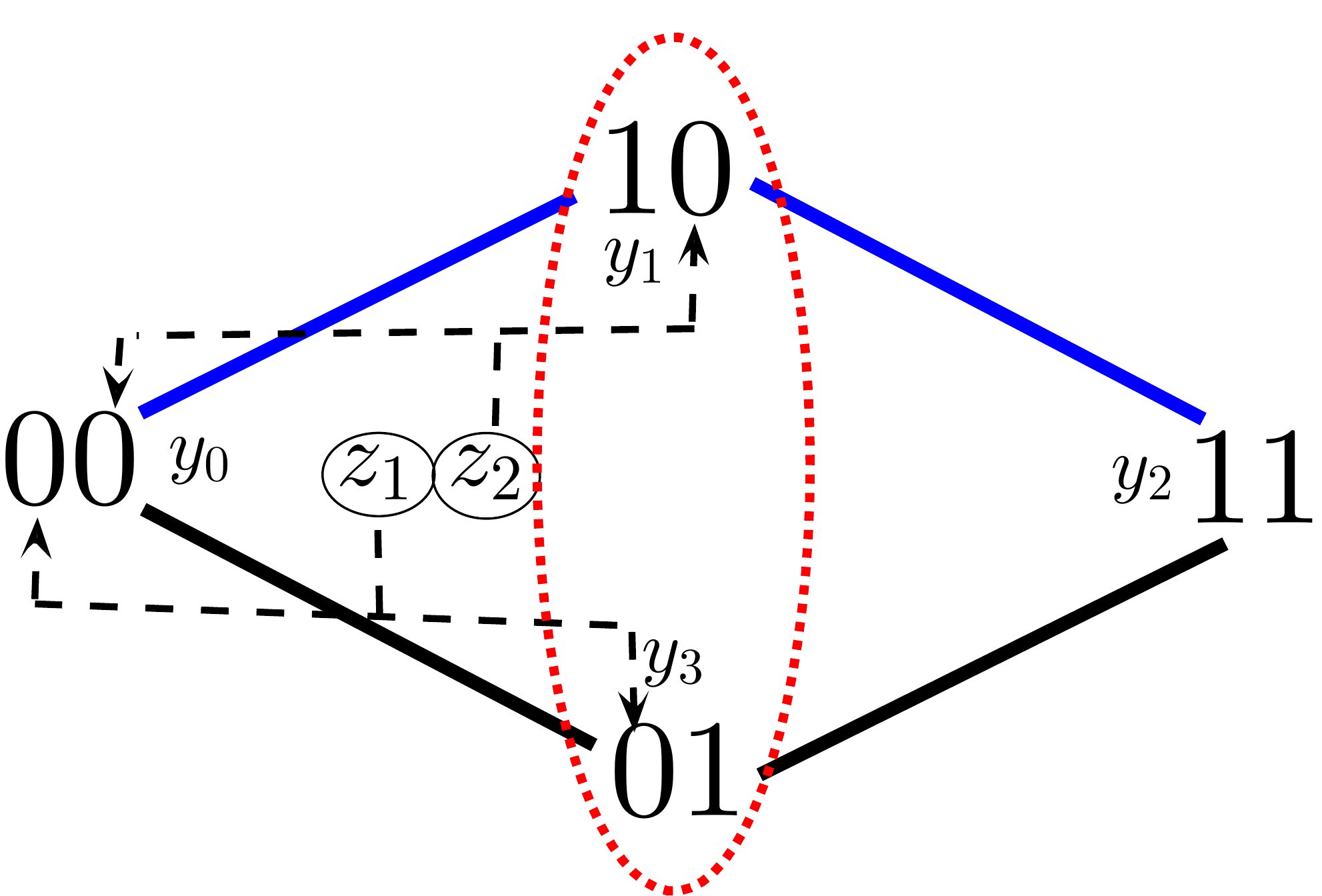}}
\subfigure[][]{ \label{3b}\includegraphics[width=.51\textwidth,height=.25\textwidth]{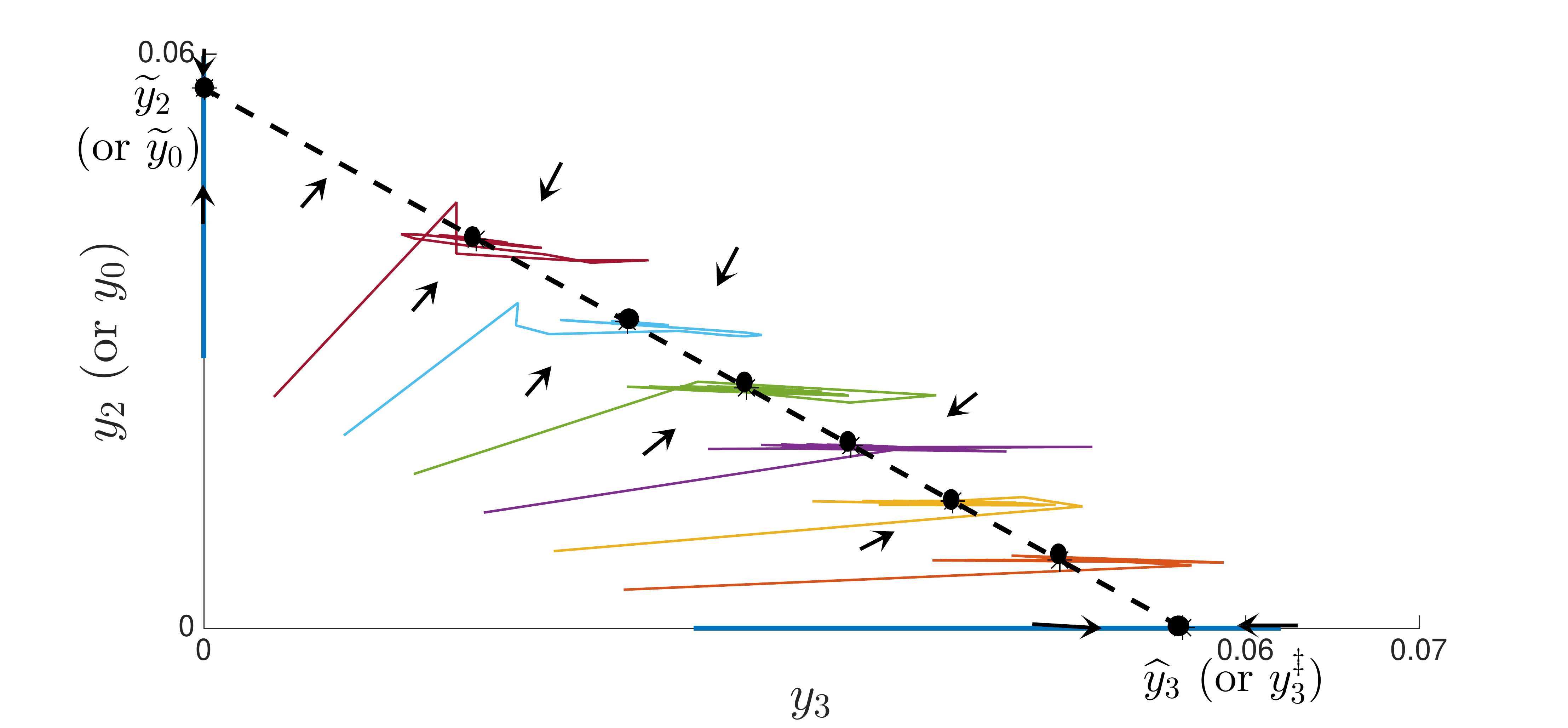}}   
\caption{ \emph{In the case $n=2$ epitopes, a single circuit of all viral strain binary sequences with corresponding linear combination of reproduction numbers determines epistasis and steady state}. There are 10 distinct regimes of persistent variants (global asymptotic stability in 8 regimes). In particular, the viral fitness epistasis measure $\mathcal A:=\mathcal R_{00}-\mathcal R_{10}-\mathcal R_{01}+\mathcal R_{11}$ (\textcolor{blue}{synergistic $\mathcal A>0$} versus \textcolor{red}{antagonistic $\mathcal A<0$}) decides stable nested equilibrium ($\widetilde{\mathcal E}_2$ or $\bar{\mathcal E}_2$ containing $\left\{00,10\right\}$ persistent strain set) versus stable one-to-one or $\leq 1$ mutation equilibrium ($\mathcal E^{\ddagger}_2$, $\mathcal E^{\dagger}_2$ or $\widehat{\mathcal E}_2$ containing $\left\{10,01\right\}$ persistent strain set), as a result of Thm.  \ref{newthm1}, \ref{ssF} and \ref{omprop}. (b) Bifurcation at $\mathcal A=0$ (additive fitness landscape) presents line of equilibria connecting nested $\widetilde{\mathcal E}_2$ and one-to-one $\mathcal E^{\ddagger}_2$ (or $\leq 1$ mutation $\widehat{\mathcal E}_2$), projected on $y_3, y_2 \ (y_0)$ axis.  }
  \label{fig_2ep}
 \end{figure} 
 
A major advantage of investigating the critical case of virus strain $y_i$ invading a known equilibria structure (here the nested network) is that new equilibria can be obtained by application of Proposition \ref{prop1new}, the circuit coefficients, and known equilibria values.  The line of equilibria (virus and immune components denoted by $\mathbf y$ and $\mathbf z$) remain positive in some neighborhood around the bifurcation parameter set where the circuit linear form, $\mathcal A_i$, is zero.  Indeed, the values of $\mathbf z$ remain constant throughout the line of equilibria for $\mathcal A_i=0$, so the positive components in the boundary nested equilibrium carry over to the boundary equilibrium of the new invasion equilibrium.
For the simple case of $n=2$ described above, the loss of stability of $\widetilde{\mathcal E}_2$ when $\mathcal A=0$ results in strain $y_3 \ (01)$ replacing either $y_0 \ (00)$  if $\widetilde{y}_{0}<\widetilde{y}_{2}$, or $y_2 \ (11)$ if $\widetilde{y}_{0}>\widetilde{y}_{2}$ (displayed in Fig. \ref{3b}).  By \eqref{equilib1}, the strain which is replaced depends upon the sign of $s_1-1+\frac{ \mathcal Q_{2}}{\mathcal R_{2}}$.
In the case of $n=3$, there are $2^3-3-1=4$ circuits corresponding to a non-nested invading strain.   
Explicitly the circuits, characterized by the corresponding linear form (with the non-nested strain term appearing first), are as follows: (i) $\mathcal A_4=-\mathcal R_{010}+\mathcal R_{000}-\mathcal R_{100} + \mathcal R_{110}$, (ii) $\mathcal A_5=-\mathcal R_{001}+\mathcal R_{000}-\mathcal R_{110} + \mathcal R_{111}$, (iii) $\mathcal A_6=-\mathcal R_{101}+\mathcal R_{100}-\mathcal R_{110}+\mathcal R_{111}$, (iv) $\mathcal A_7=-\mathcal R_{011}+\mathcal R_{000}-\mathcal R_{100} +\mathcal R_{111}$.  Thus, Theorem \ref{newthm1} implies the nested equilibrium is stable if and only if all of the quantities (i)-(iv) are positive.  Furthermore, in each case that a single inequality fails, the following bifurcation occurs where the missing strain replaces a nested strain $y_j$ where $j$ is determined by $ (i) \ \arg\min_{j=0,2} (\widetilde y_j), (ii) \ \arg\min_{j=0,3} (\widetilde y_j) , (iii) \ \arg\min_{j=1,3} (\widetilde y_j) ,(iv) \ \arg\min_{j=0,3} (\widetilde y_j)$, where $\widetilde y_j$ are defined in terms of viral and immune response fitness quantities in \eqref{equilib1}.  For example, if a bifurcation from nested equilibrium  $\widetilde{\mathcal E}_3$ occurs through inequality (iii) switching sign, then $y_6 \ (101)$ replaces either $y_1$ or $y_3$, depending on whether $\widetilde y_1<\widetilde y_3$, i.e. $s_1<1-\frac{ \mathcal Q_{3}}{\mathcal R_{3}}$.  In the case this inequality holds and $y_1$ is replaced, the new stable equilibrium will consist of persistent strain (sequence) set $\left\{ 101,000,110,111\right\}$.  For $n=4$, there are 11 circuits determining stability of nested network, 10 of which consist of 4 strains (ones' complement circuits) and one that has 6 strains in the circuit, $\mathcal A_{0101}:=-\mathcal R_{0101}+\mathcal R_{0000}-\mathcal R_{1000}+\mathcal R_{1100}-\mathcal R_{1110}+\mathcal R_{1111}$.    Thus in the case of invasion of the nested equilibrium $\widetilde{\mathcal E}_4$ by strain $0101$, there are 3 possible strain replacements and (in terms of integer indexing) $\arg\min_{j=0,2,4} (\widetilde y_j)$ determines which nested strain is replaced.  
 
 We can expand upon our observation of the importance of immunodominance in determining viral evolution.  We notice that in any of the invasion scenarios, a viral strain containing minimal sequential mutations to the most immunodominant responses will remain in the equilibrium, no matter the fitness costs.  For $n=2$, we had observed that $y_1 \ (strain \ 10)$ always persists.  For $n=3$, the only invasion scenario where $y_1$ does not persist can be the case of $101$ invasion with invasion equilibrium consisting of strain sequences $\left\{ 000,110,101,111\right\} $.  For the nested equilibrium with $n+1$ strains, $\widetilde{\mathcal E}_n$, replacement of the immunodominant resistant strain $y_1$ only can occur with invasion by a non-nested strain with resistance at the first epitope (sequence of form $10\dots$ with at least two ``$1$'' alleles), so that all strains will have at least 2 mutations.

\subsection{One-to-one network determined by epistasis}

Now we turn to another possible persistent equilibrium assemblage of virus and immune response variants; the one-to-one (or strain-specific) network.  Consider the viral strains that have gained resistance to $n$ or $n-1$ immune response, forming a subsystem of \eqref{odeS} with the $m=n+1$ strains containing more than $n-1$ mutations ($n-1$ ones in binary sequence).  For convenience, we index the strains according to the position of the susceptible epitope (zero in binary sequence), so that in more general equations \eqref{ode3},  $y_i, \ i=1,\dots, n+1$ has epitope set $\Lambda_i=\left\{i\right\}$ or $\Lambda_{n+1}=\emptyset$ and $A$ is a $n+1\times n$ matrix comprised of the diagonal matrix ${\rm diag}\left(a_1,\dots, a_n\right)$ and a row of zeros. This subsystem of a ``one-to-one'' interaction network, where each immune response population attacks a unique specific viral strain, has been considered in \cite{wolkowicz1989successful,korytowski2015nested,bobko2015singularly}.  Stability and persistence results, analogous to \cite{browne2016global} for the nested subsystem, were proved in \cite{wolkowicz1989successful} for the one-to-one network under the assumption of decreasing reproduction numbers $\mathcal R_i<\mathcal R_{i+1}, \ i=1,\dots,n$.  In this case, for $k\in[0,n]$, the relevant strain-specific equilibria are $\mathcal E^{\ddagger}_{k+1}= (x^{\ddagger},y^{\ddagger},z^{\ddagger}), \mathcal E^{\dagger}_{k}   = (x^{\dagger},y^{\dagger},z^{\dagger})$, where:
 \begin{align}
x^{\ddagger}&=\frac{1}{\mathcal R_{k+1}}, \quad y_i^{\ddagger}=s_i, \quad z_i^{\ddagger}=\frac{\mathcal R_i}{\mathcal R_{k+1}}-1, \quad  i=1,\dots, k, \quad y_{k+1}^{\ddagger}=1-\frac{\mathcal P_k}{\mathcal R_{k+1}}, \label{SSeq1} \\ z_{k+1}^{\ddagger}&= 0, \quad y_i^{\ddagger}= z_i^{\ddagger}= 0, \quad k+1<i\leq n, \quad \text{with}   \  \mathcal P_{k}= \mathcal P_{k-1}+ s_{k}\mathcal R_{k}, \quad  \mathcal P_0=1, s_k=1/\mathcal I_k, \notag \\
x^{\dagger}&=\frac{1}{\mathcal P_k},  \quad y_i^{\dagger}=s_i,  \quad z_i^{\dagger}=\frac{\mathcal R_i}{\mathcal P_k}-1, \quad  i=1,\dots, k,  \quad y_i^{\dagger}= z_i^{\dagger}= 0, \quad k+1\leq i\leq n. \label{SSeq2} 
\end{align}

If the assumption of strictly decreasing reproduction numbers is relaxed, then the strain-specific subsystem can have multiple degenerate saturated equilibria.  However, the full hypercube network for $n$ epitopes containing $2^n$ virus strains (model \eqref{odeS}) allows us to relax this particular assumption on reproduction numbers.  
 Indeed, we previously proved that the only \emph{strain-specific equilibria} (with persistent strains contained in one-to-one network) which can be stable in the full hypercube network \eqref{odeS} are equilibria with persistent strains $y_1,\dots,y_n$ ($\mathcal E^{\dagger}_n$), and  with persistent strains $y_1,\dots,y_{n+1}$ ($\mathcal E^{\ddagger}_{n+1}$) \cite{browne2018dynamics}.  Here we expand upon these results by showing, analogous to the nested network, the stability of the one-to-one network is determined by $2^n-n-1$ circuits corresponding to potential invading strains as proved in the following theorem.
\begin{theorem}\label{ssF}
Consider system \eqref{odeS} on the full network with $n$ epitopes ($m=2^n$ virus strains) and fitness costs \eqref{fitnesscost}.  Suppose the viral strains, $y_i \ i=0,\dots, 2^n-1$, are ordered so that $\Lambda_j=\left\{j\right\}$ for $j=1,\dots,n$ and $\Lambda_{n+1}=\emptyset$ (where $\Lambda_j$ denotes strain $j$ epitope set \eqref{assumeaj}).  If $\mathcal E^{\dagger}_n$ or $\mathcal E^{\ddagger}_{n+1}$  is positive, then 
$\mathcal E^{\dagger}_n$ (if  $\mathcal R_{n+1}\leq \mathcal P_n$) or $\mathcal E^{\ddagger}_{n+1}$ (if  $\mathcal R_{n+1} > \mathcal P_n$)
is stable if and only if $\mathcal A_{\mathbf i}>0$, where $i=0,n+2,\dots 2^n-1$, and linear forms $\mathcal A_i$  correspond to invasion circuits $\mathcal C_{\mathbf i}$, as characterized below:
\begin{align}
\mathcal C_{\mathbf i}&=y_{\mathbf i} \cup \left\{ y_j \right\}_{i_j=1}, \qquad \mathcal A_{\mathbf i}= -\mathcal R_{\mathbf i} -\left( |\Lambda_{i}| - 1 \right) \mathcal R_{n+1} + \sum_{j\notin\Lambda_i} \mathcal R_{j} \qquad (j\in [1,n]). \label{ss_circuit}
\end{align}
Furthermore $y_1,\dots,y_n$ are persistent strains ($y_{n+1}$ also if $\mathcal R_{n+1} > \mathcal P_n$) and this the only scenario where strain-specific equilibria, \eqref{SSeq1}- \eqref{SSeq2}, can be stable in the full model.
\end{theorem}

Note the proof of this theorem is in Appendix, and here we make a few remarks to interpret the result.  First, observe that the reproduction number $\mathcal R_i$ of a potential invading strain $y_i$, depends on its epitope set $\Lambda_i$.  Because each mutation comes with a fitness cost \eqref{fitnesscost}, $\mathcal R_i$ roughly correlates with number of susceptible (non-mutated) epitopes, $|\Lambda_{i}|$, and thus both negative terms and the positive summation in \eqref{ss_circuit} increase with $|\Lambda_{i}|$.  Therefore, there is no general rule for determining the sign of invading strain circuits corresponding to the one-to-one network, each depending on the relevant combinations of fitness costs, i.e. epistasis.  We can discuss possible strain replacements for invasion of $\mathcal E^{\ddagger}_{n+1}$ as before.  In this case, we find that the replaced strain is $y_{\min( j\in \Lambda_i)}$, i.e. the strain susceptible to strongest immune response among the susceptible epitopes of strain $\mathbf i$, since this strain has lowest value in equilibrium corresponding to positive coefficient in circuit.  Compared to the $n+1$ strain nested network ($\widetilde{\mathcal E}_n$), the ``invasion circuit'' and strain replacement of the $n+1$ strain one-to-one network ($\mathcal E^{\ddagger}_{n+1}$) is simpler to determine.  Note that invasion of the $n$ strain $\mathcal E^{\dagger}_{n}$ can result in addition of the new strain rather than replacement, and the critical case does not correspond to a line of equilibria as with the $n+1$ strain equilibria.  As an example of circuit linear forms \eqref{ss_circuit} for stability, consider the case $n=3$, where strains with 1-mutation have : $\mathcal A_{100} = -\mathcal R_{100} - \mathcal R_{111} +  \mathcal R_{101}+ \mathcal R_{110}$, $\mathcal A_{010} = -\mathcal R_{010} - \mathcal R_{111} +  \mathcal R_{011}+ \mathcal R_{110}$, $\mathcal A_{001} = -\mathcal R_{001} - \mathcal R_{111} +  \mathcal R_{011}+\mathcal R_{101}$.  Each corresponds to an embedded 2-cube measuring marginal epistasis with their 1 mutation fixed.  Note that $\mathcal A_{100}=-\mathcal A_6$, where $\mathcal A_6$ also is the circuit corresponding to invasion of nested equilibrium by $(101)$.   Now consider potential invasion by the wild-type strain ($000$) given by $\mathcal A_{000}= -\mathcal R_{000}-2\mathcal R_{111}+\mathcal R_{011}+\mathcal R_{101}+\mathcal R_{110}$, which biologically tells us whether the two-mutation associations predict
the three-mutation combination.  Of note, the sign of this circuit does not have a two-locus interpretation,
making them truly of higher-order \cite{gould2018microbiome}.

\subsection{Other equilibrium network structures and open questions}
The full utility of the circuit analysis comes with bifurcations of equilibria with $n+1$ strains, as our above examples illustrate, because the critical state corresponds to persistent strains forming a circuit in Proposition \ref{prop1new}.  How far can we go with this analysis?  Can we generalize to all equilibrium structures?  Observe from the proofs of Theorem \ref{newthm1} and Theorem \ref{ssF} that the two equilibrium networks considered, nested and one-to-one, with $n+1$ strains ($\widetilde{\mathcal E}_{n}$ and $\mathcal E^{\ddagger}_{n+1}$) form a basis of $\mathbb R^{n+1}$ when the strains are considered as binary sequences with a one addended at the the end of sequences, and moreover every binary sequence has integer coordinates with respect to this basis.  This directly leads to the ``invasion circuits'', and this is generalized to any assemblage of $n+1$ strain sequences in the following proposition (proof in Appendix):  
\begin{proposition}\label{genprop}
Suppose $\mathcal S\subset \left\{0,1\right\}^n$ is the set of binary sequences of an equilibrium, $\mathcal E^*$, with $n+1$ strains ($|\Lambda_y|=n+1$).  Assume that $\mathcal S\times \left\{1\right\}$ is a basis of $\mathbb R^{n+1}$ and any addended binary sequence $\mathbf i1\in \left\{0,1\right\}^n\times \left\{1\right\}$ has integer coordinates with respect to this basis.  Then for all $\mathbf i \in \left\{0,1\right\}^n \setminus \mathcal S$, $\mathcal C= \left\{\mathbf i\right\}\cup \mathcal S$ forms a circuit where a linear form $\mathcal A_{\mathbf i}$ is given by the coordinates of $\mathbf i$ with respect to $\mathcal S\times \left\{1\right\}$.  Furthermore, the stability of $\mathcal E^*$ is determined by the sign of $\mathcal A_{\mathbf i}$. 
\end{proposition} 
  \begin{figure}[t!]
\subfigure[][]{\label{2a_n} \includegraphics[width=.51\textwidth,height=.25\textwidth]{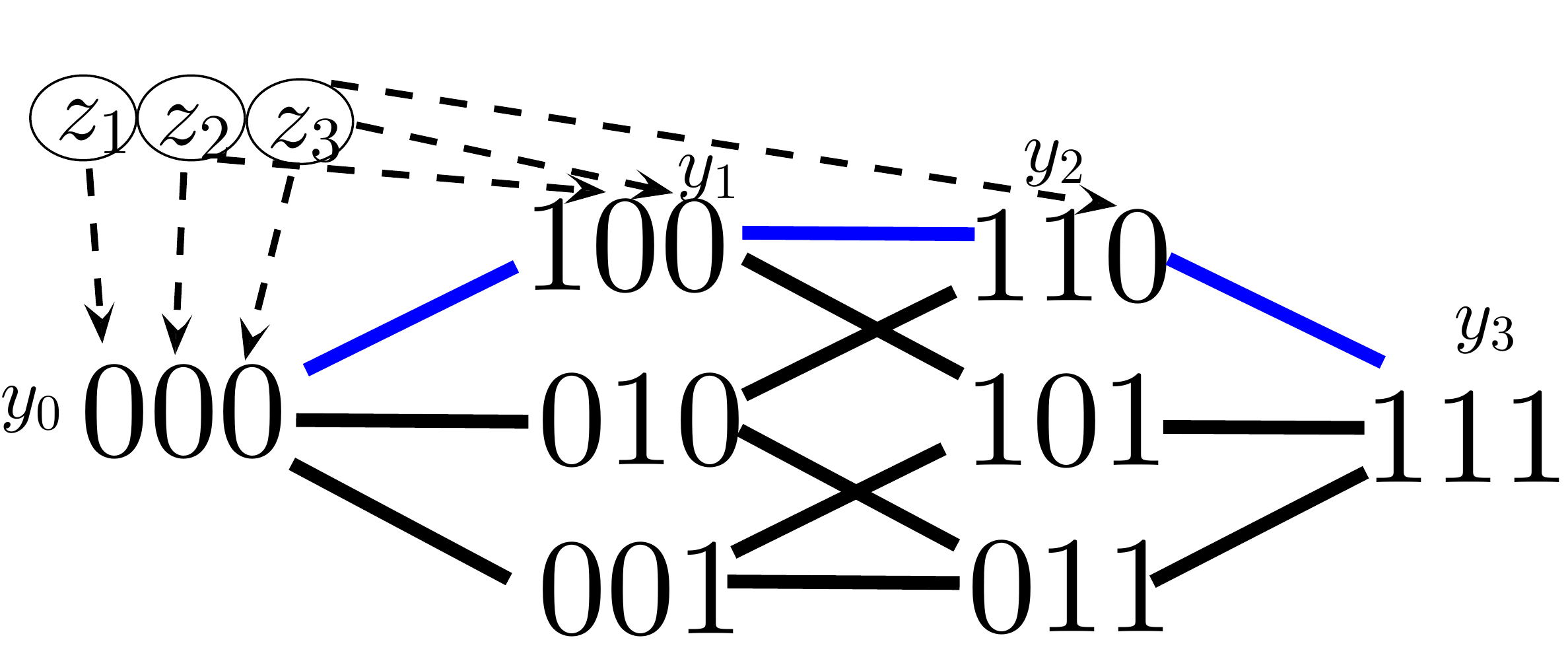}}
\subfigure[][]{ \label{2b_n}\includegraphics[width=.51\textwidth,height=.25\textwidth]{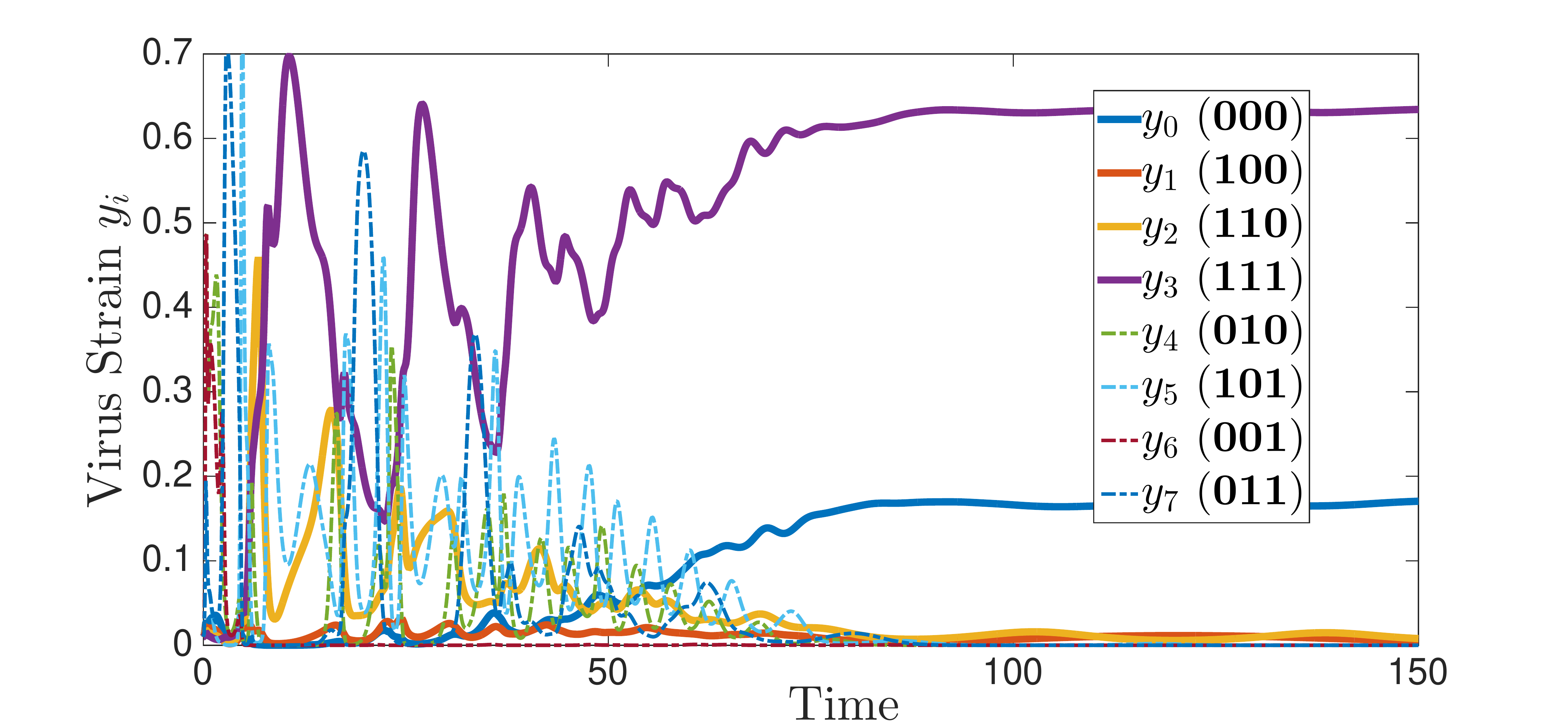}}   
\caption{ \emph{Convergence to nested network assuming multiplicative viral fitness landscape}.  (a) Persistent viral strains is reduced to nested network (in blue) as $t\rightarrow\infty$ in Prop. \ref{propMult}.  (b) Viral strain components $y_0(t),\dots,y_2(t)$ persist as system \eqref{odeS} converges to equilibrium $\widetilde{\mathcal E}_{3}$.  }
  \label{fig_3ep}
 \end{figure} 
Now consider the scenario that strain $\mathbf k\in \mathcal S$ is replaced by $\mathbf i$, then the new equilibrium sequences $\mathcal S'=\left\{\mathbf i\right\}\cup \mathcal S\setminus \left\{\mathbf k\right\}$ forms a basis of $\mathbb R^{n+1}$ since any proper subset of a circuit is linearly independent.  Thus the strain replacement with invader $\mathbf i$ will result in this new equilibrium structure $\mathcal S'$ also forming a circuit \emph{if any sequence has integer coordinates with respect to} $\mathcal S'\times \left\{1\right\}$.   In this fashion, we might observe a sequence of strain invasions determined by circuits.  Notice that strain invasions of the two $n+1$ strain equilibria structures explored here, nested and one-to-one networks, would result in a strain replacement whose new equilibrium has stability determined by linear form on circuit.  Indeed, because the coordinate of any potential invader $\mathbf i$ was shown to be $+1$ corresponding to the strain it can replace, it is not hard to show that the new basis will also yield integer coordinates for any other sequence.  Once we move past this initial invasion though, it would not be clear if the circuit stability pattern continues though.

Another consideration is whether a strain can be added to an $n$ strain equilibrium (where $n$ is number of persistent immune responses) in order to have a positive $n+1$ strain equilibrium which satisfies Proposition \ref{genprop}, i.e. forms a set $\mathcal S$ corresponding to a basis with integer coordinates in the extended $n+1$ dimensional binary sequence space.  In our examples, we add the completely resistant strain (with sequence $\mathbf 1$) to the $n$ strain nested or one-to-one networks (with $n$ persistent immune responses) to get an $n+1$ strain equilibrium satisfying the hypotheses of Proposition \ref{genprop}.  In general, this might not always be the case.  First, we recall that determining the feasibility of a $n+1$ strain positive equilibrium is dependent on calculation of $C=\left( A'  \  \vec{\mathcal R'} \right)^T$ by Proposition \ref{prop33}, with $A'$ as the virus-immune interaction network of the $n+1$ strains where the rows of $A'$ correspond to the complements ($\mathbf 1- \mathbf i$) of the viral sequences in $\mathcal S$.  If there is a feasible $n$ strain equilibrium with network $A$ and reproduction numbers $\vec R$, then the complete resistance strain $\mathbf 1$ can be added if $\mathcal R_{\mathbf 1}>1+ \mathbf \rho A^{-1} \vec R$.    However,  the calculation for adding other strain sequences is more complicated, thus the problem of both determining feasibility and whether an equilibrium satisfies Proposition \ref{genprop} may be difficult.

 As an example, consider another possible equilibrium type, the $n$ strain 1-mutation network: $\mathcal S_1=\left\{y_i^1 \ | \  i=1,\dots,n\right\}$ in which $y_i^{1}$ has only escaped $z_i$ so that its binary sequence is $\mathbf i^1=(\delta_{\ell i})_{\ell=1}^n$ where $\delta_{\ell i}$ is Kronecker delta function.  If we add $y_w \ (\mathbf 0)$ to $\mathcal S_1$, then circuits determine stability, however adding $\mathbf 1$ does not yield circuits determining stability (in particular stability condition for invasion by $\mathbf 0$) is not a circuit.  Indeed, we can derive some conditions for positivity of an equilibrium consisting of viral strains $\tilde{\mathcal S}_1=\left\{\mathbf 0\right\}\cup\mathcal S_1$ (see Appendix \ref{A4}).  Consider the case $n=3$, where the circuit for invasion of $\tilde{\mathcal S}_1$ by strain $\mathbf i=\mathbf 1$ can be calculated according to coordinate basis description in extended sequence space:  $$ \begin{pmatrix} 1 \\ 1\\ 1\\ 1 \end{pmatrix}=\begin{pmatrix} 1 & 0 & 0 &0 \\ 0 & 1&0&0\\ 0&0&1&0\\ 1&1&1&1 \end{pmatrix} \mathbf a \Rightarrow  \quad \mathcal A=-\mathcal R_{\mathbf i} -\sum_{\mathbf k\in\mathcal S}a_{\mathbf k}\mathcal R_{\mathbf k}=-\mathcal R_{111}-2\mathcal R_{000}+\mathcal R_{100}+\mathcal R_{010}+\mathcal R_{001} $$
Similar, to the example circuit given in the one-to-one network, this measures higher-order epistasis, in particular whether the one-mutation associations predict
the three-mutation combination.  Here, the strain replacement would be $111$ replacing $001$ because this sequence would have the smallest equilibrium value of positive coefficient strains in $\tilde{\mathcal S}_1$.  It can be shown the other invasion circuits correspond to conditional epistasis (embedded 2-cubes), where the single non-mutated epitope of the invader remains fixed.     Indeed, using the coordinate basis method above, we have the following proposition for invasion of the ``$\leq1$ mutation'' network:
\begin{proposition}\label{omprop}
Consider the $\leq 1$ mutation network, $\tilde{\mathcal S}_1$, consisting of wild-type and 1-mutation viral strains $y_0,y_1,\dots,y_n$ where the sequence of $y_j$ is $\mathbf j=\left(\delta_{\ell j}\right)_{\ell=1}^{n}$ for $j=1,\dots,n$.  Suppose that there is a positive equilibrium, $\widehat{\mathcal E}_n$, with $\tilde{\mathcal S}_1$ as persistent viral strain set $\Omega_y$.  $\widehat{\mathcal E}_n$ is stable if and only if 
$\mathcal A_i>0$, where $i=n+1,\dots 2^n-1$, and linear forms $\mathcal A_i$  correspond to invasion circuits $\mathcal C_i$, as characterized below:
\begin{align}
\mathcal C_{\mathbf i}&=y_{\mathbf i} \cup y_0 \cup \left\{ y_j \right\}_{i_j=1}, \qquad \mathcal A_{\mathbf i}= -\mathcal R_{\mathbf i} -\left(n- |\Lambda_{i}| - 1 \right) \mathcal R_{0} + \sum_{j\notin\Lambda_i} \mathcal R_{j} \qquad (j\in [1,n]). \label{om_circuit}
\end{align}
\end{proposition}

Observe that for the case of $n=3$, we have now highlighted all the circuits determining stability of three equilibria structures: the nested, one-to-one, and one-mutation network.  While there are 4 corresponding linear forms for each network dictating invasion by each missing strain, together this results in 10 distinct circuits since $\mathcal C=\left\{000,100,110,010\right\}$ and $\mathcal C=\left\{100,110,111,101\right\}$ are invasion circuits that the nested network shares with the one-mutation and one-to-one network, respectively.  There are 20 total circuits for $n=3$ \cite{beerenwinkel2007epistasis}, and we leave it to future work as to whether the any of the other 10 circuits correspond to stability of feasible ``transitional equilibria'' between the highlighted networks.   However, the immunodominance hierarchy will impose an effective fitness ordering on the virus genotypes so that for example the ``reverse nested'' network $\left\{000,001,011,111\right\}$  would never be feasible.  Therefore, some circuits should not correspond to any meaningful bifurcation under the assumptions of our model.   


\subsection{Special cases of fitness landscapes}

While fitness landscapes on the $n$-dimensional hypercube generally yield a multitude of circuits determining bifurcations and stability of equilibria, there are some simple landscapes that can be analyzed.  First, consider the pairwise interaction case as described by equation \eqref{pwfit}, where $\mathcal R_{\mathbf i}=\mathcal R_{\mathbf 0}-\mathbf c \cdot \mathbf i +\sum_{j=1}^n i_j \sum_{k>j} i_k B_{jk}$ for a (strictly) upper triangular matrix $B$.  If the matrix $B$ is positive, then the fitness of any sequence with at least 2 mutations will always be larger than the additive case, whereas if $B$ is negative, the resulting fitness from a pair of mutations is less than expected under additivity.  Thus, in the former case of $B$ positive, synergistic interactions should favor double mutants, while in the latter antagonistic interactions might discourage consecutive mutations.  The exact translation of these informal notions to expected results in our model with sign-definite pairwise interactions is not obvious due to there being a dynamic overall fitness landscape when taking into account immune response (predator) populations and other variables/parameters which might influence the viral escape pathway.  Nevertheless, we prove here that the nested network is generally stable when pairwise loci interaction matrix $B$ is positive, whereas a non-nested network, such as one-to-one or $\leq 1$ mutation network, is stable when $B$ is negative.  

\begin{figure}[t!]
\subfigure[][]{\label{4a} \includegraphics[width=.5\textwidth,height=.25\textwidth]{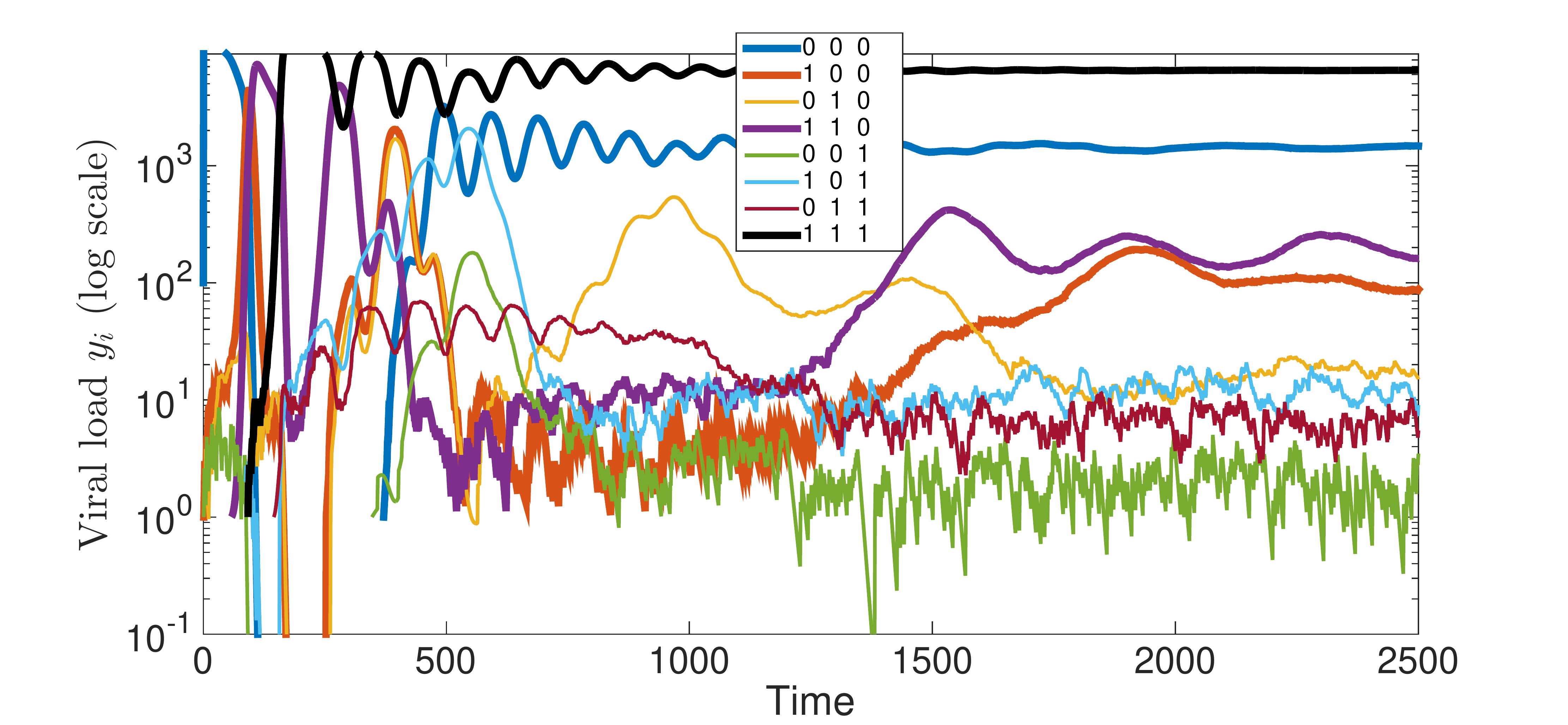}}
\subfigure[][]{ \label{4b}\includegraphics[width=.5\textwidth,height=.25\textwidth]{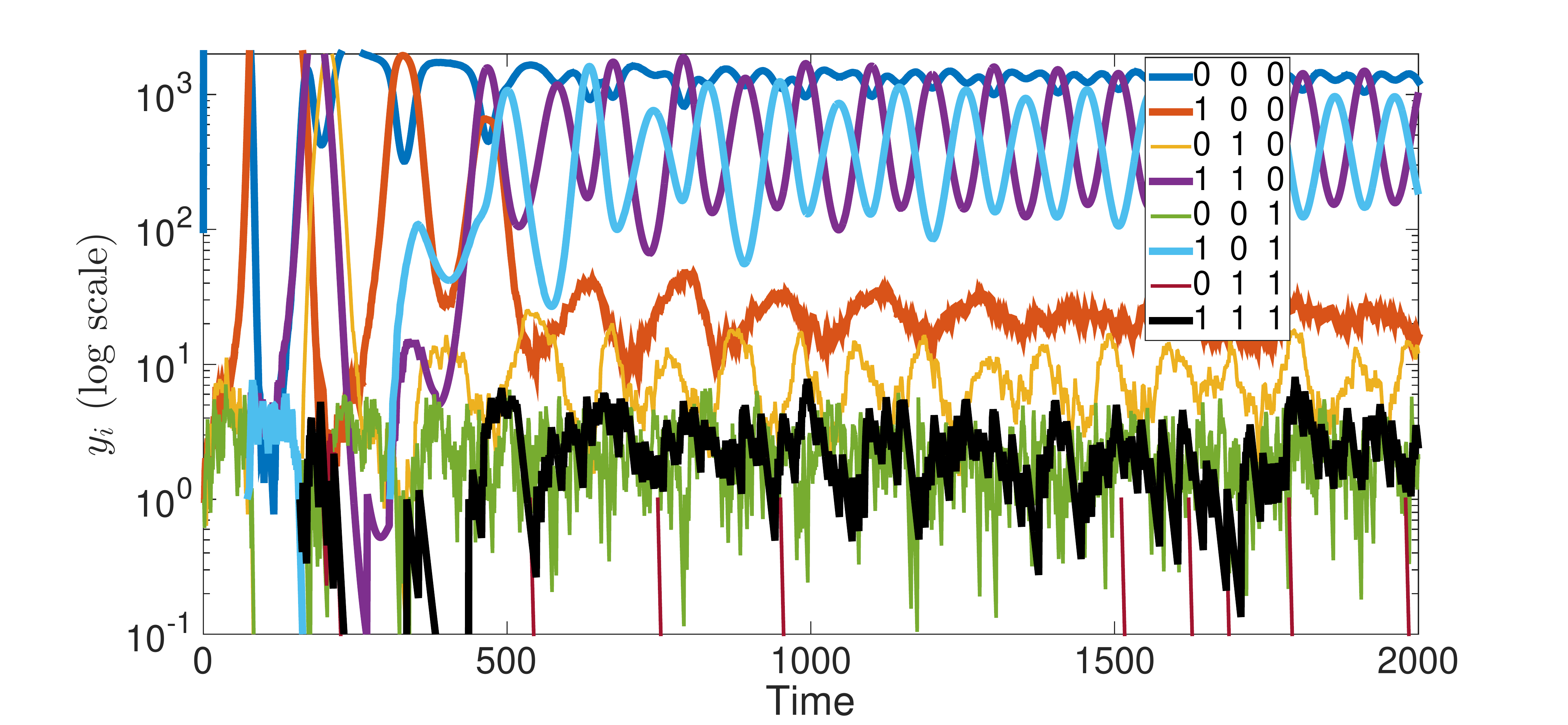}}    \\
\subfigure[][]{\label{4c} \includegraphics[width=.5\textwidth,height=.25\textwidth]{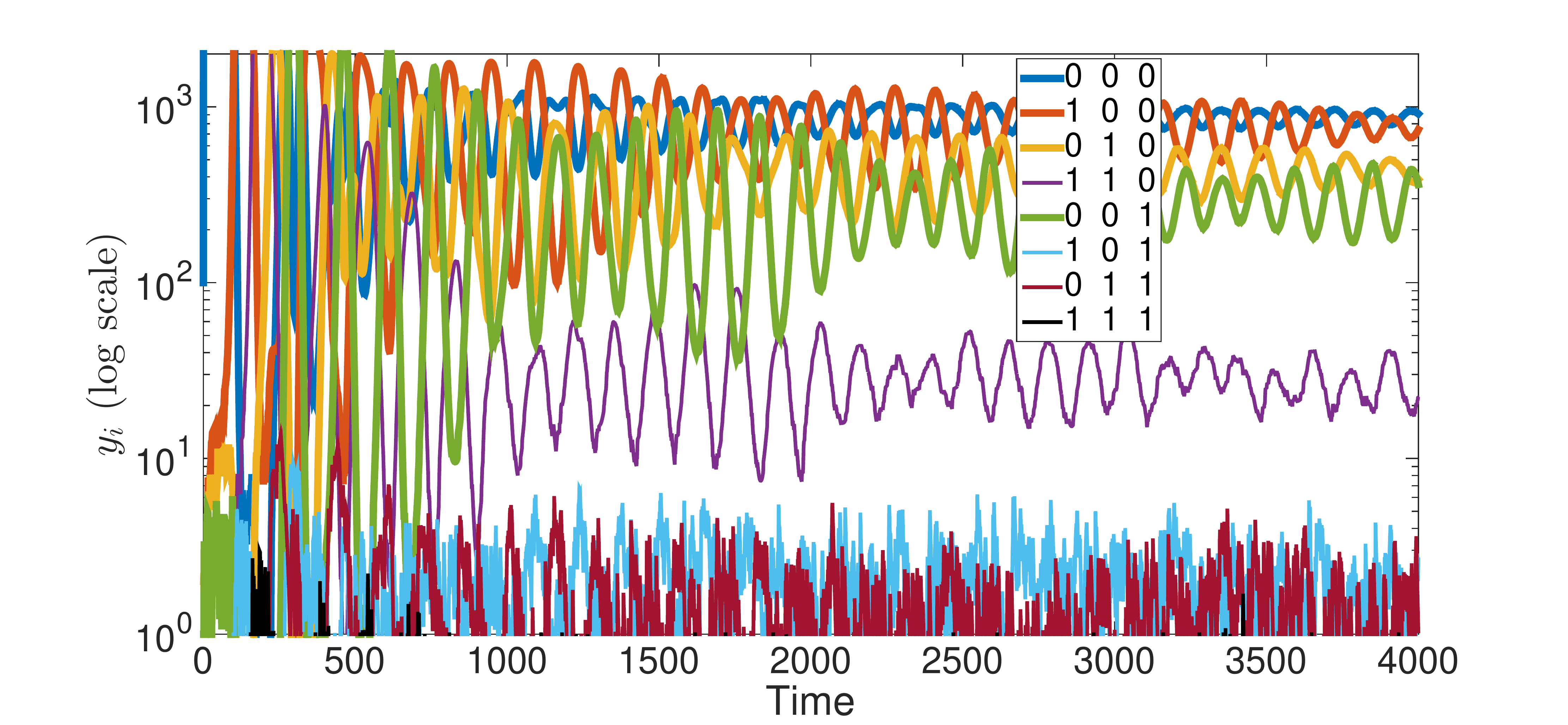}}
\subfigure[][]{ \label{4d}\includegraphics[width=.5\textwidth,height=.25\textwidth]{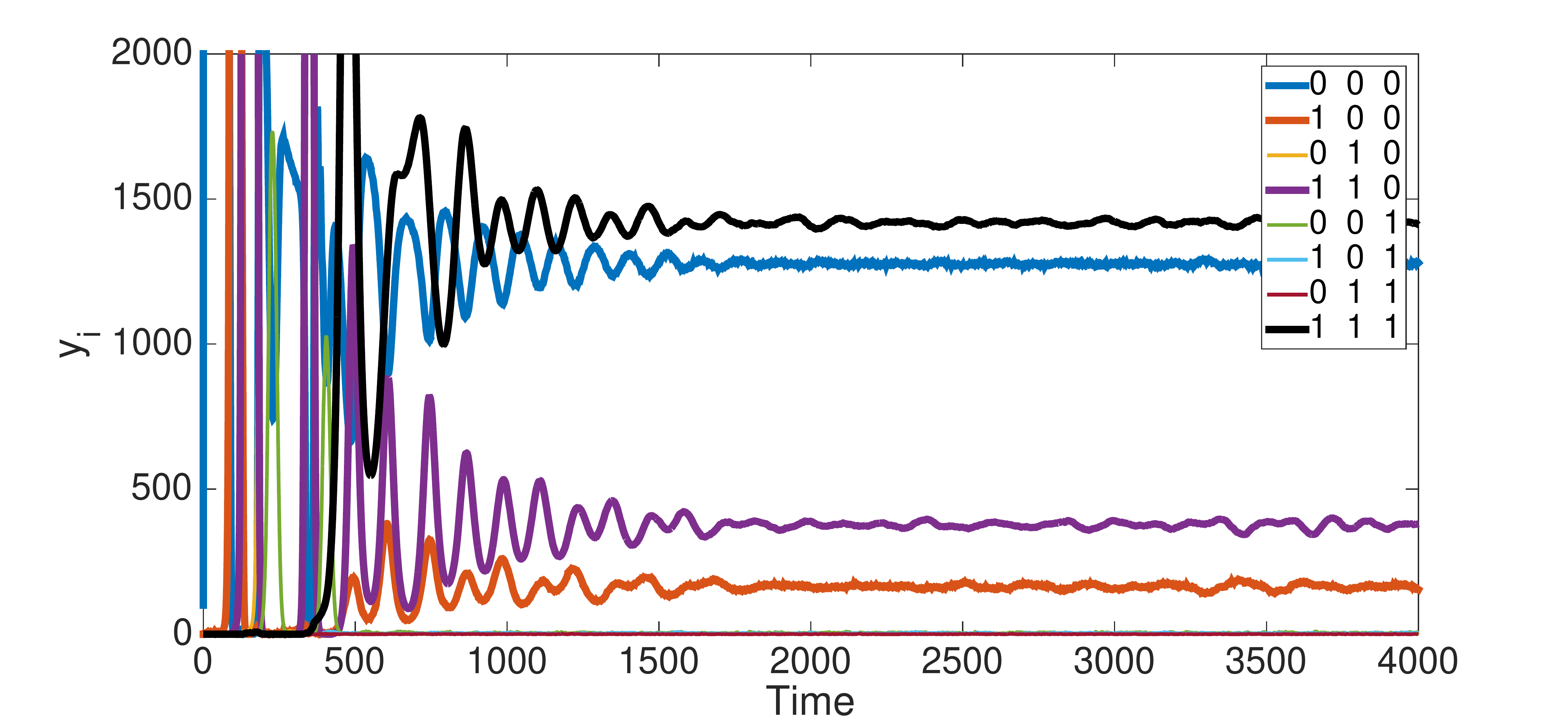}}    \\
\subfigure[][]{\label{4e} \includegraphics[width=.5\textwidth,height=.25\textwidth]{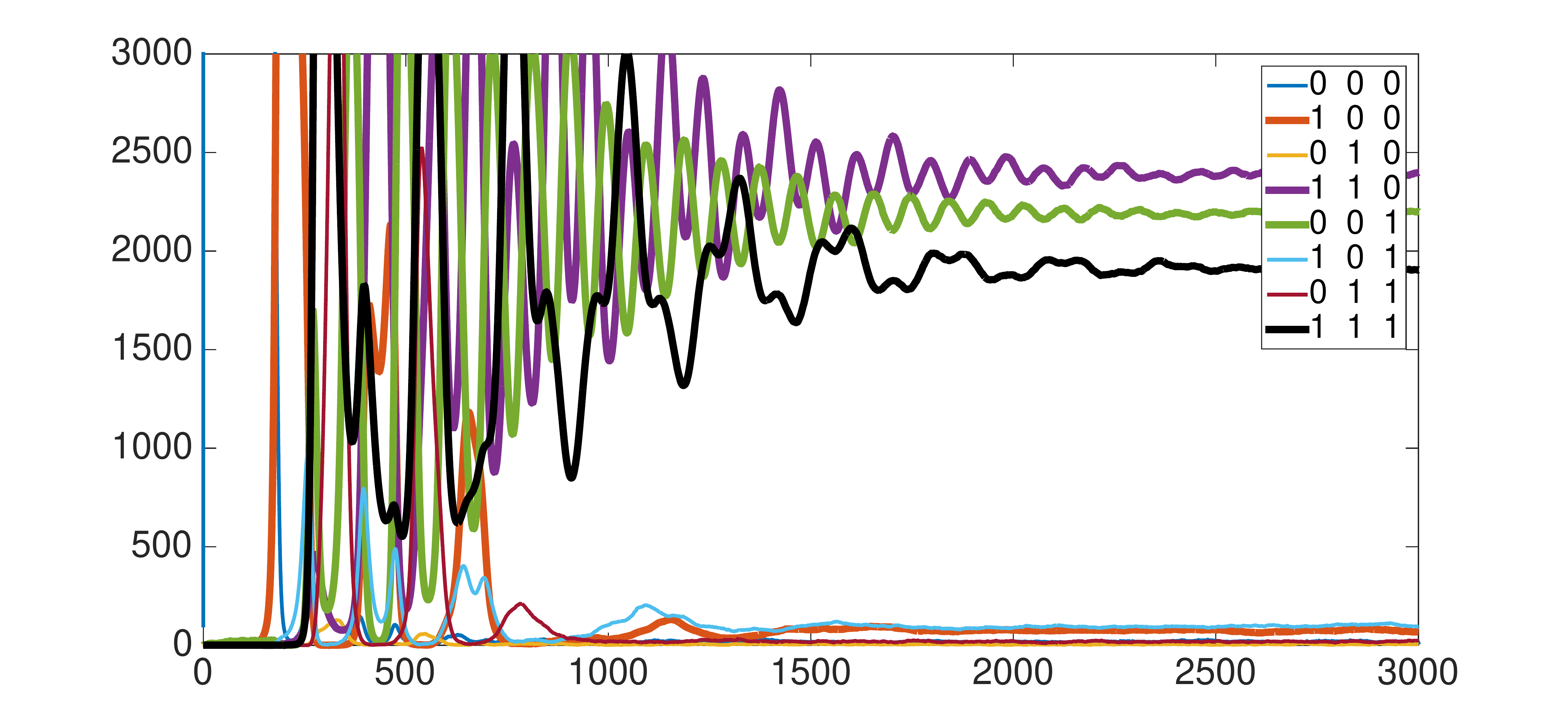}}
\subfigure[][]{ \label{4f}\includegraphics[width=.5\textwidth,height=.25\textwidth]{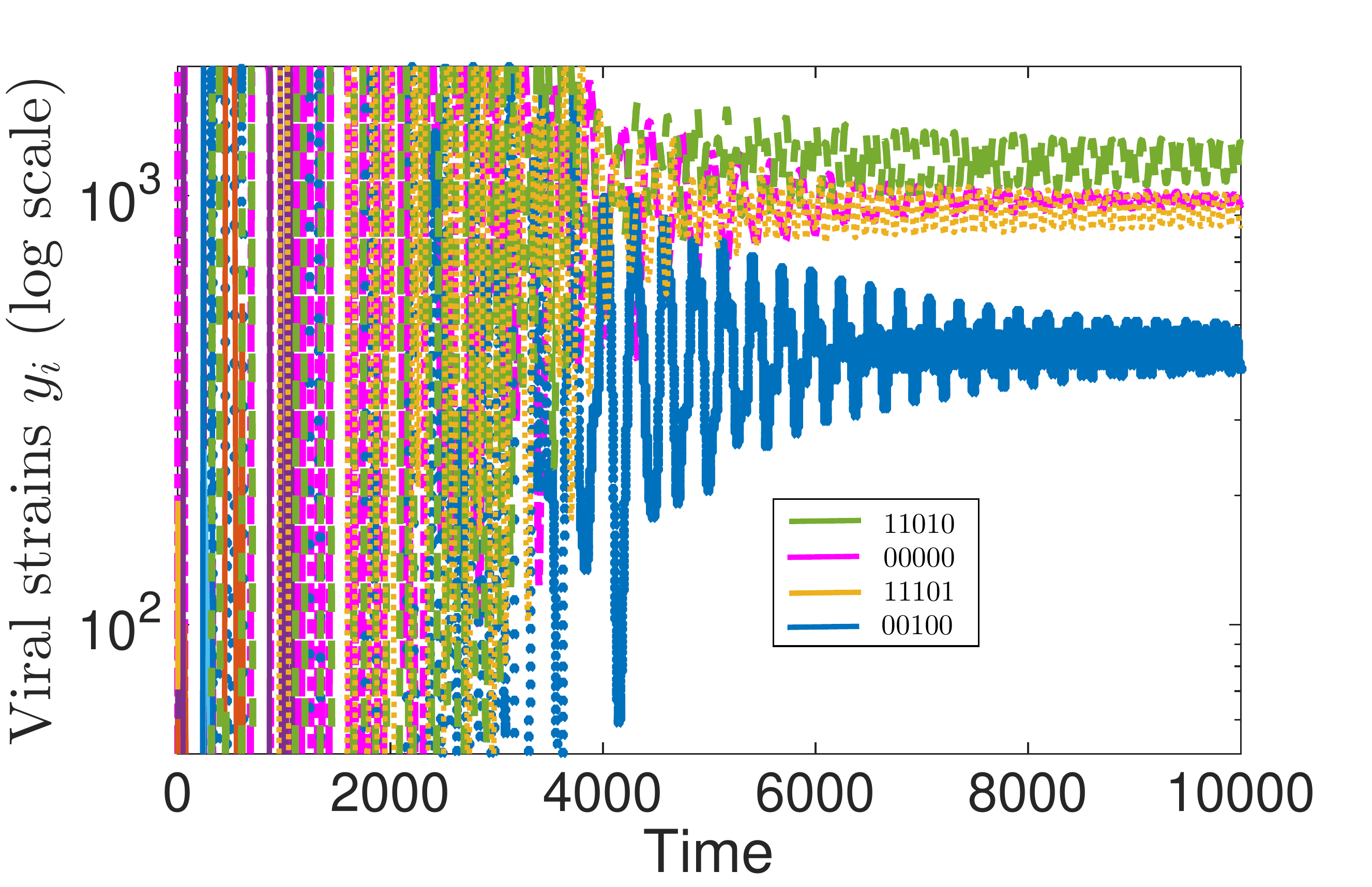}}      \\
\subfigure[][]{\label{4g} \includegraphics[width=.5\textwidth,height=.25\textwidth]{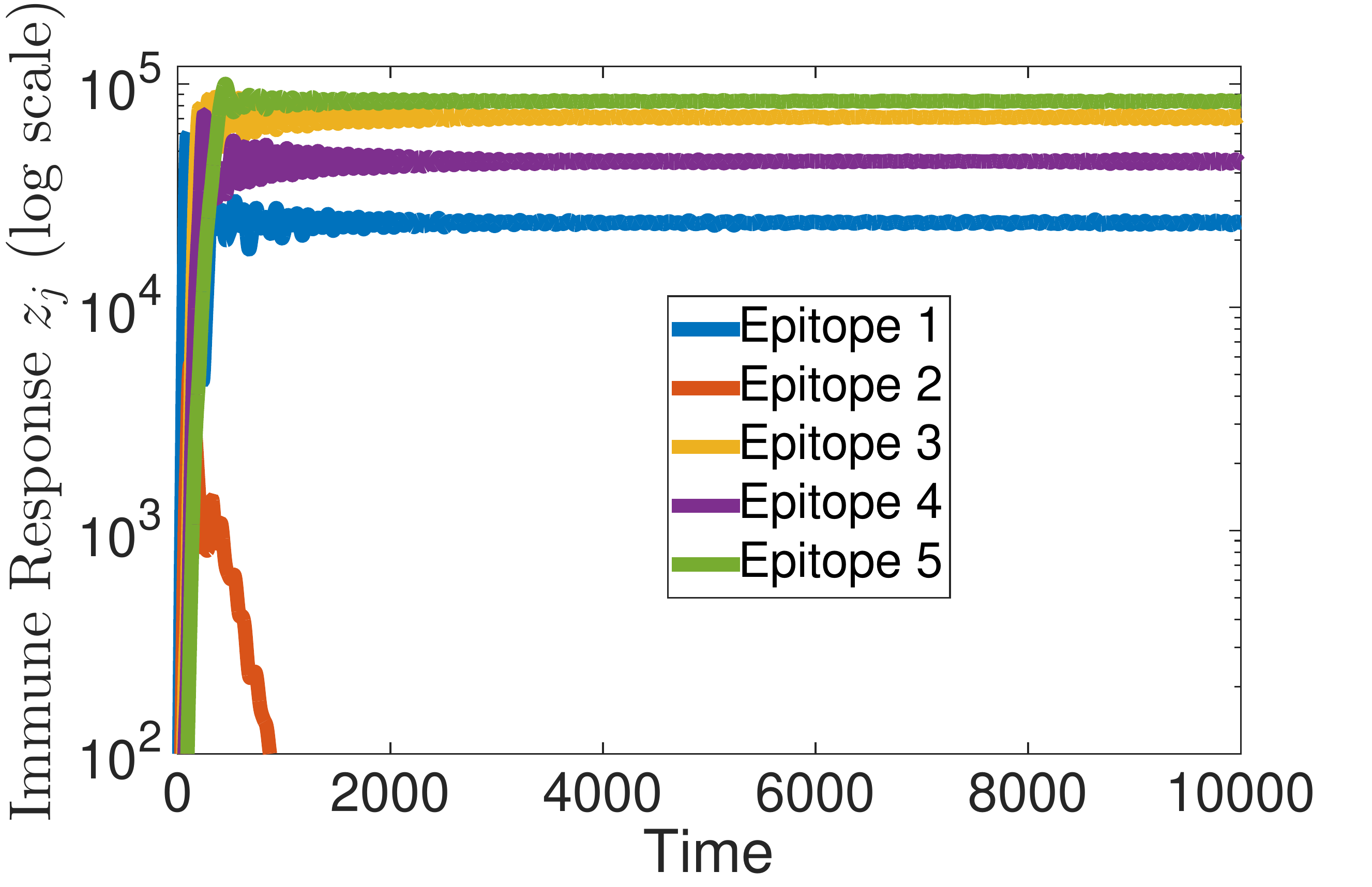}}
\subfigure[][]{ \label{4h}\includegraphics[width=.5\textwidth,height=.25\textwidth]{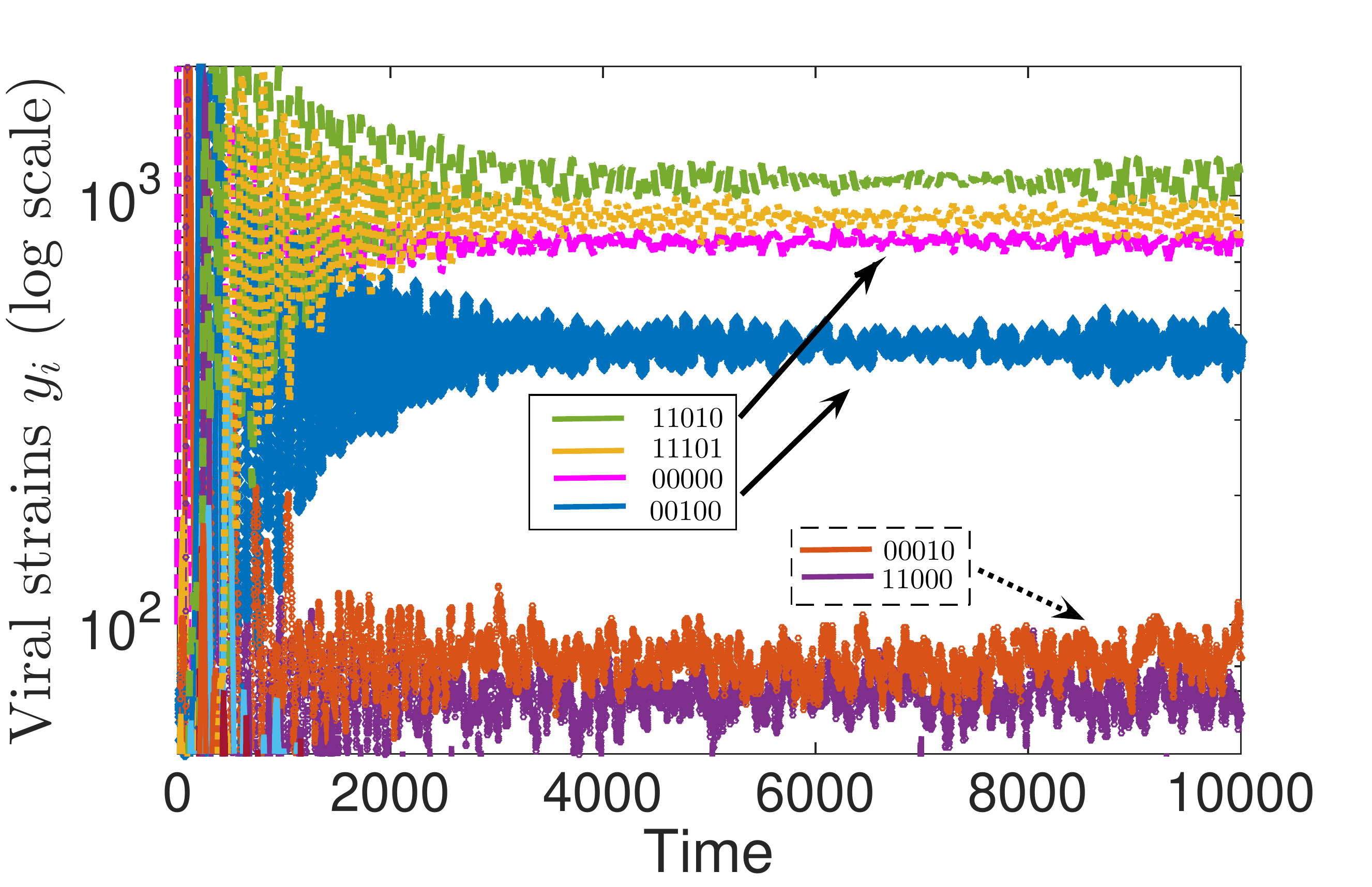}} 
 \caption{{\small \emph{Simulations of extended model with stochastic mutation and pairwise epistatic interactions illustrate eco-evolutionary dynamics consistent with analysis of viral fitness epistasis in deterministic system}.  Trajectories of virus strains in the case of $n=3$ epitopes with uniformly distributed viral and fitness quantities, and (random) \emph{positive} pairwise interactions, $B_{jk}$, which implies positive epistasis with respect to ``nested circuits'' and convergence to nested steady state containing $\left\{000,100,110,111\right\}$.  (b)  Increasing viral strain $y_6 \ (101)$ reproduction number ($\mathcal R_6$) changes the sign of its invasion circuit so that epistasis is no longer positive, resulting in it replacing $y_3 \ (111)$ and non-nested persistent strains. (c) Assuming negative pairwise interactions also leads to non-nested convergence, here to $\leq 1$ mutation network containing $\left\{000,100,010,001\right\}$.  Note that $(110)$ strain persists at low levels due to invasion circuit being close to zero, along with random mutation.  Gaussian distributed pairwise interactions ($B_{jk}$ random sign) result in convergence to nested network in (d) because positive $B_{jk}$ randomly drawn, but generally can converge to other steady states in simulations (e) and (f,g,h) with $n=5$ epitopes.  Observe that the dynamics in original (deterministic) ODE solution displayed in (f) are consistent with stochastic mutation simulations (g,h), except for low level persistence of two strains with small negative invasion rates.  } }
  \label{fig4}
 \end{figure}

   
 \begin{proposition} \label{propPair}
Consider  binary sequence model \eqref{odeS} having pairwise interaction fitness landscape \eqref{pwfit} with upper triangular matrix $B$ that is sign-definite.  Assume that $\mathcal R_0>\mathcal Q_1$ (so that at least one virus strain and immune response persists).  If $B_{jk}>0$ for all $k>j$, then the nested network is stable.  On the other hand, if $B_{jk}<0$ for all $k>j$, then one-to-one network (or $\leq 1$ mutation network) is stable against invasion and persistent if components of associated equilibrium are positive.
\end{proposition}

Another basic example of a fitness landscape is \emph{multiplicative}, where each mutation at a fixed locus reduces the reproduction number of a strain by a fraction regardless of the of sequence background at other loci.  Thus the loci act independently, but not additively.    This \emph{multiplicative} fitness landscape has been assumed in several studies of HIV-immune evolution at multiple epitopes, e.g. \cite{Althaus,vanDeutekom}.  We prove the following proposition, generalizing a theorem in \cite{browne2018dynamics} showing multiplicative equal fitness costs evolve a nested network.
\begin{proposition} \label{propMult}
Assume that fitness costs of mutating locus $j$ come with a multiplicative reproductive loss $f_j$, i.e. $\mathcal R_{\mathbf i}=\mathcal R_0 \prod\limits_{i_j=1}  f_j$  where $0<f_j<1$, $j=1,\dots,n$.  Then the nested network is stable.
\end{proposition}

\section{Simulations \& predicting virus-immune evolution}

 In this section, we conduct simulations of model \eqref{odeS}, along with a hybrid stochastic/deterministic version, in order to illustrate our results.   The model was coded in MATLAB, where the built-in ODE solver ODE45 was utilized for simulations.   For the deterministic model, we find numerical solutions to \eqref{odeS} under the multiplicative viral fitness landscape for $n=3$ epitopes, initiating the simulation with positive concentrations of all virus and variants immune variants, $y_i, \ i=0,\dots,7$ and $z_j, \ j=1,2,3$, where we adopt the nested priority indexing from Section \ref{nestsect}.  The immunodominance hierarchy utilized in the simulation is $\mathcal I_1=6, \mathcal I_1=5.7,\mathcal I_1=5.4$.  We assume each epitope mutation imparts equal independent multiplicative fitness costs, i.e. if $(i_1\dots i_n)$ represents the epitope sequence of strain $i$ and $\mathcal R_i=\mathcal R_0 (1-\kappa)^{i_1+\dots+i_n}$ where   $\mathcal R_0=11.8$ and $\kappa=0.1$ is fitness cost in our simulation.  The scaling factors for viral and immune variant growth rates in \eqref{odeS} are set to:  $\gamma_{i}=3.5, i=0,\dots,3$ and $\gamma_{i}=18.5, i=4,\dots,7$  The corresponding calculations lead to positive epistasis in the invasion circuits of the nested equilibrium $\widetilde{\mathcal E}_3$ (Theorem \ref{newthm1} and Proposition \ref{propMult}) and, as shown in prior work \cite{browne2018dynamics}, result in a sequential nested immune escape trajectory (Fig. \ref{fig_3ep}). 

 An important question concerns if the predicted patterns from our theoretical results on \eqref{odeS} hold when random mutation is included as is in the scenario of HIV infection.  Thus we consider a stochastic extension of the model, along with parameters representative of HIV.  However, since this is a preliminary simulation effort, we choose a rather large viral wild-type (basic) reproduction number $\mathcal R_0$ and low death rate of immune response to better mimic virus-immune evolution for the stochastic model, as in \cite{magalis}.  Similar to the methods in \cite{vanDeutekom}, we simulate mutations of the $n$ loci by drawing from a binomial distribution in a hybrid ODE-stochastic algorithm.   With a mutation rate of $\epsilon=1.67\times 10^{-3}$ per site per day, we compute the number of mutations during replication as follows.  We update mutations at fixed time steps, taken as $\Delta t=1 \ day$, where we approximate the daily number of cells that become de novo infected per viral variant as $M_i = \beta_i XY_i$ cells.   To improve computation speed, we assume that only one of the $n$ loci mutates per replication, i.e. the small probability of simultaneous mutations are neglected.  Then for each viral variant $i=1,\dots,m$ and locus $\ell=1,\dots, n$, the number of mutations is given by ${\rm Bin}(M_i,\epsilon)$.  The viral populations are updated accordingly, and the ODE solver is run for $\Delta t$ time units and then the process repeats.  In the following simulations, we assume that initially there is just the wild-type virus, $y_0(0)>0$, all other strains are absent $y_i(0)=0, \ i=1,\dots, 2^n-1$, and each immune response is present, $z_j(0)>0, \ j=1,\dots, n$. Thus the extended model allows for random mutation and deterministic selection evolving from initial infection by the founder (wild-type) strain.

  \begin{figure}[t!] 
  \subfigure[][]{\label{5a} \includegraphics[width=.5\textwidth,height=.25\textwidth]{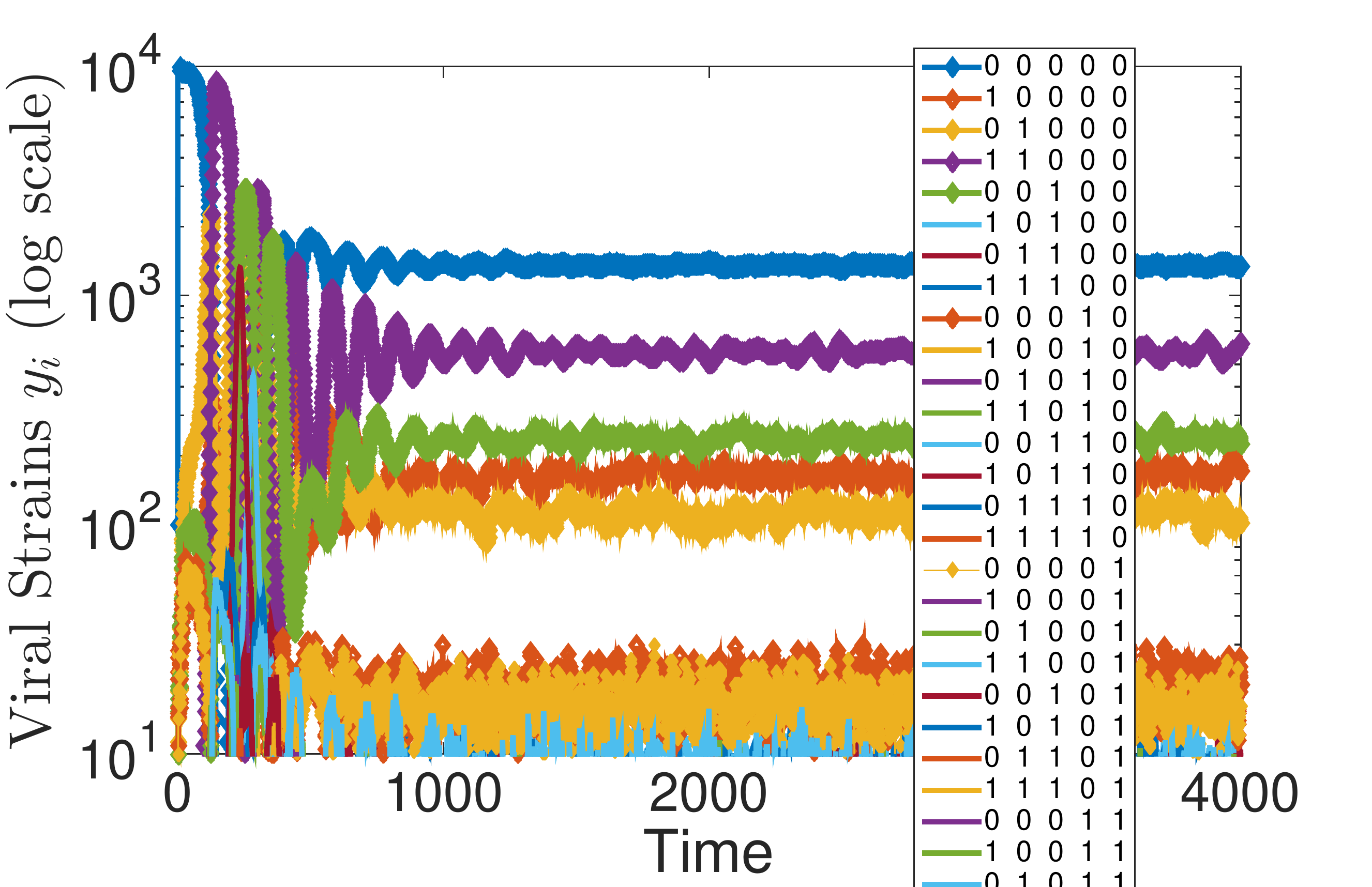}}
     \subfigure[][]{ \label{5b} \includegraphics[width=.5\textwidth,height=.25\textwidth]{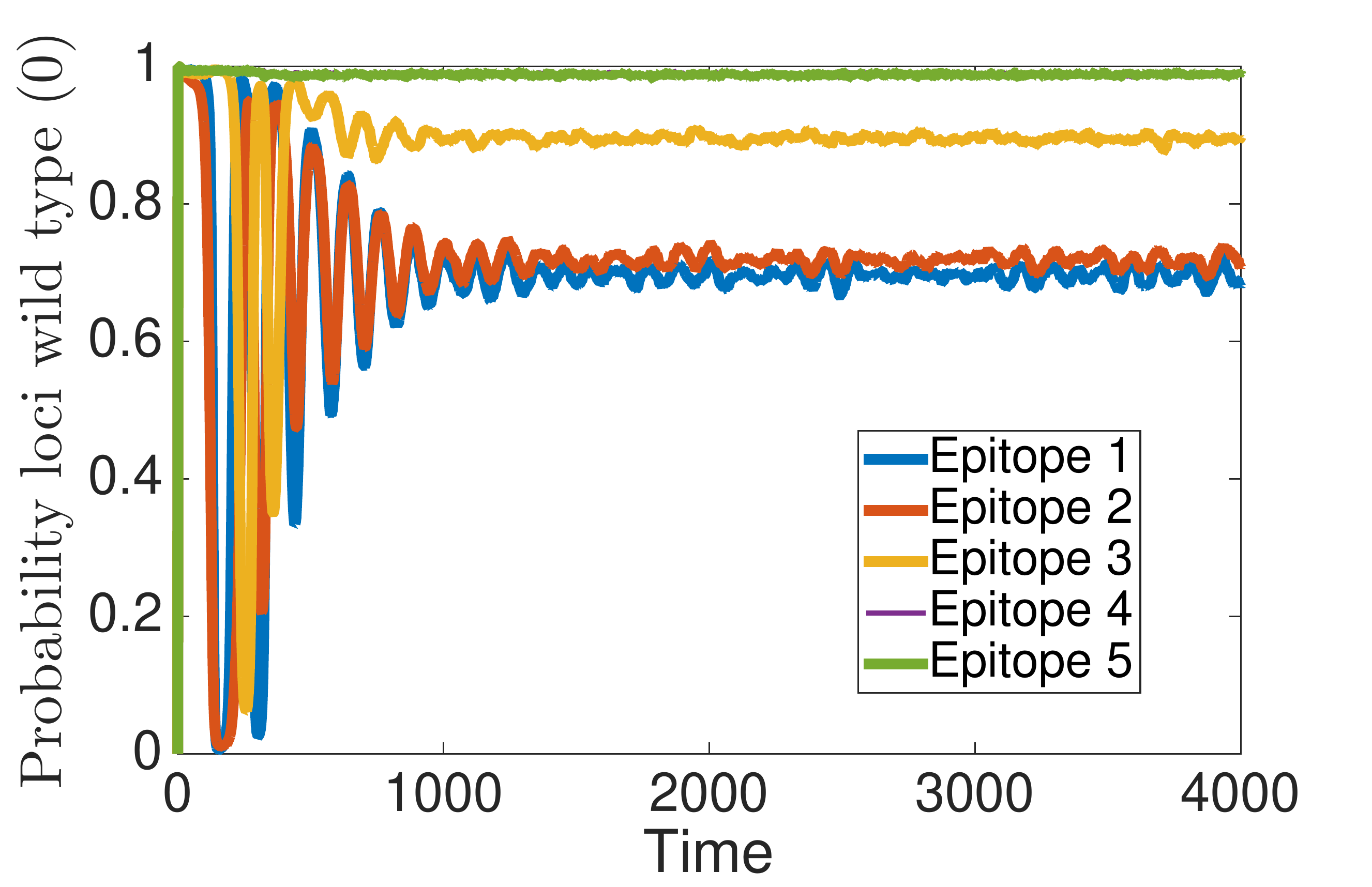}} \\
 \subfigure[][]{ \label{5c} \includegraphics[width=.5\textwidth,height=.25\textwidth]{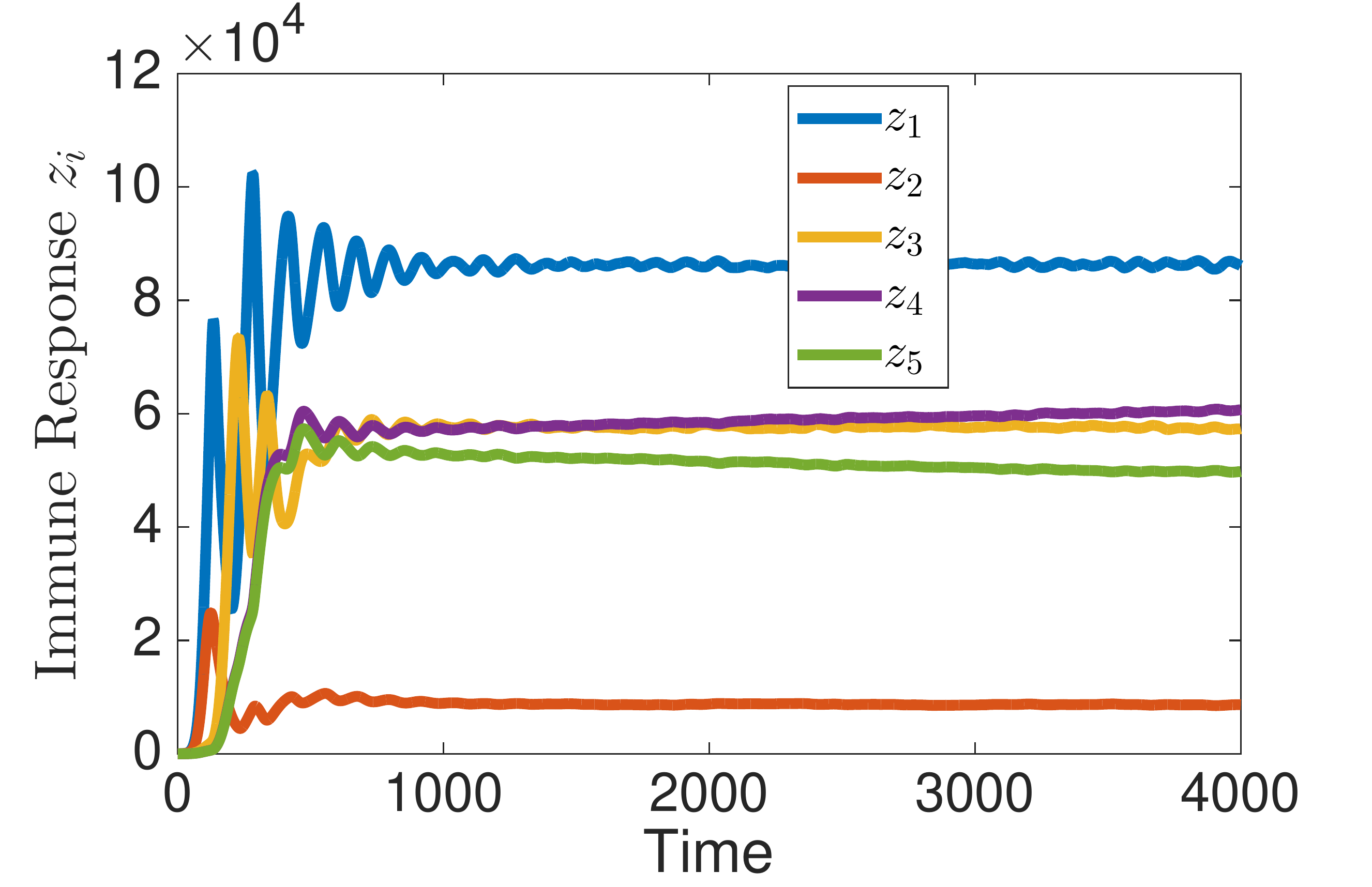}}
   \subfigure[][]{ \label{5d} \includegraphics[width=.5\textwidth,height=.25\textwidth]{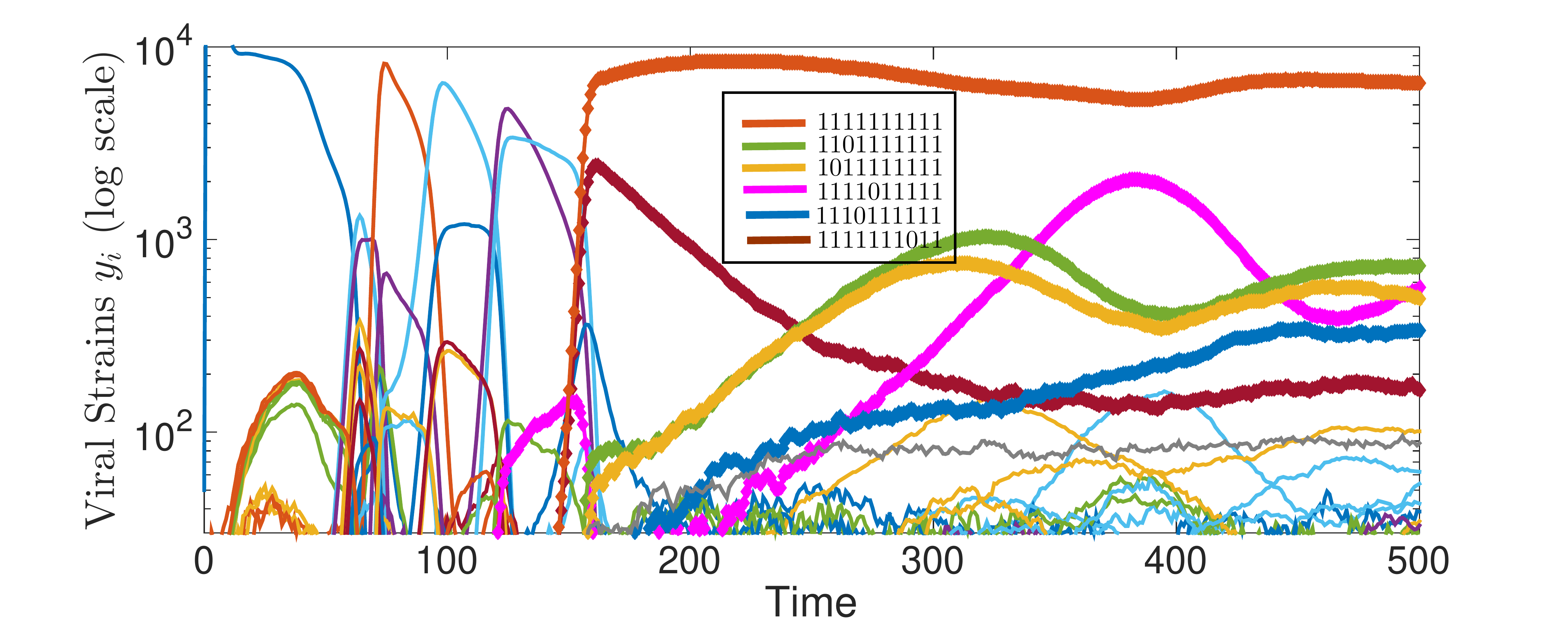}} \\
 \subfigure[][]{\label{5e} \includegraphics[width=.5\textwidth,height=.25\textwidth]{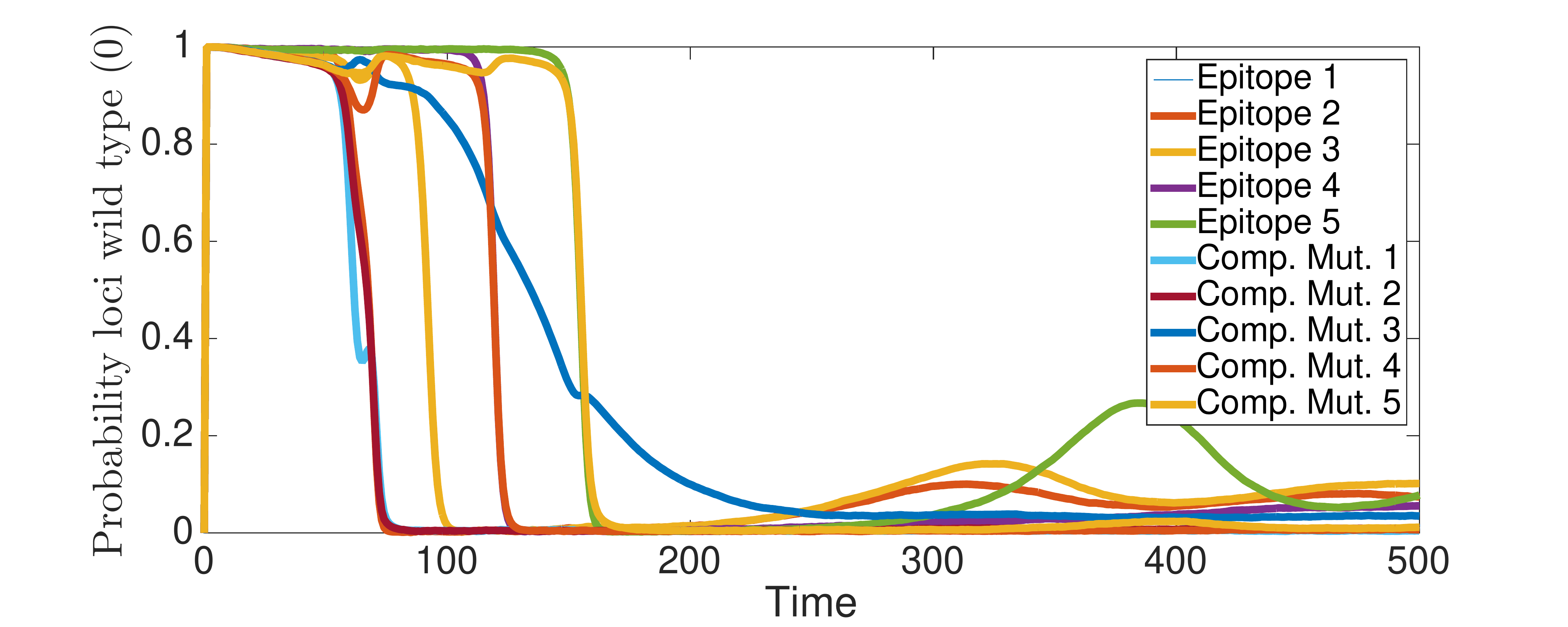}}
\subfigure[][]{ \label{5f}\includegraphics[width=.5\textwidth,height=.25\textwidth]{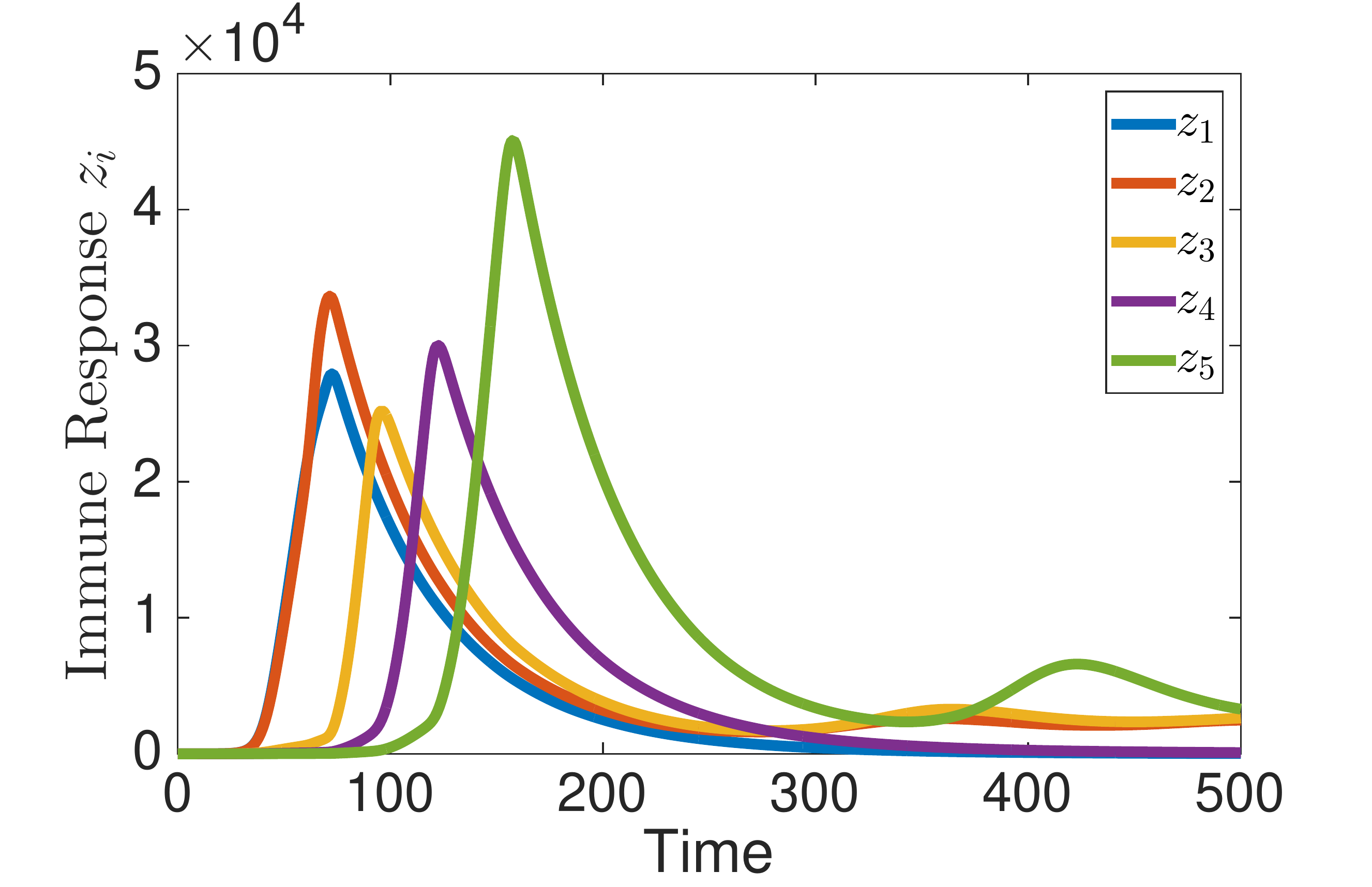}}

 \caption{ \emph{Fixation of resistant alleles and nestedness increases when incorporating compensatory mutations. }  (a,b,c) Simulations of model \eqref{odeS} with random viral mutations for $n=5$ epitopes under Gaussian distributed pairwise epitope interaction fitness cost landscape show that (a) viral strains $y_i(t)$, (b) allele frequency at each epitope and (c) immune responses $z_j(t)$ converge to steady state with large prevalence of wild-type (0) allele in viral population at each epitope.  (d,e,f) Adding complementary loci for each epitope to model which can compensate for $95\%$ of fitness cost of resistance mutations.  The compensatory mutations drive (d) viral strains $y_i(t)$ to rapidly converge to ``nearly nested'' structure as (e) sequential epitope and corresponding compensatory mutations sequentially become ascendant in population, and (f) immune responses $z_j(t)$ are escaped in order of immunodominance hierarchy. }
 \label{fig5}
 \end{figure} 
 
First for the stochastic extension of \eqref{odeS}, consider $n=3$ epitopes, which for simplicity is much less than an actual HIV genome and taken to be a representative cluster or sample of loci.  We utilize variables and parameters from the unscaled version of \eqref{odeS}, system (1) in \cite{browne2018dynamics}
$X=\frac{b}{c}x,  Y_i=\frac{b}{\delta}y_i, Z_j=\frac{\rho_j}{\mathcal I_j} z_j, \rho_j=\frac{\mu_j}{bq_j}, \sigma_j=\frac{\mu_j}{c}$ in order to represent concentrations ($ml^{-1}$) of target cells, virus and immune response, along with immune decay and scaling factor.  Let $b=5\times 10^3  \ (ml\cdot d)^{-1}$, $c=0.01\ d^{-1}$, $\delta=0.5\ d^{-1}$, $\mu_j=0.01\ d^{-1}$, $q_j=1.5$.  Furthermore, for the immunodominance hierarchy, we consider $\mathcal I_j$ uniform random variable in the range $[3.75,   7.875]$.
    First, assume that the viral fitnesses are calculated as $\mathcal R_i=\left[\sum_{i_{j}=1}(1-\kappa_j)+\sum_{i_{j}=1,i_{k}=1}B_{jk}\right]\mathcal R_0$, where additive fitness costs $\kappa$ were uniformly distributed in the range $[0,   0.5]$. and pairwise interaction $B_{jk}$ is uniformly distributed (random positive epistasis) in the range $[0,   1]$..  Then, all pairwise interactions, $B_{jk}$, are positive, along with the invasion circuits which we index $i=1,\dots,4$ in ascending order with respect to the invading binary sequence conversion to decimal representation.  The system is expected to converge to the nested network by Proposition \ref{propPair}, with asymptotic stability of equilibrium $\widetilde{\mathcal E}_3$, persistence of nested strains $y_0,\dots,y_3$ and extinction of remaining viral strains $y_4,\dots,y_7$ subject to small perturbations caused by random mutations, as displayed in Figure \ref{4a}.  Next, we increase  the reproduction number of $y_6 \ (101)$, so that the corresponding invasion circuit linear form $\mathcal A_3$ switches from positive to negative.  From our feasible bifurcations based on the circuit coefficients, we predict that $(101)$ can replace $(100)$ or $(111)$.  Observe in Fig. \ref{4b}, that equilibrium $\widetilde{\mathcal E}_3$ is altered by $(101)$ invading $(111)$, although the mutations allow to $(010)$ to be only at slightly lower levels than $(100)$ in the new strain hierarchy.  

When epistatic interactions become negative by subtracting the pairwise matrix terms, $B_{jk}$, from additive fitnesses, we project a non-nested pattern according to Proposition \ref{propPair}.  Indeed, in Fig. \ref{4c}, simulations converge to the $\leq 1$ mutation network, and hence the antagonism of negative interactions between epitopes thwarts the escape of virus at multiple epitopes.  Finally, we consider Gaussian distributed pairwise interactions, where $B_{jk}$ are random normal variables with mean zero (random signs) and variance of $0.1$ affecting magnitude of epistasis. Observe that the system may (Fig. \ref{4d}) or may not converge (Fig. \ref{4e})  to the nested network depending on the sign of the invasion circuits determining the overall epistasis encoded in the nested pathway.  Furthermore, in the latter case, simulations converge to an equilibrium structure that is not ``close'' to being nested, one-to-one, or $\leq 1$ network, indicating the presence of additional stable equilibrium structures and corresponding circuits not analyzed in this study for the $n=3$ epitope setting.  We also consider $n=5$ under epitopes with the same fitness landscape structure, although a variance of $0.05$ in the normally distributed pairwise epistasis is set to counteract accumulated fitness cost from strains with more mutated epitopes.  Simulations displayed for this case show that numerical solutions of the (deterministic) model \eqref{odeS} (Fig. \ref{4f} are consistent with the stochastic extension (Fig. \ref{4g} and Fig. \ref{4h}), supporting our argument that theoretical results in the differential equations carry over to the eco-evolutionary dynamics with random mutation.  Here, the fitness costs and non-positive epistasis circuits (with respect to nested network) prevent the dominance of strains with several mutations, and lead to the extinction of the weakest immune response $z_5$, along with persistence of only 4 strains, despite the 5 epitopes.  

  \FloatBarrier
In Fig. \ref{fig5}, we simulate eco-evolutionary dynamics again for 5 epitopes under Gaussian distributed pairwise interactions, where $B_{jk}$ are zero-mean normal random variables with variance of $0.05$, and all other parameter assumptions remaining the same.  The balance between immune response pressure selecting for resistance and the fitness costs occurring with each epitope mutation results in the virus mutant strains evolving to escape some immune responses, but the ancestral strains, including wild-type $y_0$ can still persist (Fig. \ref{5a}).  In addition, ``backward'' mutations allow mutated epitopes to revert back to wild-type ($0$) in a large proportion of viral population (Fig. \ref{5b}), even after invasion by mutant allele ($1$), as the sign of the invading circuit linear form and rise of more immune response populations (Fig. \ref{5c}) determine strain additions or replacements which result in the persistent strain structure of the equilibrium.  In HIV infection, resistance mutations often to become more dominant in viral population with several escapes persisting in the population without reversion because of compensatory mutations in linked loci which allow the virus to regain most of the fitness cost associated with an epitope mutation \cite{Althaus}.  We simulate compensatory mutations by adding a complementary loci for each epitope $j=1,\dots,5$, which is either neutral ($0$), not impacting fitness or if mutated ($1$) can result in the virus restoring $95\%$ of its original fitness value if the strain has mutated epitope $j$ from wild-type ($0$) to resistant ($1$).  Indeed, consider loci $5+j$, $j=1,\dots,5$, and viral sequence $\bm i'$ with $\bm i_{5+j}=0$ which has undergone mutation and fitness cost in epitope $j$ from neighboring strain $\bm i$ ($\bm i_j=0 \rightarrow \bm i'_j=1 \Rightarrow \mathcal R_{\bm i}>\mathcal R_{\bm i'}$).  Then assuming all other epitopes remain fixed, we suppose the strain $\bm i''$ gaining compensatory mutation has the following update in fitness: $\bm i'_{5+j}=0\rightarrow \bm i''_{5+j}=1 \Rightarrow \mathcal R_{\bm i''}=\mathcal R_{\bm i'}+.95\left(1-\frac{\mathcal R_{\bm i'}}{\mathcal R_{\bm i}}\right)$.  In contrast to the case of reproduction numbers solely dependent on epitope sequence, the addition of these complementary loci allows for sequential epitope escapes with concurrent compensatory mutations dominating the viral population (Fig. \ref{5d} and \ref{5e}) and suppressing the immune response (Fig. \ref{5f}).

\begin{figure}[t!]
\subfigure[][]{\label{fig6a}\includegraphics[width=7.5cm,height=4cm]{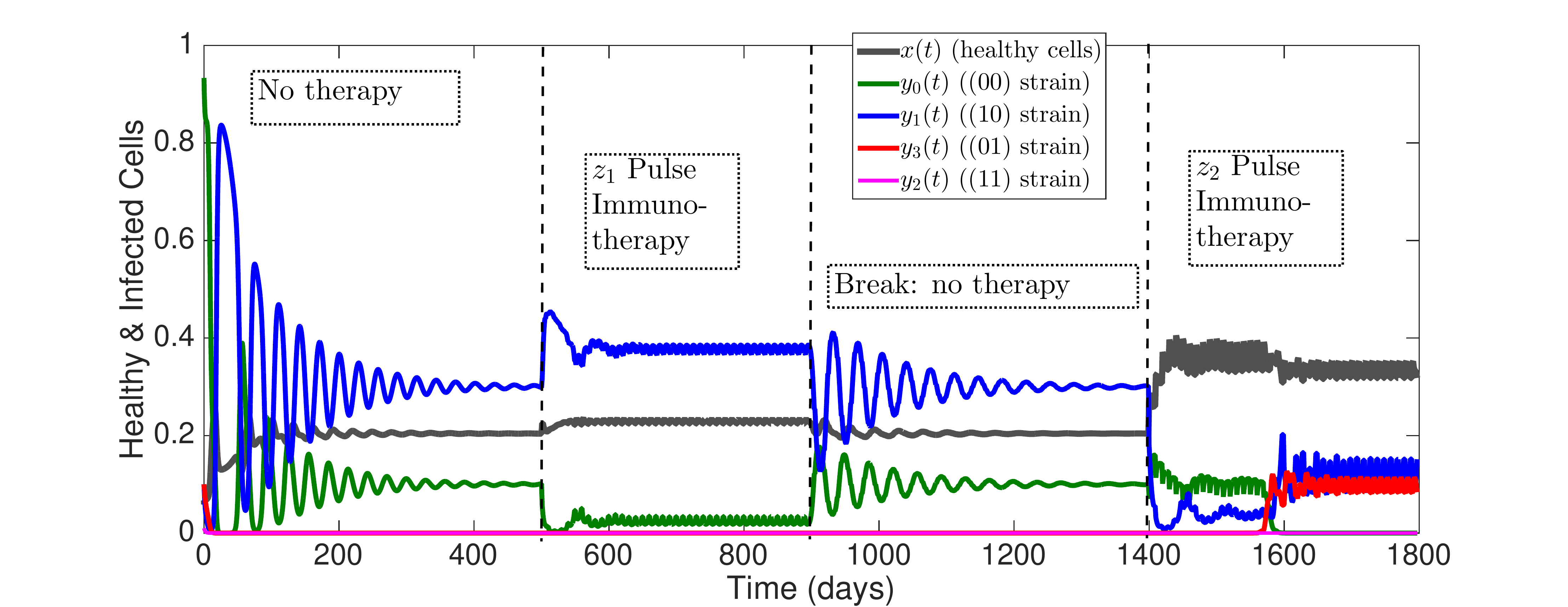}}
\subfigure[][]{\label{fig6b}\includegraphics[width=7.5cm,height=4cm]{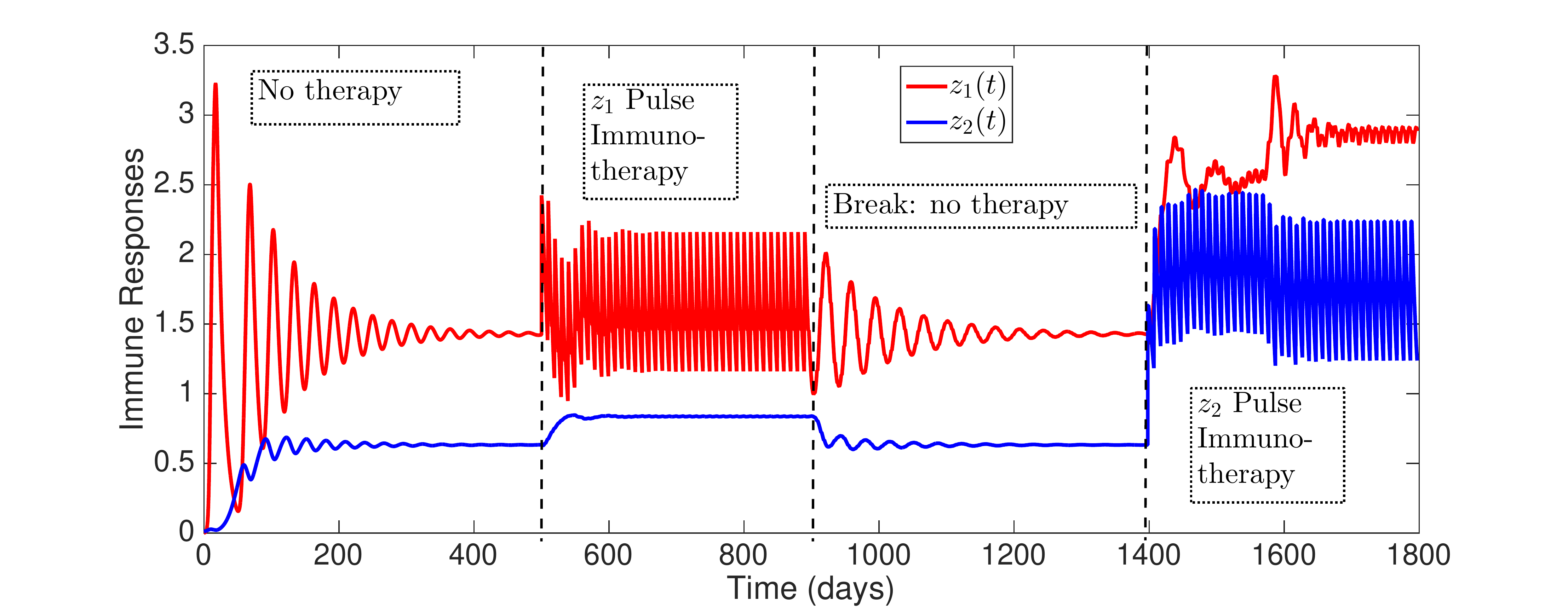}}  
\caption{\emph{Simulating pulse immunotherapies in two-epitope model shows priming subdominant response $z_2$ is more effective than therapy with dominant response $z_1$.} Viral strain $y_i(t)$ and healthy cell $x(t)$ (a), along with immune response $z_j(t)$ (b), trajectories in model \eqref{odeS} under no treatment initially, then periodic $z_1$ immune infusions, followed by treatment interruption, and finally periodic $z_2$ immune infusions.  Even though subdominant resistant strain $y_3 \ (01)$ has higher reproduction number, the $z_2$ therapy restores larger healthy cell count as system birfurcates from nested (strains $\left\{00,10\right\}$) to $\leq 1$ mutation (strains $\left\{00,10,01\right\}$) persistent network.}
  \label{fig6}
  \end{figure}

Finally, we numerically illustrate implications of our results for designing potential immunotherapy strategies against an immune escaping virus such as HIV.  We consider the deterministic ODE \eqref{odeS} with $n=2$ epitopes (diagram shown in Fig. \ref{3a}), and add periodic infusions of the immune response populations, $z_1(t)$ and $z_2(t)$.  In particular, we incorporate periodic infusion times, $t_k=t^(j)+k\tau$, $k=1,\dots,N$, of the immune population $z_j$ by applying an impulsive increase of $D$ units to the model, i.e. Dirac delta distributions ($D\delta(t-t_k)$) are added to the $\dot z_j$ component in \eqref{odeS}, and numerically solve in the cases of no treatment and distinct immunotherapies (see Fig. \ref{fig6}).  The viral fitness parameters utilized are $\mathcal R_0=15$, $\mathcal R_1=8$, $\mathcal R_2=3$, $\mathcal R_3=11.5$, $\mathcal I_1=10$, and $\mathcal I_2=2.5$ so that without therapy the system converges to nested equilibrium  $\bar{\mathcal E}_2$ with $y_0,y_1,z_1,z_2$ persisting.  Upon convergence to this rest point after perturbing the wild-type (immune-free) virus equilibrium $\widetilde{\mathcal E}_0$ by introducing mutant strains and immune responses, at $t^(1)=500 \ days$ we begin to pulse the dominant immune response $z_1$ by adding $D=1$ units of cells every $T=10 \ days$ (Fig. \ref{fig6a} and \ref{fig6b}).  The persistent variants remain in the same nested structure and the system settles into a periodically forced solution with an increase in the ``$z_1$-resistant'' viral strain ($y_1$ or $10$) prevalence, decrease in $y_0$, and modest $12.4\%$ jump in healthy cell count.  After removing the $z_1$-therapy and solutions returning to original state $\widetilde{\mathcal E}_0$, at $t^(2)=1400 \ days$ we test the periodic $z_2$-therapy with the same impulse magnitude of $D=1$ and frequency $T=10 \ days$.  Contrary to the first therapy, the periodic infusion of $z_2$ immune cells causes a bifurcation from the nested to the $\leq 1$-mutation network with addition of the subdominant $z_2$-resistant strain $y_3 \ (01)$ into the viral quasi-species.  Furthermore, both $z_1$ and $z_2$ populations are enhanced and the healthy cells increase by around $67\%$.  In each case, the stability condition given by inequality \eqref{inequ} is altered, so that even though the viral fitnesses $\mathcal R_i$ are constant, the pulsed $z_j$ levels can be thought to induce effective reproduction numbers which may change the sign of epistasis in the circuits \eqref{nestcirc1} (or \eqref{nestcirc2}) corresponding to the nested equilibria.  Here, the strategy of priming the subdominant response $z_2$ tilts this effective fitness landscape toward negative epistasis, convergence to $\leq 1$ mutation network, and, although invasion by the higher mutant fitness strain $y_3$ occurs, an improved outcome for host is obtained.

\FloatBarrier
\section{Discussion}

In this paper, we rigorously connect population dynamics thresholds with concepts from evolutionary genetics, which allows us to characterize distinct regimes of multi-strain persistence, stability, and resistance pathways in a virus-immune ecosystem in a biologically meaningful manner.  The complexity of the viral (binary sequence) genetic structure, along with dynamic virus fitness landscape and immune response populations, lead to a multitude of equilibria and general stability conditions which challenge interpretation, classification or simplification in terms of fundamental parameters such as reproduction number.  By finding equivalent sharp thresholds based on an appropriate definition of epistasis in the fitness landscape governing persistent equilibrium network structures, we are able to gain insight on eco-evolutionary dynamics.  In particular, the prediction of the virus escape pathway against immune attack on multiple epitopes is determined by epistasis in the ``invasion circuits'' controlling the bifurcations in our dynamical system.  

Our theoretical results lend support to \emph{circuits}, the minimal additive combinations of binary sequences \cite{beerenwinkel2007epistasis}, as the fundamental measure of epistasis in a fitness landscape.  Other ways to quantify epistasis may be simpler or offer other advantages, but circuits underly fitness landscape shape, and here we show that they also dictate prey-predator dynamics on top of building the phenotypic/genetic structure of the prey (virus) population.   This connection between population dynamics and genetics naturally comes from applying linear algebra to formulate the invasion rates of missing virus strains at an equilibrium as minimal combinations of virus reproduction numbers.  Moreover, the invasion circuit and corresponding linear form encode the resident strains which can be replaced by a mutant strain, and together with their equilibrium strain densities, determine the bifurcations resulting in new feasible steady states.  

The persistent network structures of virus and immune response populations analyzed in this work represent distinct patterns formed by the forces of viral resistance and fitness costs, and immunodominance.  The nested network equilibria admits a diverse ecosystem with generalist to specialist ordering in prey-predator interactions, as opposed to the modularity of the one-to-one (strain-specific) and $\leq 1$ mutation network.  In terms of viral escape from the immune response, the nested pathway offers the most efficient evolution as mutant strains sequentially gain resistance to immune populations strongest to weakest.  That the special case of positive (or synergistic) pairwise interactions between epitopes presents a nested trajectory (Proposition \ref{propPair}) highlights how convergence to this network coincides with the classical definition of positive epistasis favoring double mutants. While this proposition may be expected, both the dominant epitope escape being favored even when exacting a larger fitness cost than other epitopes and the viral (prey) fitness epistasis determining fate of the virus-immune (prey-predator) ecosystem, are less intuitive features of the result, along with our more general Theorem \ref{newthm1} on nested network equilibrium stability.  In contrast, the one-to-one and $\leq 1$ mutation network are instances of resulting dynamics for negative (antagonistic) pairwise interactions, and particularly the $\leq 1$ mutation structure is ideal from the host perspective of containing multi-epitope resistance.

 Numerical simulations of the ordinary differential equation \eqref{odeS}, along with an extended stochastic version including random mutation, demonstrate how eco-evolutionary trajectories are determined by epistasis in the viral fitness landscape, as predicted by our analytical results.  Indeed efficient viral escape in a nested fashion occurs when our necessary and sufficient conditions regarding positive epistasis are satisfied, and becomes more complex as negative epistatic interactions allow different combinations of resistance mutations to persist in the virus population. Under random epistatic pairwise interactions, any number of equilibria structures can be realized which may hinder multi-epitope resistance, but compensatory mutations may allow for sequential viral escape of immune responses, as shown in Fig. \ref{fig5}.     Furthermore, our model and results may inform upon immunotherapy for HIV.  In most clinical trials of therapeutic vaccines, potentially favorable T cell responses were of limited success due to viral escape from epitopes used in vaccine \cite{pantaleo2013vaccine}, but one possible strategy is to immunize with a set of the most conserved (associated with high fitness cost of resistance), subdominant epitopes \cite{mcmichael2006hiv,ahmed2019sub}.  Thus, it may be desirable to guide the virus-immune trajectory toward a non-nested network structure by priming subdominant immune responses.  Here, we illustrate that this strategy can work even when resistance to subdominant response comes with less fitness cost, as a bifurcation is induced to a state with viral mutant competition and optimal healthy cells compared to an immunotherapy with the dominant response (see Fig. \ref{fig6}).

Future work can build upon our results in several directions.  While the dynamics for $n=2$ epitopes is resolved for model \eqref{odeS}, the case $n\geq 3$ has not been completely classified, and our work shows that feasible stable equilibria may be discovered through analysis of relevant circuits, although even $n=3$ is challenging due to large number of strain combinations.   One way to explore how a particular ecosystem structure evolves is to follow the convergence of stepwise mutations and selection from wild-type strain in the hybrid stochastic/deterministic approach of polymorphic evolution sequences \cite{champagnat2011polymorphic}.  However,  simulations conducted (not shown here) revealed that the attracting (saturated) equilibrium was not obtained by a sequence of viral strain and immune response invasions starting from initial infection by the $\mathbf 0$ strain, thus multi-loci mutations and invasions are necessary, perhaps in the spirit of the ``adaptive walks'' jumping between equilibria of Lotka-Volterra systems developed in \cite{kraut2019adaptive}.  This approach of obtaining Lotka-Volterra dynamics from limits of stochastic models relies on strong conditions guaranteeing global stability for the ODE, and so it is an open problem for our system.   Finally, by incorporating data on the vial fitness landscape at multiple epitopes in the face of epistatic interactions and concurrent immune response attack, model parameterization with calculation of ``invasion circuits'' may verify theoretical results, predict eco-evolutionary trajectory, and inform upon potential immunotherapies.   

\section*{Acknowledgement}
CJB and FY acknowledge support by a U.S. National Science Foundation grant (DMS-1815095).  We also thank Hal Smith for insightful discussions.

\section*{Appendix}

\subsection*{Proofs of Theorems}
\begin{proof}[Proof of Proposition \ref{prop1new}]
  For any $\mathcal C$ and equilibrium of \eqref{odeS}, $\mathcal E*=(x^*,\mathbf{y}^*,\mathbf{z}^*)$, we find that
\begin{align*}
0&=\sum_{\mathbf k\in\mathcal C} a_{\mathbf k}  \frac{\dot y_{\mathbf k}}{\gamma_{\mathbf k} y_{\mathbf k}^*} \\ 
&=x^*\sum_{\mathbf k\in\mathcal C} a_{\mathbf k} \mathcal R_{\mathbf k} -\sum_{\mathbf k\in\mathcal C} a_{\mathbf k} - \sum_{\mathbf k\in\mathcal C} a_{\mathbf k} \sum_{j=1}^n (1-\mathbf k_j) z_{j}^*  \\
&=x^*\sum_{\mathbf k\in\mathcal C} a_{\mathbf k} \mathcal R_{\mathbf k} -\sum_{\mathbf k\in\mathcal C} a_{\mathbf k}- \sum_{\mathbf k\in\mathcal C} a_{\mathbf k} \left(\mathbf 1- \mathbf k\right) \cdot \mathbf{z}^*  \\
&=x^*\sum_{\mathbf k\in\mathcal C} a_{\mathbf k} \mathcal R_{\mathbf k} -\left(1+\mathbf 1 \cdot \mathbf{z}^*  \right) \sum_{\mathbf k\in\mathcal C} a_{\mathbf k} + \left( \sum_{\mathbf k\in\mathcal C} a_{\mathbf{k}} \mathbf{k} \right) \cdot \mathbf{z}^*  \\
&=x^*\sum_{\mathbf k\in\mathcal C} a_{\mathbf k} \mathcal R_{\mathbf k},
\end{align*}
because $\sum_{\mathbf k\in\mathcal C} a_{\mathbf k} \mathbf k=0$ and $\sum_{\mathbf k\in\mathcal C} a_{\mathbf k} =0$.  This proves the first statement.  The next statement follows from Proposition \ref{prop33} upon assuming $\sum_{\mathbf k\in\mathcal C} a_{\mathbf k} R_{\mathbf k}=0$.  Indeed, uniqueness of equilibrium in a certain positivity class is equivalent to ${\rm Ker}(A')^T\cap \vec{\mathcal R}'^{\perp}=\left\{\mathbf 0\right\}$, which is equivalent to the condition that the augmented matrix $C$ consisting of
adding the final row $\vec{\mathcal R}'^T$ to $(A')^T$ has trivial kernel \cite{browne2018dynamics}.  Here $A'$ is the $m'\times n'$ interaction matrix consisting of the $m'$ strains comprising the circuit and $n'$ (positive component) immune responses.  Consider the vector $\mathbf a$ consisting of the circuit weights.  Then from the previous points, we find that $C\mathbf a=\mathbf 0$.  Thus there cannot be a unique equilibrium and if there exists an equilibrium with $y^*_{\mathbf k}>0$ for all $\mathbf k\in\mathcal C$, then there are infinitely many such equilibria, with virus component vector denoted $\bar{\mathbf y}$.  Observe that since $\bar{\mathbf y}-\mathbf y^*\in {\rm Ker}(A')^T\cap \vec{\mathcal R}'^{\perp}$, then $\bar{\mathbf y}= \mathbf y^* + \alpha \mathbf a$ for $\alpha\in\mathbb R$.  
\end{proof}

\begin{proof}[Proof I of Theorem \ref{newthm1}]
If $\mathcal R_0> \mathcal Q_0:=1$, let $k$ be the largest integer in $[0,n-1]$ such that $\mathcal R_k > \mathcal Q_{k+1}$.  Without loss of generality, let $k=n-1$.
Consider a given missing viral strain $y_{i}$ ($i\in [n+1,2^n-1]$) with sequence $\mathbf i$.  Define the following linear form based on it's invasion rate:
\begin{align}
\frac{\dot y_i}{\gamma_i y_i} = -\frac{\mathcal A_{i}}{C_n}, \ \ \text{where} \ \ -\mathcal A_{i}:= \mathcal R_{i}-\mathcal R_{n} -\sum_{j=1}^n (1-i_j) \left( \mathcal R_{j-1}-  \mathcal R_{j} \right), \notag
\end{align}
and $C_n=\mathcal R_n$ when $\mathcal R_n>\mathcal Q_n$ and $C_n=\mathcal Q_n$  $\mathcal R_n\leq \mathcal Q_n$.  The telescoping sum above is determined by the following sequence: $\left(a_{j}\right), \ j=0,1,\dots,n$, where $a_{0}=1-i_{1}$, $a_{j}=i_{j}-i_{j+1}$ for $j=2,\dots,n-1$, $a_{n}=i_{n}$. In this way, $-\mathcal A_{i}:= \mathcal R_{i}-\sum_{j=0}^{n}\ a_j \mathcal R_j$.  In order to prove that this is a vanishing linear form of a circuit, we show that it is the linear form of a minimally linearly dependent collection of extended binary sequences.  Denote the binary sequences of nested network as $\mathbf k_0,\dots,\mathbf k_n$ corresponding to ordered strains $y_0,\dots, y_n$.  Let $\mathcal N\subset \left\{0,1\right\}^{n+1}$ denote the subset of nested extended binary sequences, where $\hat{\mathbf i}=\mathbf i1 \in  \left\{0,1\right\}^{n+1}\setminus \mathcal N$ and $\hat{\mathbf k}=\mathbf k1\in \mathcal N$ represent binary sequences extended by digit 1.  Notice that $\mathcal N$ forms a basis of $\mathbb R^{n+1}$ (since the $n+1\times n+1$ matrix $\left( \mathbf k_n,\mathbf k_{n-1},\dots,\mathbf k_0 \right)$ has a triangular row reduced eschelon form with values $\pm1$ on diagonal).  Thus for $\hat{\mathbf i} \in  \left\{0,1\right\}^{n+1}\setminus \mathcal N$, there is a unique set of coefficients $a_{j}$, $j=0,1,2,\dots,n$, yielding $\hat{\mathbf i}$ as a linear combination of the nested network vectors:
$$ \hat{\mathbf i}=a_{0}\hat{\mathbf k}_{0}+a_{1}\hat{\mathbf k}_{1}+\dots+a_{n}\hat{\mathbf k}_{n} .$$ 
The above linear system resolves as follows:
\begin{align*}
a_{1}+a_{2}+\dots+a_{n}&=i_{1}\\
a_{2}+a_{3}+\dots+a_{n}&=i_{2}\\
& \vdots\\
a_{k}+a_{k+1}+\dots+a_{n}&=i_{k}\\
& \vdots\\
a_{n}&=i_{n}\\
a_{0}+a_{1}+\dots+a_{n}&=1
\end{align*}

which leads us to the set of coefficients $a_{k}$ where $k=0,1,\dots,n$ defined by the following:

\[a_{0}=1-i_{1}\]
\[a_{k}=i_{k}-i_{k+1} \quad  \text{for} \quad k=1,\dots,n-1\] 
\[a_{n}=i_{n}\]
Therefore the set $\hat{\mathbf i}\cup \left\{\hat{\mathbf k}\right\}_{\hat{\mathbf k}\in\mathcal N}$ is linearly dependent.  Let $\alpha_i$ be the nonzero terms in sequence $(a_j)$, i.e. $\alpha_i:=\left\{ j\in[0,n]: a_j\neq 0\right\}$, where $a_j=\pm 1$ for $a_j\in\alpha_i$.  Since $(a_j)$ is unique linear combination with respect to basis $\mathcal N$, the set $\hat{\mathbf i}\cup \left\{\hat{\mathbf k}\right\}_{\hat{\mathbf k}\in\alpha_i}$ is a minimal linearly dependent set.  Thus we obtain the following circuit and corresponding vanishing linear form:
\begin{align*}
\mathcal C_i&=y_i \cup \left\{ y_j \right\}_{j\in\alpha_i}, \qquad \mathcal A_i= -\mathcal R_i -\sum_{j\in\alpha_i}a_j \mathcal R_{j}. 
\end{align*}
\end{proof}

\begin{proof}[Proof II of Theorem \ref{newthm1}]
 If $\mathcal R_0> \mathcal Q_0:=1$, let $k$ be the largest integer in $[0,n-1]$ such that $\mathcal R_k > \mathcal Q_{k+1}$.  WLOG let $k=n-1$.
Consider a given missing viral strain $y_{i}$ ($i\in [n+1,2^n-1]$) with sequence $\mathbf i$.  Define the following linear form based on it's invasion rate:
\begin{align}
\frac{\dot y_i}{\gamma_i y_i} = -\frac{\mathcal A_{i}}{C_n}, \ \ \text{where} \ \ -\mathcal A_{i}:= \mathcal R_{i}+\mathcal R_{n} +\sum_{j=1}^n (1-i_j) \left( \mathcal R_{j-1}-  \mathcal R_{j} \right), \notag
\end{align}
and $C_n=\mathcal R_n$ when $\mathcal R_n>\mathcal Q_n$ and $C_n=\mathcal Q_n$  $\mathcal R_n\leq \mathcal Q_n$.
We claim that $\mathcal A_{i}=0$ in additive case, and furthermore $\mathcal A_{i}\neq 0$ if any (non-zero) viral fitness is removed from $\mathcal A_{i}$ in the resulting sum.  In other words we claim that $\mathcal A_{i}$ defines a 
 circuit $\mathcal C$ containing strain $i$ and other strains on nested network.  To test additivity, it suffices to consider the linear form on the binary sequences:
 \begin{align}
-f_{i}&:= \mathbf i-1^n - \sum_{j=1}^n (1-i_j) \left( 1^{j-1}0^{n-j+1}-  1^{j}0^{n-j} \right) \notag
\end{align}
Since $\mathbf i$ is not in nested network ($i\in [n+1,2^n-1]$), there exists $p\in[1,n-1]$ such that $i_p=0,i_{p+1}=1$.  In other words, there exists a $01$ string  in the binary sequence $\mathbf i$.  We prove that $\mathcal A_i$ defines a circuit by induction on the number of $01$ strings, $s$.  First suppose that $s=1$.  Let $0\leq m_1 <p$ be maximal such that $i_{p}=1$ and $p+1\leq m_2\leq n$ be maximal such that $i_{m_2}=0$.    With these conditions, $\mathbf i=1^{m_1}0^{p-m_1}1^{m_2-p}0^{n-m_2}$.  Then
 \begin{align}
-f_{i}&:= \mathbf i- 1^n - \sum_{j=1}^n (1-i_j) \left( 1^{j-1}0^{n-j+1}-  1^{j}0^{n-j} \right) \notag \\
&=\mathbf i-1^n - \left(1^{m_1}0^{n-m_1}-1^{p}0^{n-p}\right) - \left(1^{m_2}0^{n-m_2}-1^n \right)\notag \\
\Rightarrow & f_i=  \mathbf i - 1^{m_1}0^{n-m_1}+1^{p}0^{n-p} - 1^{m_2}0^{n-m_2} \notag  \\
&=  0^p1^{m-p-1}0^{n-m_2}-0^p1^{m-p-1}0^{n-m_2} \notag  \\
&=0 .
\end{align}
Furthermore $f_i=  \mathbf i - 1^{m_1}0^{n-m_1}+1^{p}0^{n-p} - 1^{m_2}0^{n-m_2}$ contains the viral sequences corresponding the non-zero fitness quantities in $\mathcal A_i$.  Thus $\mathcal A_i$ defines a circuit since the minimal circuit size is $4$.   Now for the induction step, consider $s>1$.  Assume that $\mathcal A_{\ell}$ defines a circuit for any sequence $\mathbf{\ell}$ with $s-1$ or less $(01)$ strings, and suppose the sequence $\mathbf i$ has $s$ $(01)$ strings. Let $p_1<p_2<\dots<p_s$ be locations of the $01$ strings (with $i_{p_j}=0,i_{p_j+1}=1$).  Let $0\leq m_1 <p_1$ be maximal such that $i_{m_1}=1$ and $p_1+1\leq m_2\leq p_2$ be maximal such that $i_{m_2}=1$.    So
$\mathbf i=1^{m_1}0^{p_1-m_1}1^{m_2-p_1}i_{p_2}\dots i_n$.   Then 
 \begin{align}
f_{i}&:= \mathbf i- 1^n - \left(1^{m_1}0^{n-m_1}-1^{p_1}0^{n-p_1}\right) - \sum_{j=p_1+2}^n (1-i_j) \left( 1^{j-1}0^{n-j+1}-  1^{j}0^{n-j} \right) \notag \\
&=1^{m_1-1}0^{p_2-m_1}i_{p_2}\dots i_n-1^n - \sum_{j=p_1+2}^n (1-i_j) \left( 1^{j-1}0^{n-j+1}-  1^{j}0^{n-j} \right) \notag \\
&=\mathbf{\tilde i} - 1^n - \sum_{j=1}^n (1-\tilde i_j) \left( 1^{j-1}0^{n-j+1}-  1^{j}0^{n-j} \right) \notag \\
 & = f_{\tilde i},
\end{align}
 where $\mathbf{\tilde i}=1^{m_1}0^{p_2-m_1}1i_{p_2+2}\dots i_n$ has $s-1$ $(01)$ strings.  Thus by induction hypothesis, we obtain $f_i=f_{\tilde i}=0$.   Let $\mathcal C_i$ denote the collection of viral sequences corresponding the non-zero fitness quantities in $\mathcal A_i$.   Notice that it is not hard to ascertain from the above calculations that 
\begin{align*}
\mathcal C_i&=\mathbf i \cup \left\{1^{m_j}0^{n-m_j},1^{p_j}0^{n-p_j}\right\}_{j=1}^s \cup 1^{m_{s+1}}0^{n-m_{s+1}}, \\
\mathcal A_i&=\mathcal R_i -\sum_{j=1}^{s+1}\mathcal R_{m_j} +  \sum_{j=1}^{s}\mathcal R_{p_j}, \quad (0\leq m_1 < p_1 < m_2< \dots <p_s < m_{s+1} \leq n+1 ). 
\end{align*}
Consider an arbitrary proper subset $\mathcal B$ of  $\mathcal C_i$.  First, we claim that there can not be a circuit consisting solely of sequences in the nested network.  Suppose by way of contradiction that there exists a linear form with $g:=\sum_{j=1}^{n+1} b_j \mathbf j=0$.  Let $k=\max\left\{1\leq j\leq n+1 |  b_j\neq 0 \right\}$.  Then for the $k^{th}$ digit in the binary sequence of $g_N$, we find $(g)_k\neq 0$.  So there are no vanishing linear forms on the nested network.  Thus it suffices to consider the case where $\mathbf i\in \mathcal B$.  Motivated from calculations above, define
\begin{align*}
\mathbf{\tilde i}=\mathbf i + \sum_{\mathcal C_i\setminus \mathcal B} \left(-1^{m_j}0^{n-m_j}+1^{p_j}0^{n-p_j}\right),
\end{align*}
where $\mathbf{\tilde i}$ is not in nested network since $\mathcal B\neq \emptyset$.  Furthermore because $\mathcal C_i\setminus \mathcal B\neq \emptyset$, we obtain that $\mathbf{\tilde i}$ has less than $s$ (01) strings.  By induction hypothesis, $\mathcal A_{\tilde i}$ defines a circuit $\mathcal C_{\tilde i}$ for the sequence $\mathbf{\tilde i}$, where $\mathcal C_{\tilde i}=\left\{\mathbf{\tilde i}\right\} \cup \mathcal B\setminus\left\{\mathbf{i}\right\}$.  Denote the vanishing linear form as  $f_{\tilde i}=\sum_{\ell\in \mathcal C_{\tilde i}} a_{\ell} \mathbf{\ell}$.  Now for arbitrary coefficients $b_j$,
\begin{align*}
\sum_{j\in\mathcal B} b_j  \mathbf j&= b_i \mathbf i + \sum_{ \mathcal B\setminus\left\{\mathbf{i}\right\}} b_j  \mathbf j \\
&= b_i \left(\mathbf{\tilde i}-\sum_{\mathcal C_i\setminus \mathcal B} \left(-1^{m_j}0^{n-m_j}+1^{p_j}0^{n-p_j}\right)\right) + \sum_{ \mathcal B\setminus\left\{\mathbf{i}\right\}} b_j  \mathbf j \\
&= \sum_{ \mathcal B\setminus\left\{\mathbf{i}\right\}} (b_j-b_ia_j)  \mathbf j -b_i \sum_{j\in \mathcal C_{\tilde i}} a_{j} \mathbf j ,
\end{align*}
The above sum consists solely of sequences in the nested network and thus there are no vanishing linear forms.  This implies that the above sum is zero only if $b_i=0$, which further leads to conclusion that $b_j=0$ for $j\in\mathcal B\setminus\left\{\mathbf{i}\right\}$.  Thus the proper subset $\mathcal B$ can not be a circuit for any linear form. 

\end{proof}

\begin{proof}[Proof of Theorem \ref{ssF}]
By Proposition 6 in \cite{browne2018dynamics}, an equilibrium $\mathcal E^*$ with a strain-specific subgraph, i.e. $\Omega_y\subseteq [1,n+1]$, is stable if and only if one of the following holds:
\begin{itemize}
\item[i.]  $\mathcal R_{n+1}\leq \mathcal P_n$ and $\left( |\Lambda_{i}| - 1 \right) \mathcal P_n + \mathcal R_{i} \leq  \sum\limits_{j\in\Lambda_{i}} \mathcal R_j  \quad \forall i \in [n+2, 2^n]$, in which case $\Omega_y=\Omega_z= [1,n]$.
\item[ii.]  $\mathcal R_{n+1}> \mathcal P_n$ and $\left( |\Lambda_{i}| - 1 \right) \mathcal R_{n+1} + \mathcal R_{i} \leq  \sum\limits_{j\in\Lambda_{i}} \mathcal R_j  \quad \forall i \in [n+2, 2^n]$, in which case $\Omega_y=[1,n+1]$ and $\Omega_z= [1,n]$.
\end{itemize}
Fix an invading strain $y_i$, $i\in [n+2, 2^n]$, with binary sequence.  First note that the inequalities in cases (i) and (ii) can be re-written as $\frac{\mathcal A_{i}}{K_n}\geq 0$ where $\mathcal A_i= -\mathcal R_i -\left( |\Lambda_{i}| - 1 \right) \mathcal R_{n+1} + \sum_{j\in\Lambda_i} \mathcal R_{j}$, and $K_n=\mathcal P_n$ if $\mathcal R_{n+1}\leq \mathcal P_n$ and $K_n=\mathcal R_{n+1}$ if $\mathcal R_{n+1}> \mathcal P_n$.  To show that $\mathcal C_i=y_i \cup \left\{ y_j \right\}_{j\in\Lambda_i}$ is a circuit with linear form $\mathcal A_i$, we proceed with a similar approach to our first proof of Theorem \ref{newthm1}. Denote the binary sequences of one-to-one network as $\mathbf k_1,\dots,\mathbf k_{n+1}$ corresponding to ordered strains $y_1,\dots, y_{n+1}$. Let $\mathcal S\subset \left\{0,1\right\}^{n+1}$ denote the subset of strain-specific extended binary sequences, where $\hat{\mathbf i}=\mathbf i1 \in  \left\{0,1\right\}^{n+1}\setminus \mathcal S$ and $\hat{\mathbf k_j}=\mathbf k_j1\in \mathcal S$ represent binary sequences extended by digit 1.  Notice that $\mathcal S$ forms a basis of $\mathbb R^{n+1}$. Indeed, it is not hard to show the row reduced echelon form of $n+1 \times n+1$ matrix is triangular.   Thus for $\hat{\mathbf i} \in  \left\{0,1\right\}^{n+1}\setminus \mathcal N$, there is a unique set of coefficients $a_{j}$, $j=1,2,\dots,n+1$, yielding $\hat{\mathbf i}$ as a linear combination of the nested network vectors:
$$ \hat{\mathbf i}=a_{1}\hat{\mathbf k}_{1}+\dots+a_{n+1}\hat{\mathbf k}_{n+1} .$$ 
The above linear system resolves as follows:
\begin{align*}
a_{2}+\dots+a_{n+1}&=i_{1}\\
& \vdots\\
\sum_{j\neq k}a_{j} &=i_{k}\\
& \vdots \\
a_{1}+\dots+a_{n+1}&=1
\end{align*}

which leads us to the set of coefficients $a_{k}$ where $k=1,\dots,n+1$ defined by the following:
\[a_{k}=1-i_{k}\quad  \text{for} \quad k=1,\dots,n\] 
\[a_{n+1}=1-(n-\sum_{k=1}^n i_k)= -\left( |\Lambda_{i}| - 1 \right) \]
Thus, with analogous argument as before, we obtain the indicated circuit $\mathcal C_i$ and corresponding linear form $\mathcal A_i$.

\end{proof}

\begin{proof}[Proof of Proposition \ref{propPair}]
First assume that pairwise interaction matrix $B$ is positive and consider the stability of the nested equilibrium, $\widetilde{\mathcal E}_{n}$ (or $\overline{\mathcal E}_{n}$), as characterized by circuits in Corollary \ref{maincor}.
We proceed by induction on the number of $(01)$ strings denoted by $s$ for the invading strain.  Suppose $s=1$ and the invading strain is written as in prior proof as $\mathbf i=1^{m_1}0^{p-m_1}1^{m_2-p}0^{n-m_2}$ and the collection of strains in the circuit is given by $\mathcal C_i=\mathbf i \cup \left\{1^{m_1}0^{n-m_1},1^{p_1}0^{n-p_1}\right\} \cup 1^{m_{2}}0^{n-m_{2}}$.  Then since the additive elements will sum to zero in the linear form $\mathcal A_i$, the only remain terms come from pairwise interactions in $B$ and can be calculated as:
\begin{align*}
-\mathcal A_i&=\sum_{j=1}^{m_1}\sum_{k>j}^{m_1}B_{jk}+\sum_{j=1}^{m_1}\sum_{k=p_1+1}^{m_2}B_{jk}+\sum_{j=p_1+1}^{m_2}\sum_{k>j}^{m_2}B_{jk} \\
& \qquad -\sum_{j=1}^{m_1}\sum_{k>j}^{m_1}B_{jk}+\sum_{j=1}^{p_1}\sum_{k>j}^{p_1}B_{jk}-\sum_{j=1}^{m_2}\sum_{k>j}^{m_2}B_{jk} \\
&= \sum_{j=1}^{p_1}\sum_{k>j}^{p_1}B_{jk}- \sum_{j=1}^{m_1}\sum_{k>j}^{p_1}B_{jk}-\sum_{j=m_1+1}^{p_1}\sum_{k>j}^{m_2}B_{jk} \\
&= -\sum_{j=m_1+1}^{p_1}\sum_{k=p_1+1}^{m_2}B_{jk}
&<0
\end{align*}
Now for the induction step, suppose that $\mathbf i$ has $s$ (01) strings.  It is not hard to see that $\mathcal A_i=\mathcal A_{\tilde i}$, for invading strain $\mathbf{\tilde i}$, where $\mathbf{\tilde i}=1^{m_1}0^{p_2-m_1}1i_{p_2+2}\dots i_n$ has $s-1$ $(01)$ strings.  Thus by induction hypothesis $-\mathcal A_i<0$, or $\mathcal A_i>0$ giving positive epistasis and stability of nested network.  

Next suppose that matrix $B$ is negative and consider the stability of $\mathcal E^{\dagger}_n$ and $\mathcal E^{\ddagger}_{n+1}$ consisting of strains $y_1,\dots,y_n,y_{n+1}$ with binary sequences $\mathbf k_1,\dots,\mathbf k_{n+1}$, where $\Lambda_j=\left\{j\right\}$ for $j=1,\dots,n$ and $\Lambda_{n+1}=\emptyset$.  We inspect the invasion circuit of a strain with sequence $\mathbf i$ outside the one-to-one network.  Let $s$ be the number of $1s$ in sequence $\mathbf i$, located at loci $\ell_1,\dots,\ell_s$, where $0\leq s\leq n-2$. Again the additive terms in $\mathcal A_{\mathbf i}$ are zero and thus we have:
 \begin{align*}
\mathcal A_i&= -\mathcal R_i -\left( |\Lambda_{i}| - 1 \right) \mathcal R_{n+1} + \sum_{j\in\Lambda_i} \mathcal R_{j} \\
&=-\sum_{j=1}^{s}\sum_{k>\ell_j} B_{\ell_j k}-(n-1-s)\sum_{j=1}^{n}\sum_{k>j}B_{jk}+(n-s)\sum_{j=1}^{n}\sum_{k>j}B_{jk}\\ & \qquad -\sum_{m=1}^{n}\left[\sum_{k>m}B_{mk}+\sum_{k=1}^{m-1}B_{km}\right] + \sum_{j=1}^{s}\left[\sum_{k>\ell_j}B_{\ell_j k}+\sum_{k=1}^{\ell_j-1}B_{k \ell_j}\right] \\
&= -\sum_{j=1}^{n}\sum_{k=1}^{j-1}B_{kj} + \sum_{j=1}^{s}\sum_{k=1}^{\ell_j-1}B_{k \ell_j} \\
&>0 \quad \text{since} \quad B<0, \ s\leq n-2.
\end{align*}
Finally, for the $\leq 1$-mutation network, only the invading strain $\mathbf i$ will have $\geq 2$ mutations, so 
 \begin{align*}
\mathcal A_{\mathbf i}&= -\mathcal R_{\mathbf i} -\left( n-|\Lambda_{i}| - 1 \right) \mathcal R_{0} + \sum_{j\notin\Lambda_i} \mathcal R_{j} \\
&=-\sum_{j=1}^{n}\sum_{k>j} B_{j k} >0 \quad \text{since} \quad B<0.  
\end{align*}
\end{proof}

\begin{proof}[Proof of Proposition \ref{propMult}]
Let $0<f_j<1, j=1,\dots,n$ represent the multiplicative fitness costs for each epitope.  We prove by induction on the number of $(01)$ strings denoted by $s$ for the invading strain.  Suppose $s=1$ and the invading strain is written as in prior proof as $\mathbf i=1^{m_1}0^{p-m_1}1^{m_2-p}0^{n-m_2}$ and the collection of strains in the circuit is given by $\mathcal C_i=\mathbf i \cup \left\{1^{m_1}0^{n-m_1},1^{p_1}0^{n-p_1}\right\} \cup 1^{m_{2}}0^{n-m_{2}}$.  Then the linear form $\mathcal A_i$ can be calculated as:
\begin{align*}
-\mathcal A_i&=\mathcal R_0\left(f_1\cdots f_{m_1}f_{p_1+1}\cdots f_{m_2}-f_1\cdots f_{m_1}+f_1\cdots f_{p_1}-f_1\cdots f_{m_2}\right) \\
& =\mathcal R_0f_1\cdots f_{m_1}(1-f_{p_1+1}\cdots f_{m_2})(f_{m_1+1}\cdots f_{p_1}-1) \\
&<0
\end{align*}
Now for the induction step, suppose that $\mathbf i$ has $s$ (01) strings.  It is not hard to see that $\mathcal A_i=\mathcal A_{\tilde i}$, for invading strain $\mathbf{\tilde i}$, where $\mathbf{\tilde i}=1^{m_1}0^{p_2-m_1}1i_{p_2+2}\dots i_n$ has $s-1$ $(01)$ strings.  Thus by induction hypothesis $-\mathcal A_i<0$.  
\end{proof}

\begin{proof}[Proof of Proposition \ref{genprop}]
Let $\mathbf i \in \left\{0,1\right\}^n \setminus \mathcal S$ with integer coordinates $\left(a_{\mathbf k}\right)_{\mathbf k\in\mathcal S}$.  Clearly $\mathcal C \times \left\{1\right\}= \mathcal S\times \left\{1\right\} \cup \left\{\mathbf i 1\right\}$ is a linearly dependent set $\mathbb R^{n+1}$ with linear form on fitnesses given by $\mathcal A_{\mathbf i}=-\mathcal R_{\mathbf i}-\sum_{\mathbf k\in\mathcal S} a_{\mathbf k} \mathcal R_{\mathbf k}$.  Furthermore any proper subset is linearly independent since $\mathcal S\times \left\{1\right\}$ is a basis of $\mathbb R^{n+1}$.  Thus $\mathcal C$ is a circuit with linear form $\mathcal A_{\mathbf i}$.  By proof of Prop \ref{prop1new}, \begin{align*} \frac{\dot y_i}{\gamma_i y_i} &= \frac{\dot y_i}{\gamma_i y_i} + \sum_{\ell=1}^{n+1} a_{\ell} \frac{\dot y_{\ell}}{\gamma_{\ell} y_{\ell}} \\ & =  \sum_{\mathbf k\in\mathcal C} a_{\mathbf k}  \frac{\dot y_{\mathbf k}}{\gamma_{\mathbf k} y_{\mathbf k}}  \\ &=x^*\sum_{\mathbf k\in\mathcal C} a_{\mathbf k} \mathcal R_{\mathbf k} \\  &= x^*\left[ \mathcal R_i +  \sum_{\ell=1}^{n+1} a_{\ell}\mathcal R_{\ell}\right]  \\ 
&= -x^*\mathcal A_{\mathbf i} . \end{align*}
\end{proof}

\subsection*{The ``$\leq 1$-mutation'' network equilibria} \label{A4}
Consider the $\leq 1$ mutation network, $\tilde{\mathcal S}_1$, consisting of wild-type and 1-mutation viral strains $y_0,y_1,\dots,y_n$ where the sequence of $y_j$ is $\mathbf j=\left(\delta_{\ell j}\right)_{\ell=1}^{n}$ for $j=1,\dots,n$.  First it is simpler to look at the $n$ strain equilibrium $\mathcal E^{1*}$ containing positive components for $y_1^*,\dots,y_n^*$, where $y_0^*=0$, i.e. leaving out the wild-type strain.
By Proposition \ref{prop33} and \eqref{genEqcon}, such a positive equilibrium $\mathcal E^{1*}=(x^*,y^{1*},z^*)$ of system (\ref{odeS}) satisfies
\begin{align*}
 x^*&=\frac{1}{\sum_{j=1}^n\mathcal R_j - (n-1)\mathcal R_0}, \quad A y^{1*} = \vec s,  \quad Az^*=\vec{\mathcal R}^1x^*-\vec 1,  \quad
\text{where} \quad  A=\vec 1 \left(\vec 1\right)^T - I_n, \\ A^{-1}&=\frac{1}{n-1}\vec 1 \left(\vec 1\right)^T - I_n, \quad y^{1}=\left(y_1,y_2,\dots,y_n\right)^T, \quad \vec{\mathcal R}^1=\left(\mathcal R_1,\mathcal R_2,\dots,\mathcal R_n\right)^T
\end{align*}
with $I_n$ is the $n\times n$ identity matrix.  Here we find that:
\begin{align*}
y_i^{*} &= \frac{1}{n-1}\left(-(n-2)s_i+\sum_{j\neq i}s_j\right), \quad x^*=\frac{n-1}{n-1+ \sum_i \mathcal R_is_i}, \quad z^*_i=\frac{1}{n-1}(\mathcal R_ix^*-1)
\end{align*}
With the immunodominance hierarchy $s_i\leq s_{i+1}$, then $y_i^{*}>0$ if $s_1>\sum_{i>1}(s_n-s_i)$ and $z_i^*>0$ if $\mathcal R_i\left(n-1-\sum_i s_i\right)>n-1$.  If these conditions are satisfied, then the equilibrium $\mathcal E^{1*}$ is saturated in the subsystem restricted to $\mathcal S_1$.  In \cite{browne2018dynamics} we showed that in the larger network of viral strains, the equilibrium $\mathcal E^{1*}$ is always unstable in the case with equal reproduction numbers $\mathcal R_1=\mathcal R_2=\dots=\mathcal R_n$.  

 Now consider invasion by the wild strain $y_0$, which can result in an $n+1$ strain equilibrium $\widetilde{\mathcal E}^{1*}$ consisting of the viral strain network $\tilde{\mathcal S}_1$.  By Proposition \ref{prop33}, the positive components $x^*, \widetilde{y}^{1*},z^*$ of  $\widetilde{\mathcal E}^{1*}$ satisfies:
 \begin{align*}
x^*&=\vec 1^{\,T} C^{-1}_{(n+1)}, \quad \text{here} \quad C=\begin{pmatrix} A & \vec{\mathcal R} \\ \vec 1^{\,T} & \mathcal R_0 \end{pmatrix}, \\
\Rightarrow x^*&=\frac{1}{\sum_{j=1}^n \mathcal R_j - (n-1) \mathcal R_0}, \\
\widetilde{y}^{1*}&=\begin{pmatrix} \vec s^{\,T} & \frac{1}{x^*}-1 \end{pmatrix} C^{-1} \\
A z^*&= \vec{\mathcal R} x^* -1, \quad \sum_{i=1}^n z^*_i = \mathcal R_0 x^*-1 . 
\end{align*}
The above equations are difficult to analyze in general, but when $x^*>0, \widetilde{y}^{1*}> \mathbf 0, z^*> \mathbf 0$, the $n+1$ strain $\leq 1$ mutation equilibrium will be positive. Furthermore, if the linear forms of invasion circuits \eqref{om_circuit} are positive, then by Proposition \ref{omprop},  $\widetilde{\mathcal E}^{1*}$ will be stable.

\bibliography{hivctlNetwork_References.bib}
\bibliographystyle{spmpsci}
\end{document}